\documentclass{dmathesis}
\newcommand{\be}{\begin{equation}}
\newcommand{\ee}{\end{equation}}
\newcommand{\bea}{\begin{eqnarray}}
\newcommand{\eea}{\end{eqnarray}}
\newcommand{\bse}{\begin{subequations}}
\newcommand{\ese}{\end{subequations}}
\newcommand{\beqa}{\begin{eqnarray}}
\newcommand{\eeqa}{\end{eqnarray}}
\newcommand{\beqar}{\begin{eqnarray*}}
\newcommand{\eeqar}{\end{eqnarray*}}
\newcommand{\bi}{\begin{itemize}}
\newcommand{\ei}{\end{itemize}}
\newcommand{\bn}{\begin{enumerate}}
\newcommand{\en}{\end{enumerate}}
\newcommand{\fixme}[1]{\textbf{FIXME: }$\langle$\textit{#1}$\rangle$}

\newcommand{\ba}{\begin{array}}
\newcommand{\ea}{\end{array}}
\newcommand{\bc}{\begin{center}}
\newcommand{\ec}{\end{center}}
\newcommand{\nnr}{\nonumber \\}
\newcommand{\nn}{\nonumber}

\newcommand{\ads}{AdS$_{3}\;$}
\newcommand{\cf}{{\em cf.}\ }
\newcommand{\ie}{{\em i.e.}\ }
\newcommand{\eg}{{\em e.g.} }

\newcommand{\etal}{{\em et al.}\ }

\def\pd{\partial}
\def\sltr{SL$(2,\mathbb{R})$}

\newcommand{\p}{\partial}
\newcommand{\dq}{\dot{q}}

\newcommand{\mn}{{\mu\nu}}

\newcommand{\cd}{\nabla}
\newcommand{\de}{\delta}
\newcommand{\dd}{\boldsymbol{\mathrm{d}}}

\newcommand{\eps}{\epsilon}
\def\xp{x^+}
\def\xn{x^-}

\newcommand{\bomega}{{\boldsymbol \omega}}
\newcommand{\bTheta}{{\boldsymbol \Theta}}

\newcommand{\bL}{{\mathbf L}}
\newcommand{\bM}{{\boldsymbol M}}
\newcommand{\cL}{{\mathcal L}}

\newcommand{\cH}{{\mathcal H}}
\newcommand{\bG}{{\boldsymbol \Gamma}}

\newcommand{\bk}{{\boldsymbol k}}
\newcommand{\bE}{{\boldsymbol E}}

\newcommand{\bQ}{{\mathbf {Q}}}

\newcommand{\bY}{{\boldsymbol Y}}
\newcommand{\bm}{\boldsymbol}

\def\half{\frac{1}{2}}

\def\sltr{SL$(2,\mathbb{R})\;$}
\def\sltruod{SL$(2,\mathbb{R})\times$U$(1)^{d-3}\;$}
\def\nhegalgebra{\widehat{\mathcal{V}_{\vec{k},S}}}

\usepackage[inner=4cm,outer=3cm,top=4cm,bottom=4cm]{geometry}
\usepackage{fancyhdr}
\usepackage{epsfig}
\usepackage{graphicx}

\usepackage{epic}
\usepackage{amsmath,amsfonts,amssymb,amsthm,amstext,amscd}
\usepackage{afterpage}

\usepackage{centernot}
\usepackage[all]{xy}
\usepackage{hyperref}
\hypersetup{linktocpage}
\usepackage{epsfig}
\usepackage{graphicx}
\usepackage{tikz}
\usetikzlibrary{decorations.pathmorphing}
\usepackage{caption}
\usepackage{subcaption}
\usepackage{lipsum}
\usepackage{psfrag}
\usepackage{comment}

\newtheorem{theorem}{Theorem}
\newtheorem{proposition}{Proposition}
\theoremstyle{definition}
\newtheorem{definition}{Definition}[chapter]

\begin{document}

\pagenumbering{roman}

\setcounter{page}{1}

\newpage

\thispagestyle{empty}
\begin{center}
  
  {\Huge \bf Conserved charges,\\\vspace*{-.3cm} surface degrees of freedom,\\  \vspace*{.2cm} and black hole entropy}

  \vspace*{2cm}
  {\LARGE\bf Ali Seraj\\}
    \vspace*{2cm}
\textit{Under supervision of\\}\textbf{M.M. Sheikh-Jabbari}
  \vfill

  {\Large A Thesis presented for the degree of\\
         [1mm] Doctor of Philosophy}
  \vspace*{0.9cm}
  
   \begin{center}
\includegraphics[width=0.25\linewidth]{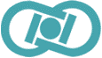}
   \end{center}

  {\large \textbf{Institute for Research in Fundamental Sciences (IPM)}\\
          Department of Physics\\
  \vspace*{0.9cm}
           Tehran, Iran\\
          February 2016}

\end{center}

\clearpage
\thispagestyle{empty}
\begin{figure}
\centering
\includegraphics[width=0.7\linewidth]{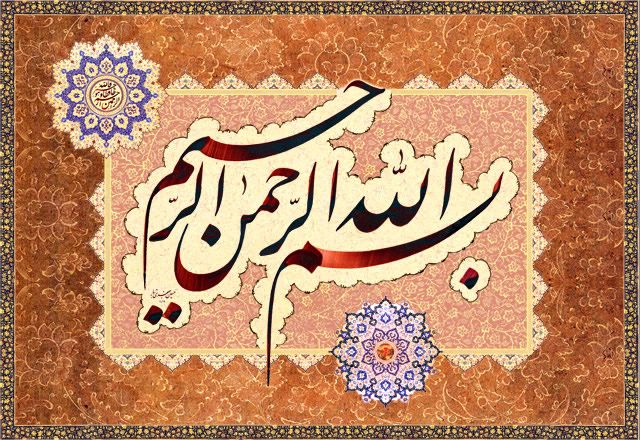}
\end{figure}

\clearpage
\thispagestyle{empty}
\begin{center}
  \vspace*{1cm}
  \textbf{\LARGE Abstract}
\end{center}

In this thesis, we study the Hamiltonian and covariant phase space description of gravitational theories. The phase space represents the allowed field configurations and is accompanied by a closed nondegenerate 2 form- the \textit{symplectic form}. We will show that local/gauge symmetries of the Lagrangian formulation will fall into two different categories in the phase space formulation. Those corresponding to constraints in the phase space, and those associated with nontrivial conserved charges. We argue that while the former is related to redundant gauge degrees of freedom, the latter leads to physically distinct states of the system, known as \textit{surface degrees of freedom} and can induce a lower dimensional dynamics on the system. 

These ideas are then implemented to build the phase space of specific gravitational systems: 1) asymptotically AdS3 spacetimes, and 2) near horizon geometries of extremal black holes (NHEG) in arbitrary dimension.

In the \ads phase space, we show that Brown-Henneaux asymptotic symmetries can be extended inside the bulk of spacetime and hence become symplectic symmetries of the phase space.

We will show that in the  NHEG phase space, surface gravitons form a Virasoro algebra in four dimensions, and a novel generalization of Virasoro in higher dimensions. The central charge of the algebra is proportional to the entropy of the corresponding extremal black hole. We study the holographic description of NHEG phase space and show that the charges can be computed through a Liouville type stress tensor defined over a lower dimensional torus. We will discuss whether surface gravitons can serve as the microscopic origin of black hole entropy. 

\tableofcontents
\clearpage

\newpage
\begin{center}
\textbf{The work in this thesis is based on the results\\ in the following publications}
\end{center}
\vspace{1cm}

\begin{enumerate}

\item G.~Comp\`ere, P.~Mao, A.~Seraj and M.M. Sheikh-Jabbari,
``Symplectic and Killing Symmetries of $AdS_3$ Gravity: Holographic vs Boundary Gravitons'',  [arXiv:1511.06079 [hep-th]].\\

\item   G.~Comp\`ere, K.~Hajian, A.~Seraj and M.~M.~Sheikh-Jabbari,
  ``Wiggling Throat of Extremal Black Holes,''
  JHEP {\bf 1510}, 093 (2015)
  [arXiv:1506.07181 [hep-th]].\\

\item    G.~Comp\`ere, K.~Hajian, A.~Seraj and M.~M.~Sheikh-Jabbari,
    ``Extremal Rotating Black Holes in the Near-Horizon Limit: Phase Space and Symmetry Algebra,''
    \textit{Phys.\ Lett.\ B} {\bf 749}, 443 (2015)
    [arXiv:1503.07861 [hep-th]].\\

\item      K.~Hajian, A.~Seraj and M.~M.~Sheikh-Jabbari,
      ``Near Horizon Extremal Geometry Perturbations: Dynamical Field Perturbations vs. Parametric Variations,''
      JHEP {\bf 1410}, 111 (2014)
      [arXiv:1407.1992 [hep-th]].\\
      
\item        K.~Hajian, A.~Seraj and M.~M.~Sheikh-Jabbari,
        ``NHEG Mechanics: Laws of Near Horizon Extremal Geometry (Thermo)Dynamics,''
        JHEP {\bf 1403}, 014 (2014)  
        [arXiv:1310.3727 [hep-th]].\\

\item K.~Hajian, A.~Seraj, ``Symplectic structure of extremal black holes'', Proceedings of the 2nd Karl Schwarzschild Meeting on Gravitational Physics, 2015, \textit{To appear}

\end{enumerate}


\newpage
\thispagestyle{empty}
\begin{center}
 \vspace*{2cm}
  \textit{\LARGE {Dedicated to}}\\
  \vspace{1cm} 
 {\large My wife \\for her patience and encouragement\\\vspace{1cm}and to my little daughter\\for joining our family in the last 3 months of this project}
\end{center}

\vspace{-1cm}
\chapter*{Acknowledgements}
The work presented in this thesis was completed in the Physics department of Institute for research in fundamental sciences (IPM) during the last three years.
I am extremely grateful to all who helped me finish this work and apologize anyone I have forgot to mention below.

I want to specially thank Prof. M. M. Sheikh Jabbari. He was an excellent teacher, an excellent supervisor, and an excellent collaborator. I can never forget our extended discussions in his office during which I learned how to think and research in physical problems. He was also supportive in many scientific activities I had during these years. I also appreciate other professors in the department specially Y. Farzan and A. Naji for many things I learned in their lectures. 

I am grateful to everyone who made IPM a pleasant place for scientific activities. I specially thank N. Pilehroudi, M. Babanzadeh, S. Jam, J. Aliabadi, H. Zarei and Ms Bagheri for making the physics department a lovely place.

I also thank all past and current PhD students and Postdocs at IPM including: K. Hajian, A. Mollabashi, R. Mohammadi, A. A. Abolhassani, M.H. Vahidinia, M. Noorbala, H. Afshar, H. Ebrahim, R. Fareghbal, M. Ali-Akbari, P. Bakhti, S. Sadeghian, A. Maleknejad. and B. Khavari. We had many useful discussions and spent beautiful moments together. Thanks also to the new students R. Javadinejad, E. Esmaili, H.R. Safari, J. Ebadi and M. Rajaee. 

During our scientific projects, I enjoyed many discussions with Kamal Hajian. Thanks to him, and his novel ideas and intuitions. I am also grateful to Hamid Afshar for great discussions we had in the course of writing this thesis. 

I spent 6 months at the ``Physique Math\'ematique des Interactions Fondamentales'' in ``Universit\`e Libre de Bruxelles'' between Feb. to Aug. 2015. I am specially grateful to Geoffrey Comp\`ere for many things I learned from him and for the very nice collaborations we had. I also thank all the friendly professors and students there, specially Prof. G. Barnich, and graduate students B. Oblak, V. Lekeu, P. Mao and A. Ranjbar. Special thanks go to M. F. Rogge for her great helps during my visit.

In the end, I really appreciate the unconditional love from my parents and my wife and thank them for all their support and encouragement during the years.


\pagenumbering{arabic}
\setcounter{page}{1}

\chapter{Introduction}

\section{Symmetries in physics}
\textit{Symmetry} has played a critical role in the development of modern physics. The notion of symmetry appears in different settings. Specially we would like to distinguish two notions: I) symmetry of a \textit{theory}, II) symmetry of a special \textit{solution}. 

A symmetry of a solution is a transformation which keeps the field configuration of the solution intact. In gravitational physics, these symmetries are usually called isometries and are represented by Killing vectors.

The other notion of symmetry is attributed to a theory which is determined by its action (or equivalently its field equations) together with a set of initial/boundary conditions. A symmetry in this sense is a transformation in fields that does not change the action for all field configuration allowed by the initial/boundary conditions. Note that such a symmetry transformation maps a solution of the theory to another solution. A symmetry can be labeled by either a \textit{discrete} or a \textit{continuous} parameter (or set of parameters). While the former can also lead to charges with discrete values (like parity, time reversal, etc.) here we are interested in the latter. Continuous symmetries  are best described by Lie groups, a rigorous construction in mathematics. 

Continuous symmetries of a theory can be either \textit{global} or \textit{local}. A local symmetry is by definition a transformation in fields parametrized by one or more \textit{arbitrary} functions of spacetime, while global symmetries are specified by a set of parameters. For example any field theory in the context of Special Relativity has a global symmetry with Poincar\'e algebra, wherein the parameters of transformation are rotation angles, boost parameters, and displacement vectors. On the other hand, local symmetries appear in gauge theories like Electrodynamics with a $U(1)$ internal symmetry, or General Relativity which is invariant under local coordinate transformations (diffeomorphisms). 

In this thesis we will be concerned with local symmetries of a theory as described above. In a field theory with local symmetries, field equations determine solutions up to an arbitrary local transformation. Since local symmetries involve arbitrary functions of spacetime, the evolution of the system is not unique, and produces infinitely many solutions. This is because one can transform the solution using a local symmetry such that although the initial state of the system is not altered, the future of the field is changed. However, according to the assumption of \textit{determinism} in theories of physics, we want physical theories to produce unique evolutions. The only way to circumvent this problem is to assume that all solutions obtained by such symmetry transformations are physically equivalent. Therefore solutions of a theory with local degrees of freedom will fall into equivalence classes called \textit{gauge} classes. Because of this, local symmetries are called \textit{gauge} symmetries, and the theory is called a gauge theory. Gauge symmetries describe \textit{redundant} degrees of freedom (gauge freedom) in the theory. 

Another concept that is well known since the early stages of Newtonian mechanics is the notion of ``conservation laws'' and ``constants of motion''. Energy and momentum are probably the most famous constants of motion that satisfy conservation laws. A breakthrough of Emmy Noether was to make a direct link between symmetries and conservation laws. Noether's \textit{first} and \textit{second} theorems \cite{Noether:1918zz} relate symmetries to conserved charges (constants of motion) and constrain the dynamical evolution of the theory under consideration.

Noether's first theorem associates a conserved charge with any symmetry of the theory. The main idea is that corresponding to any symmetry of the Lagrangian, there exists a conserved current$J^\mu$ which is conserved, \ie $\p_\mu J^\mu=0$. The charge defined by an integration over volume $\int d\Sigma_\mu J^\mu$ is then conserved in time. On the other hand, Noether's second theorem applies when the theory possesses gauge symmetries. It puts strong constraints on the form of field equations known as \textit{Bianchi identities}. Using this, one can again define nontrivial charges for a class of gauge symmetries. 

There is still a deeper link between symmetries and conserved charges coming from the Hamiltonian formulation of mechanics. Within the Hamiltonian setup, one can show that a charge \textit{generates} a symmetry transformation through the Poisson bracket. More precisely, a charge is the on-shell value of the generator of a symmetry (either global or local). Energy is simply the on-shell value of the Hamiltonian, which is the generator of time translation. We will explain these issues in detail in chapters  \ref{chapter-Hamiltonian} and \ref{chapter-covariant phase space}. It should be noted that gauge invariance must persist also at the quantum level.
\section{Different notions of conserved charges in Gravitational physics}
There exists an extensive literature on conservation laws and conserved charges in gravitational physics. (see \cite{Szabados:2004vb} for a non-exhaustive review). A natural notion of conserved charges appeared first in the Hamiltonian formulation of general relativity by  Arnowitt, Deser and Misner \cite{Arnowitt:1962hi} as the on-shell value of the Hamiltonian.
In the same context, Regge and Teitelboim stressed the role of surface terms in order to make the Hamiltonian differentiable which leads to a unique definition of charges corresponding to an asymptotic symmetry. Using this method, Brown and Henneaux investigated the Poisson bracket of conserved charges corresponding to asymptotic symmetries and showed the possibility of a central extension \cite{Brown:1986ed,Brown:1986nw}. Applying this to \ads led to the appearance of two Virasoro algebras which can be considered as a first evidence for the celebrated AdS/CFT correspondence. 

Later, attempts were made to give a phase space formulation of gravity, without breaking the covariance of GR, which was essential in Hamiltonian approach. This was achieved \cite{Lee:1990nz,Crnkovic:1986ex} by focusing on the symplectic geometry of Hamiltonian mechanics, called the covariant phase space. The formulation was later implemented \cite{Wald:1993nt,Iyer:1994ys,Iyer:1995kg,Wald:1999wa} to study spacetimes with a boundary and led to a robust proof of the first law of black hole mechanics valid for any diffeomorphism invariant theory of gravity (see also \cite{Hajian:2015xlp} for a recent discussion). A covariant version of previous results in Hamiltonian approach was obtained, and simple formulae for conserved charges, their algebra and the central extension were given in this approach \cite{Koga:2001vq,Hollands:2005wt,Compere:2009dp}.

In another line of thought, employing the Hamilton Jacobi analysis of action functional, Brown and York introduced the \textit{quasilocal charges} \cite{Brown:1992br,Brown:2000dz}. This approach gives a natural definition of the quasi local stress tensor of gravity which is further identified with the stress tensor of the dual field theory in the context of AdS/CFT \cite{Balasubramanian:1999re} and complete consistency was shown. 

In a more advanced mathematical point of view, again a completely covariant approach to conserved charges and asymptotic symmetries was constructed using a variational bicomplex. The main idea is that asymptotic symmetries correspond to cohomology groups of the variational bicomplex pulled back to the surface defined by the equations of motion \cite{Anderson:1996sc}. Conservation laws and central extensions were investigated using BRST techniques \cite{Barnich:2001jy} and further developed in \cite{Barnich:2003xg,Barnich:2007bf,Compere:2007az}.

Other definitions for conserved charges still exist. Ashtekar \etal \cite{Ashtekar:1984zz,Ashtekar:1999jx} used the electric part of the Weyl tensor to define charges in asymptotic AdS spacetimes. 
Abott and Deser used the linearized field equations and symmetries of the background field to define charges for the linear theory \cite{Abbott:1981ff} which can be also used in other gauge theories \cite{Abbott:1981ff}.

It should be noticed that while there are different approaches to the concept of conserved charges in gravity, they can be linked at least in specific contexts. For example an introduction and comparison between different notions of charges in AdS is given in \cite{Hollands:2005wt}, and \cite{Papadimitriou:2005ii} gives a comparison between covariant phase space methods and counterterm methods. Also a comparison of the black hole entropy in covariant phase space and other methods is given in \cite{Iyer:1995kg}.
\section{The mystery of black hole entropy and approaches}
Singularities are generic endpoints of the evolution of matter through gravitational interactions. This is due to the universal attractive nature of gravity and is guaranteed by various singularity theorems \cite{Hawking:1969sw} (see also \cite{Senovilla:2006db} and references therein). On the other hand, the singularity is covered by a hypersurface called the horizon due to the Penrose ``weak cosmic censorship conjecture'', which can be proved by reasonable assumptions (see \cite{Wald:1997wa} for a review). A spacetime containing an ``event horizon'' that prevents a part of spacetime to be in causal contact with the asymptotic region, is called a black hole. Event horizons usually hide a type of singularity (curvature singularity or conical singularity appearing \eg in BTZ black holes) from the outside observer. The formation of a black hole (at least classically) kills all the information about the initial collapsing matter. As a result, the geometry of black holes are uniquely determined by a few parameters like mass, angular momenta, and possibly electric/magnetic charges. This fact is stated in black hole No hair theorems (see \cite{Chrusciel:2012jk} for a review). Laws of black hole mechanics \cite{Bardeen:1973gs} essentially describe the properties of these few parameters and their relations.

However, this is not the end of the story. Hawking's study of a quantum field over a black hole geometry revealed that black holes are radiating with a temperature proportional to their surface gravity. On the other hand, Bekenstein argued that since the horizon divides  spacetime into two parts, requiring laws of thermodynamics forces an observer living in the \textit{outside} to attribute an entropy $S_{BH}$ to the black hole, and revise the second law of thermodynamics as follows \cite{Bekenstein:1974ax}
\begin{center}
$S_{BH}+S_{matter}$ is always increasing.
\end{center}
Upon identifying the entropy of black hole with a quarter of the area of its horizon section, \ie $S_{BH}=\dfrac{A}{4G}$, and its temperature with the Hawking temperature $T_H=\dfrac{\kappa}{2\pi}$, ($\kappa$ being the surface gravity of horizon) laws of black hole mechanics exactly coincide with laws of thermodynamics. These and other pieces of evidence indicate that black holes can be considered as thermodynamic systems having temperature, entropy and other thermodynamic quantities like energy and chemical potentials. 

Boltzman's hypothesis was that thermodynamics originates from the statistical mechanical description of an \textit{underlying theory}. Based on this assumption and that a gas is built from a large number of pointlike objects (atoms), he succeeded in giving a microscopic description of thermodynamics of a system of gas contained in a box, especially its entropy. As we mentioned, black holes express thermodynamic behaviors. Therefore one is tempted to postulate the existence of an underlying theory and try to explain these thermodynamic properties through the statistical properties of that theory. This is the aim of an active research  in black hole physics and quantum gravity. 

Within the context of string theory, Strominger and Vafa used a counting of certain ``BPS states" to give a microscopic derivation of the entropy of specific extremal supersymmetric black holes in five dimensions. This approach was later generalized to include supersymmetric black holes with angular momentum \cite{Breckenridge:1996is} and supersymmetric black holes in four dimensions \cite{Maldacena:1996gb}. In these computations, supersymmetry was a crucial assumption since it allowed them to perform a weak coupling string calculation to deduce the entropy of the semi-classical black holes which exist in the strong coupling regime. However interesting black hole solutions like Kerr black hole which is most similar to astrophysical black holes are not supersymmetric solutions. On the other hand, the universality of the area law, and the fact that entropy is related to the horizon, which involves low energy effects, suggest that the statistical derivation of black hole entropy should be independent of the details of the microscopic (planck scale) physics. This is a familiar result. The entropy of a gas in a box is not sensitive to the nature of the underlying atoms. In modern field theoretic language, entropy is an IR effect independent of UV details. A step in this direction was taken in \cite{Strominger:1997eq} by Strominger. His argument was based on the work of Brown and Henneaux \cite{Brown:1986nw}, that any theory of quantum gravity in AdS3 must be a two dimensional CFT with prescribed central charge. Accordingly for any black hole whose near horizon geometry is locally AdS3, one can use Cardy formula for computing the asymptotic growth of states. Upon explicit evaluation for a non extremal BTZ black hole, he showed that the result coincides with the black hole entropy.

A logical outcome of the Hamiltonian or the covariant phase space approach to gravity is the appearance of ``boundary gravitons''. These are essentially new degrees of freedom that appear when the spacetime has a boundary. For an observer living outside a black hole, horizon can be considered as a boundary, and therefore the appearance of surface gravitons can be expected. A proposal put forward by Carlip and also independently in \cite{Balachandran:1994up} is that the entropy of black hole can originate from these surface degrees of freedom. If this is the case, then one can reconcile between the uniqueness theorems mentioned above and the ``large number of states'' for the black hole suggested by $S_{BH}=\dfrac{A}{4G}$. This is simply because surface degrees of freedom are produced by infinitesimal coordinate transformations in the bulk and therefore are neglected by black hole uniqueness theorems. For the case of BTZ black hole in 3 dimensions, this proposal can be clearly formulated  \cite{Carlip:1994gy,Carlip:2005zn}. Einstein gravity in 3 dimensions can be written as a Chern Simons theory in vielbein formulation \cite{Achucarro:1987vz}. It is also known \cite{Moore:1989yh,Elitzur:1989nr} that Chern-Simons theory on a manifold with boundary induces a dynamical Wess-Zumino-Witten (WZW) theory on the boundary. Carlip argued that the relevant boundary for the black hole entropy (as seen by an outside observer) is the horizon itself. By using a specific fall off condition on the horizon, he determined the precise form of WZW theory on the horizon and counted the number of descendants of the ``vacuum state" and obtained the correct black hole entropy.

\section{Outline; What is new in this thesis?}
The contents of the rest of this thesis fall into two parts. The first part including chapters \ref{chapter-Hamiltonian},\ref{chapter-covariant phase space}. illustrate the theoretical framework of the thesis. Chapter \ref{chapter-Hamiltonian} is a review of the well established Hamiltonian formulation of gauge theories, with an emphasis on conserved charges, generators of gauge symmetries and the emergence of boundary degrees of freedom in open spaces. Chapter \ref{chapter-covariant phase space} introduces the covariant phase space formulation of gauge theories. Although the contents of this chapter already exist in the literature, they are distributed in different references. Here we have tried to give a coherent picture of the construction and its implications. For example, \cite{Lee:1990nz} which originally introduced the covariant phase space of gauge theories, ignored all boundary terms due to a restrictive assumption on the boundary conditions. Also later references focus more on the analysis of conserved charges without paying enough attention to the symplectic structure of the phase space. We hope that this chapter is a good starter for the interested reader. We will again stress the appearance of surface degrees of freedom in this setup.

The second part of the thesis contains essentially the application of the above mentioned construction in gravitational physics, with the motivation of attacking the problem of microscopic  description of black hole entropy. These chapters include the new results of this thesis. In chapter \ref{chapter-AdS3}, we introduce the notion of asymptotic AdS geometry. We then show that accompanied by a suitable symplectic structure, the set of asymptotic \ads geometries form a phase space. Then we discuss the symplectic symmetries of this phase space. Specifically, we show that the asymptotic symmetries of Brown-Henneaux can be fully extended into the bulk and thereby become symplectic symmetries of the phase space. Accordingly, the charges can be computed at any closed surface in the bulk.

In chapter \ref{chapter-NHEG}, we will introduce the near horizon geometry of extremal black holes (NHEG). These are geometries that contain the information about  thermodynamics of extremal black holes, specially their entropy. We will exhibit interesting features about these geometries, \eg that they possess Killing vectors generating bifurcate Killing horizons, whose conserved charge is the entropy of the original extremal black hole. Moreover, we show that NHEG's satisfy laws similar to that of black holes. Then we review the Kerr/CFT proposal, its achievements and shortcomings, especially the fact that there is no linear dynamics allowed over these geometries.

In chapter \ref{chapter NHEG phase space}, having the results of previous chapter in mind, we build the ``NHEG phase space'' that overcomes many problems appeared in Kerr/CFT. The construction works for arbitrary dimensions. We obtain a novel symmetry algebra for the phase space that reduces to Virasoro algebra in four dimensions. We show that the entropy appears as the central charge of the algebra. We also obtain interesting results related to the holographic description.  The results here can open a way to the microscopic description of extremal black hole entropy in four and higher dimensions.
We conclude the thesis by a ``Summary and outlook'' section. Also some technical computations in different sections are gathered in an appendix. 
\chapter{Hamiltonian formulation of gauge theories and gravity}\label{chapter-Hamiltonian}
\section{Introduction}
The program of constructing the Hamiltonian formulation of constrained systems and gauge theories was started by Paul Dirac and P. Bergmann in 50s \cite{,Anderson:1951ta} in order to obtain a systematic way to quantize gauge theories. The Hamiltonian form of Einstein gravity was described by Dirac \cite{Dirac:1958sc,Dirac:1958jc} and later by Arnowitt, Deser and Misner \cite{Arnowitt:1962hi}. Since the definition of canonical momenta forced a special choice of time direction, they used the so called 3+1 decomposition of spacetime which was developed earlier, for other reasons, \eg to prove the unique evolution of GR (The Cauchy problem)\cite{ChoquetBruhat:1969cb}. Since the Hamiltonian formulation is a first order formulation, it was later used extensively in the context of numerical relativity (see \cite{Gourgoulhon:2007ue} for a review). On the theoretical side, Hamiltonian formulation gives deep insights in the study of gauge theories. It turns out that gauge symmetries of the action are related to first class constraints in the phase space. In compact spacetimes, it turns out that the Hamiltonian is a combination of constraints and therefore weakly vanishing. More interestingly, for spacetimes with boundaries, although the bulk term is vanishing on shell, there exists a necessary boundary term that makes the Hamiltonian a nonvanishing variable over the phase space. This raises the notion of boundary degrees of freedom that we will discuss in detail later in this chapter. The dynamics of surface degrees of freedom can be given by a theory in lower dimensions. This resembles (and is indeed an example of) the notion of  holography in gravity.

The organization of this chapter is as follows. We first describe the Hamiltonian formulation of constrained systems in section \ref{sec constrained Hamiltonian}. We then present the Hamiltonian formulation of field theories with local (gauge) symmetries in section \ref{sec Hamiltonian field theory} and stress the relation between local symmetries and constraints. We study the role of spacetime boundary and show that the existence of boundary for the spacetime, leads to the emergence of novel and substantial features. In section \ref{sec GR Hamiltonian}, we specialize to the case of General Relativity as a gauge theory and discuss the notion of asymptotic symmetries. This chapter closely follows references \cite{Henneaux:1992ig,Blagojevic:2002du,Afshar2012}.

\section{Constrained Hamiltonian dynamics}\label{sec constrained Hamiltonian}
In this section we analyze the Hamiltonian mechanics of systems in the presence of \textit{constraints}. These systems are also called singular because of the appearance of arbitrary functions of time in the evolution of the system. This means that the state of the system is not uniquely specified given the initial conditions.  We will show that gauge theories which are of great interest in physics correspond to constrained Hamiltonian systems. Gauge freedom in fundamental fields of a gauge theory then corresponds to the arbitrary functions appearing in the Hamiltonian description. These arbitrary functions in the evolution of fundamental variables are safe in the sense that \textit{observables} which are the physically important quantities are gauge invariant and hence evolve uniquely in time.

\subsection{Particle dynamics}
Lets start with a Lagrangian $L(q_i,\dot{q}_i)$ of a finite number of variables $q_i$ and \textit{velocities} $\dq_i$. The state of the system is determined at each instant of time by the set $(q_i,\dq_i)$. In order to transform to the Hamiltonian description, we assume that the state of the system is given by a \textit{configuration} in the \textit{phase space} $\Gamma$ spanned by the set $(q_i,p^i)$ where the conjugate momenta $p^i$ replace the velocities $\dq_i$. The time evolution is then given through the action principle but this time the Lagrangian is considered as a function of $(q_i,p^i)$. To bring the Lagrangian into this form, we use the usual definition for conjugate momenta
\begin{align}\label{momentum-def}
p^i&=\dfrac{\p L}{\p \dq_i } \,, \qquad i=1,\cdots N\,.
\end{align}
If the above equations can be inverted to give $\dq_i=\dq_i(q_i,p^k)$, then the equations of motion can be obtained by varying the action in terms of $q_i,p^i$ leading to the equations $\dq_i=\p H/\p p^i$ and $\dot{p}^i=\p H/\p q_i$ in which the Hamiltonian function $H(q_i,p^i)$ is defined as the Legendre transform of the Lagrangian
\begin{align}\label{Hamiltonian-def}
H &=p^i \dq_i -L(q_i,\dq_i)\,.
\end{align}
We regarded the Hamiltonian as a function of corrdinates $q_i$ and the conjugate momenta $p^i$. To justify this, we take a variation of \eqref{Hamiltonian-def} with respect to $q_i , \dq^i$
\begin{align}\label{delta H}
\de H&= \de p^i \dq_i +p^i \de\dq_i-\dfrac{\p L}{\p q_i}\de q_i-\dfrac{\p L}{\p \dq_i}\de\dq_i\\
&=\de p^i \dq_i -\dfrac{\p L}{\p q_i}\de q_i\,.
\end{align}
As we see, the dependence on $\de \dq^i$ has appeared only through $\de p^i$.
Therefore we can consider the Hamiltonian as a function of $(q_i,p^i)$. This is the essence of Legendre transformation. By expanding the l.h.s of \eqref{delta H} with respect to $(\de q_i, \de p^i)$ we find
\begin{align}\label{var H 2}
\big(\dfrac{\p H}{\p q_i}+\dfrac{\p L}{\p q_i}\big)\de q_i+\big(\dfrac{\p H}{\p p^i}-\dq^i\big)\de p^i&=0\,.
\end{align}
If the variations $\de p^i, \de q_i$ are all independent (this is not the case when there are constraints as we will explain below), then each term should vanish separately which leads to the Hamiltonian equations of motion after using the Euler Lagrange equations $\dfrac{\p L}{\p q_i}=\dfrac{d}{dt} \dfrac{\p L}{\p \dq_i}=\dot{p}^i$.
\subsection*{Constraints}
Equation \eqref{momentum-def} is invertible if there is a one to one correspondence between the state space $(q_i,\dq_i)$ and the phase space $(q_i,p^i)$. This happens if an infinitesimal variation of $\dq_i$ induces an infinitesimal \textit{nonvanishing} variation in the conjugate momenta. Taking a variaiton of \eqref{momentum-def} with respect to $\dq_i$ we have 
\begin{align}
dp^i=\dfrac{\p^2 L}{\p\dq_i\p\dq_j}d\dq_j\,.
\end{align}
If the Hessian matrix $\mathcal{J}\equiv\dfrac{\p^2 L}{\p\dq_i\p\dq_j}$ is degenerate (\ie its determinant is vanishing) then there exist an infinitesimal  variation $d\dq_i$ that corresponds to no variation in the phase space. Therefore if $\det \mathcal{J}=0$ then the set of states $q_i,\dq_i$ map to a submanifold $\Gamma_1$ of the phase space. This submanifold can be represented by a set of constraints
\begin{align}\label{primary constraints}
\phi_m(q_i,p^i)=0,\qquad m=1\cdots M
\end{align}
on the phase space. In the presence of constraints, the map \eqref{momentum-def} does not define an invertible map. The price of making the Legendre transformation invertible, is to add new variables to the phase space as we explain below. The point is that in the presence of constraints, \eqref{var H 2} is not valid for \textit{all} variations, but only for those \textit{tangent} to the constraint submanifold $\Gamma_1$, \ie those preserving the constraints. Now we need the following theorem (see chapter 1 of \cite{Henneaux:1992ig} for proof),
\begin{theorem}
If $\lambda^i \de q_i +\mu_i \de p^i =0$ for arbitrary variations $\de q_i,\de p^i$ tangent to the constraint surface \eqref{primary constraints}, then 
\begin{align}
\lambda^i&=u^m \dfrac{\p \phi_m}{\p q_i}\\
\mu_i&= u^m \dfrac{\p \phi_m}{\p p^i}\,,
\end{align}
where $u^m$ are arbitrary funcitons on the phase space.
\end{theorem}
\noindent According to the above theorem, \eqref{var H 2} in the presence of constraints \eqref{primary constraints} implies the following relations
\begin{align}
\dq^i&=\dfrac{\p H}{\p p^i}+u^m \dfrac{\p \phi_m}{\p p^i}\\
-\dfrac{\p L}{\p q_i}&=\dfrac{\p H}{\p q_i}+u^m \dfrac{\p \phi_m}{\p q_i}\,.
\end{align}
The first equation is particularly important since it implies that \textit{if} we extend $\Gamma_1$ by the new coordinates $u^i$, then one can invert the Legendre map to obtain $\dq_i$ in terms of $p^i$ (constrained by $\phi_m=0$) and the new variables $u^m$. The above equations together with the Euler Lagrange equations imply the Hamiltonian equations in the presence of constraints
\begin{align}\label{Hamiltonian eom }
&\dq^i=\dfrac{\p H}{\p p^i}+u^m \dfrac{\p \phi_m}{\p p^i}\\
-&\dot{p}^i=\dfrac{\p H}{\p q_i}+u^m \dfrac{\p \phi_m}{\p q_i}\\
&\phi_m(q_i,p^i)=0\,.
\end{align}
These equations can be obtained through the following action 
\begin{align}
S&=\int dt (p^i\dq_i-L-u^m\phi_m)\,,
\end{align} 
in which $u^m$ appear naturally as a set of \textit{Lagrange multipliers}. As we see, the equations of motion involve arbitrary functions $u^m$. These equations can be written in an elegant way using the Poisson bracket. The Poisson bracket between two functions on the phase space is defined as
\begin{align}\label{Poisson bracket Hamiltonian def}
\{f,g\}&=\dfrac{\p f}{\p q_i}\dfrac{\p g}{\p p^i}-\dfrac{\p f}{\p q_i}\dfrac{\p g}{\p p^i}\,.
\end{align}
Using the Poisson bracket, the equations of motion can be written in the compact form
\begin{align}
\dot{f}&=\{f,H\}+u^m\{f,\phi_m\}\,,
\end{align}
where $f$ can be $q_i$ or $p^i$ or any function of them. However we can still simplify by introducing the \textit{total Hamiltonian} $H_T$
\begin{align}
H_T&=H+u^m\phi_m\,.
\end{align} 
Using this, we have
\begin{align}
\dot{f}&=\{f,H_T\}\,.
\end{align}
Note that there is a term $\{f,u^m\}\phi_m$ that we neglected because it is multiplied by constraints and vanish on shell.
\subsubsection*{Primary and secondary constraints}
The constraints appearing in the total Hamiltonian are called \textit{primary constraints}. These are the constraints resulting from \eqref{momentum-def}. The theory is consistent if the primary constraints are preserved in time. That is 
\begin{align}\label{consistency constraints}
\dot{\phi}_n&= \{\phi_n,H\}+u^m \{\phi_n,\phi_m\}\approx 0\,.
\end{align}
These equations can lead to further constraints  on the phase space. These are called \textit{secondary constraints} which involve making use of equations of motion.  Note that secondary constraints  should also satisfy \eqref{consistency constraints} that can lead to further secondary constraints. This procedure should end somewhere and we remain with a set of constraints denoted by $\phi_i$ where $i$ runs over all primary and secondary constraitns. Since the number of equations  in \eqref{consistency constraints} (labeled by $n$) is equal to the number of all constraints, which is more than the number of unknown secondary constraints, these equations also restrict the form of arbitrary functions $u^m$. Assuming that we have the complete set of constraints at hand, we can view \eqref{consistency constraints} as a set of equations for the unknown functions $u^m$. This is an inhomogeneous equation and therefore the solution takes the form
\begin{align}
u^m=U^m+v^a V^m_a\,,
\end{align}
where $U^m$ is a special solution to the inhomogeneous equation, and $V^m_a$ are independent solutions to the homogeneous equation ${\phi_n,\phi_m}u^m\approx 0$ and $v^a$ are arbitrary functions. Accordingly the total Hamiltonian can be written as 
\begin{align}
H_T&=H'+ v^a(t) \phi_a, \qquad \phi_a=V^m_a \phi_m\\
H'&=H+U^m\phi_m
\end{align}
and the time evolution is given by 
\begin{align}\label{evolution v^a}
\dot{f}&=\{f,H'\}+v^a(t)\{f,\phi_a\}\,.
\end{align}
As we see, the dynamical equations still involve arbitrary functions of time and the evolution of canonical variables is not unique. 
\subsubsection{First class and second class constraints}
A function $f$ is called \textit{first class} (FC) if its Poisson bracket with all the constraints weakly vanish
\begin{align}
\{f,\phi_i\}&\approx 0\,.
\end{align} 
It can be shown that this is equivalent to 
\begin{align}\label{FC def}
\{f,\phi_i\}&=C_{ij}\phi_j\,.
\end{align}
A function which is not first class is called \textit{second class}. Therefore we can separate the constraints into first class and second class. According to \eqref{FC def}, the set of first class constraints form a closed algebra 
\begin{align}
\{\phi_I,\phi_j\}&=f_{ij}^{\;\; k}\phi_k\,.
\end{align}
On the contrary, for a second class constraint, there exists at least one constraint such that their Poisson bracket is not a constraint anymore. Since first class constraints form a closed algebra, they are best suited to represent gauge symmetries of a theory, since gauge symmetries also form an algebra. In the following, we will show that this is indeed the case. Any gauge symmetry maps to a first class constraint in the phase space.

As we saw, the Hamiltonian involves arbitrary functions of time, $v^a(t)$ that affect the time evolution of functions. We want to see how the evolution of a function $f$ is affected by a \textit{gauge transformation} $v^a\to v^a+\de v^a$. Using  \eqref{evolution v^a} once with $v^a$ and once with $v^a+\de v^a$ and subtracting, we find
\begin{align}
\de \dot{f}&= \de v^a \{f,\phi_a\}\,.
\end{align}
This indicates that the constraint $\phi_a$ generates the gauge transformation $v^a\to v^a+\de v^a$. Note that $H'$ and $\phi_a$ are independent solutions to the equation \eqref{consistency constraints}, and therefore they weakly commute with all constraints. This especially implies that $\phi_a$ are first class constraints. Therefore we conclude that \textit{gauge transformations are generated by first class constraints}.
\subsubsection{Dirac bracket and second class constraints}
Consider the matrix $\Delta_{mn}=\{\chi_m,\chi_n\}$ where $\chi_m$ denotes second class constraints. If $\Delta$ is degenerate, there exist a vector $\lambda^n$ such that $\Delta_{mn}\lambda^n=0$. Therefore 
\begin{align}
\{\chi_m,\chi_n\}\lambda^n=\{\chi_m,\lambda^n\chi_n\}=0\,,
\end{align} 
which means that there exists a constrain $\chi=\lambda^n\chi_n$ that commutes with all second class constraints. Moreover $\chi$ commutes with all FC constraints by definition. But this is a contradiction since a constraint that commutes with all constraints is a first class constraint. 
Therefore we conclude that $\Delta$ is nondegenerate and invertible. Denoting the inverse matrix by upper indices, the Dirac bracket is then defined as 
\begin{align}
\{f,g\}_{DB}&=\{f,g\}-\{f,\chi_m\}\Delta^{mn}\{\chi_n,g\}\,.
\end{align}
It can be checked that this new bracket has all the properties of a bracket. In addition, the Dirac bracket of a second class constraint with any function is vanishing, while the Dirac bracket of a first class function with any other function is equal to their Poisson bracket on shell. Therefore by replacing Poisson bracket with Dirac bracket, the second class constraints are automatically satisfied (\ie we can forget about them), while first class constraints still generate gauge transformations.

\section{Hamiltonian dynamics of field theories}\label{sec Hamiltonian field theory}
A field theory involves a continuous number of degrees of freedom. Therefore the formalism developed in previous section should be extended. This can be done formally by the following replacements
\begin{align}\label{particle to field trans}
q_i\to \Phi_i(x), \qquad \sum_i \to \sum_i \int d^d x , \qquad \dfrac{\p}{\p q_i}\to \dfrac{\de}{\de \Phi_i (\mathbf{x})} , \qquad \delta_i^j \to \delta_i^j \de(\mathbf{x}-\mathbf{y}).
\end{align}
$\Phi_i(x)$ is called a field over spacetime. The Lagrangian is given by $L=\int_{\Sigma}d^d x \mathcal{L}$, where $\Sigma$ is the spacelike surface of constant time. Then the momentum density is defined as $\pi^i(x)=\dfrac{\p \cL}{\p \dot{\Phi}_i(x)}$ and canonical Hamiltonian by 
\begin{align}\label{Hamiltonian-density}
H=\int_\Sigma \mathcal{H},\qquad \mathcal{H}=\pi^i(x)\dot{\Phi}_i(x)-\cL(x)\,.
\end{align}
Note that ``time'' is singled out from scratch, in the definition of Lagrangian, conjugate momenta, etc.

\subsection{Gauge symmetries, constraints, and generators}\label{gauge and constraints sec}

The unphysical transformations of dynamical variables are referred to as gauge transformations or gauge symmetries. Usually, by gauge symmetry, we mean \textit{local} gauge symmetries, \ie the parameters of the transformations are arbitrary local functions over spacetime. However, we should declare what is meant by unphysical. The answer is that unphysical transformations are those not affecting the \textit{observables}. Again it is not clear what is meant by an observable. Observables are those quantities that can be measured in an experiment. Indeed, the set of observables should be specified as an \textit{input} of the theory. For example in Electromagnetism, we postulate that observables are those quantities that can affect the acceleration of test charges, \ie electric field $\vec{E}$ and magnetic field $\vec{B}$, or equivalently $F_{\mn}$. Actually it turns out that there are more observables which are nonlocal. The integral $\oint d\ell\cdot \vec{A}$ over a closed curve is also an observable that can be measured \eg in the Aharanov-Bohm effect. This is an example of Wilson loops in field theory. However, in the Hamiltonian formulation, we have a more handy definition of gauge transformations. Remember that in the Hamiltonian formulation of constrained systems, the time evolution of canonical quantities involve arbitrary functions of time $v^a(t)$ (see equation \eqref{evolution v^a}). This makes the dynamics of the system nondeterministic. Therefore in order to cure this problem, we postulate that observables are those quantities that are not affected by changing $v^a(t)$. Accordingly, we define a gauge symmetry as follows\\
\begin{definition}
A gauge symmetry is a local transformation in canonical variables $(\de q_i,\de p^i)$, that map a solution to the Hamiltonian equation \eqref{evolution v^a} with the parameters $v^a$ to a solution of the same equation with another set of parameters $v^a+\de v^a$.
\end{definition}
A variation in parameters $v^a$ is transferred to the canonical variables through the equations of motion. Therefore corresponding to each parameter $v^a$ there exists a local gauge symmetry. Since the number of arbitrary parameters $v^a$ is equal to the number of primary first class constraints, we see that the number of gauge  degrees of freedom is at least equal to the number of first class primary constraints. However, it turns out that secondary first class constraints can also generate a gauge symmetry. Dirac conjecture was that there is a correspondence between gauge symmetries and first class constraints in phase space. Although this is not true in general, it can be proved with additional assumptions that are satisfied in physically interesting theories \cite{Henneaux:1992ig}. 

\subsection{Local symmetries and constraints}\label{Castellani}
Now we show that any gauge symmetry is generated by a combination of constraints. That is 
\begin{align}\label{gauge generator}
\de_\eps f&=\{f,G_\eps\}\,,
\end{align}
where $G_\eps$ is called the generator of the gauge symmetry $(\de_\eps q_i,\de_\eps p^i)$. A systematic way to construct the generator of all gauge symmetries of the equations of motion from first class constraints was given by Castellani \cite{Castellani:1981us}. We state the result without proof here. (see  \cite{Castellani:1981us,Blagojevic:2002du} for a proof.)

Take any primary first class constraint $G_0$. Then define $G_1$ such that 
\begin{align}
G_1 + \{G_0,H_T\}=PFC\,,
\end{align}
where $PFC$ means a combination of primary first class constraints. Then define $G_2$ such that $G_2 +\{G_1,H_T\}=PFC$. We should continue this procedure until we find $G_k$ such that 
\begin{align}
\{G_k,H_T\}=PFC\,.
\end{align}
Then the following combination will be the generator of a gauge transformation that can be easily obtained by \eqref{gauge generator}
\begin{align}\label{Generator series}
\cal G _\eps&=\eps^{(k)} G_0 +\cdots + \eps'\, G_{k-1} +\eps\, G_k\,,
\end{align}
where $\eps^{(n)}={d^n \eps}/{dt^n}$. In the case of a field theory, $\eps=\eps(t,x^\mu)$ and the generator of the gauge transformation is 
\begin{align}\label{generator FT}
G_\eps &=\int d^dx \mathcal{G} _\eps \,,
\end{align}
where $\cal G_\eps$ is given by \eqref{Generator series}. 

Moreover, it should be noted that in the case of a field theory with a boundary, there is in principle a set of boundary conditions over the dynamical fields. Therefore the symmetry transformation $\de_\eps \Phi=\{\Phi,G_\eps\}$ should also respect the boundary conditions. This restricts the set of allowed parameters $\eps(x^\mu)$. We will discuss this in more detail in the discussion of asymptotic symmetries in section \eqref{sec asymptotic symmetries}. 

\subsection{Example: Maxwell theory}
The Lagrangian of Maxwell theory is given by
\begin{align}
\cL&=-\dfrac{1}{4}F_{\mn}F^{\mn}, \qquad F_\mn =\p_\mu A_\nu -\p_\nu A_\mu
\end{align}
with the dynamical fields $A_\mu$. To employ the Hamiltonian picture, we have to determine the conjugate momenta
\begin{align}
\pi^\mu&=\dfrac{\p \cL}{\p \dot{A}_\mu}=F_{0\mu}\,.
\end{align}
Due to the antisymmetry of $F_\mn$, the conjugate momentum $\pi^0$ is a constraint
\begin{align}
\pi^0=0\,.
\end{align}
The canonical Hamiltonian after imposing the primary constraint $\pi^0=0$ and defining $A_\mu=(V,\vec{A})$ is
\begin{align}
\mathcal{H}&=\vec{\pi}\cdot\vec{A}-\cL \nnr
&=\dfrac{1}{2}(\vec{\pi}^2+\vec{B}^2)-\vec{\pi}\cdot \nabla V \nnr
&=\dfrac{1}{2}(\vec{\pi}^2+\vec{B}^2)+V \vec{\nabla}\cdot\vec{\pi} -\nabla\cdot(V\vec{\pi}  )\,,
\end{align}
where $\vec{B}=\vec{\nabla}\times \vec{A}$. The last term is a total derivative and gives no contribution to the integral \eqref{Hamiltonian-density} with the usual assumptions on electromagnetic fields. Therefore 
\begin{align}
H_c&=\int d^nx \; \Big(\dfrac{1}{2}(\vec{\pi}^2+\vec{B}^2)+V \vec{\nabla}\cdot\vec{\pi}\Big)\,.
\end{align}
 The phase space is now given by the $\vec{A}$, its conjugate $\vec{\pi}=-\vec{E}$ together with the Lagrange multiplier $V$. The total Hamiltonian is 
\begin{align}\label{H_T EM}
H_T&=H_c+\int_\Sigma d^nx\; u \,\pi^0 \,.
\end{align}
The consistency condition \eqref{consistency constraints} requires
\begin{align}
\dot{\pi}^0=\{\pi^0,H_T\}=\vec{\nabla}\cdot \vec{\pi}=0
\end{align}
and the Hamiltonian equations \eqref{Hamiltonian eom } read
\begin{align}
\dot{\vec{A}}&=-\vec{\pi}-\vec{\nabla}V=-\vec{E}-\vec{\nabla} V\\
\dot{\vec{\pi}}&=\dot{\vec{E}}=-\vec{\nabla}\times \vec{B}\,.
\end{align}
We see that one of the Maxwell equations appears as a constraint equation while others are dynamical equations. Now let us study the generator of gauge symmetries according to the Castellani construction. Since here we have one primary first class constraint $\phi_1=\pi^0$ and a secondary constraint $\phi_2=\nabla\cdot \vec{\pi}$, the generator takes the form
\begin{align}
G_\eps&= \dot \eps \,G_0 + \eps\, G_1\,.
\end{align}
Since $G_0$ should be primary first class constraint, $G_0=\pi^0$. Then $G_1$ should satisfy 
\begin{align}
G_1 + \{G_0,H_T\}=PFC\,.
\end{align}
According to \eqref{H_T EM} 
\begin{align}
\{G_0,H_T\}&=\{\pi^0,H_T\}=\nabla\cdot \vec{\pi}=\phi_2
\end{align}
since $\{\phi_2,H_T\}=0$, the chain stops by choosing $G_1=-\nabla\cdot \vec{\pi}$. Therefore the generator is 
\begin{align}\label{Maxwell generator}
G_\eps&=\int d^nx \big(\dot \eps\, \pi^0 -\eps \, \nabla\cdot \vec{\pi}\big)\,.
\end{align}
The gauge symmetry generated by $G_\eps$ is 
\begin{align}
\de_\eps A_\mu(x)=\{A_\mu(x),G_\eps\}&=\int d^nx' \big(\dot \eps\, \{A_\mu,\pi^0\} - \eps\, \{A_\mu,\nabla\cdot \vec{\pi}\}\big)\nnr
&=\int d^nx' \big(\dot \eps\, \de_\mu^{\;0} \,\de(x-x') + \p_i\eps\, \{A_\mu,{\pi}^i\}\big)\nnr
&=\int d^nx' \de(x-x')\big(\dot \eps\, \de_\mu^{\;0}\, + \p_i\eps\, \de_\mu^{\;i}\, \big)\nnr
&=\p_\mu \eps(x)\,.
\end{align}
In the second line we have also used an integration by part for the last term. The total derivative term drops using the usual boundary conditions. This is the well known gauge symmetry of electrodynamics.

\section{Open spacetimes and conserved charges}\label{sec conserved charges-Hamiltonian}
The Hamiltonian formalism relies on the action principle, \ie that the action should be \textit{stationary} over a solution
\begin{align}
\de \int dt \int d^dx (\pi^i \dot \Phi_i-H)\approx 0\,.
\end{align}
This implies that the Hamiltonian should be a \textit{differentiable} functional. That is 
\begin{align}\label{differentiable}
\de H&= \int d^dx\, \dfrac{\de H}{\de \psi}\de \psi=\int d^dx \,\Big(\dfrac{\de H}{\de \Phi_i}\de \Phi_i+\dfrac{\de H}{\de \pi^i}\de \pi^i\Big)\,,
\end{align}
where $\psi$ denotes the collection $(\Phi_i,\pi^i)$ of all canonical fields and their conjugate momenta. However, if one starts with the canonical Hamiltonian, 
it can be checked that the variation $\de H_c$ is not of the form \eqref{differentiable}, but also includes a boundary integral. This problem was investigated by Regge and Teitelboim \cite{Regge:1974zd}. They argued that for the generator of any gauge symmetry, (including the Hamiltonian as the generator of time translation) one has  to add a suitable boundary term such that the \textit{improved generator} become differentiable. 

Let us start with the generator of a gauge symmetry \eqref{generator FT} as constructed in \eqref{Castellani}
\begin{align}
G_\eps&=\int_\Sigma d^dx \,\mathcal{G}_\eps[\psi]\,.
\end{align}
The variation $\de G$ under variations allowed by the boundary conditions will generally take the form
\begin{align}
\de G_\eps&=\int_\Sigma d^dx \dfrac{\de \cal G_\eps}{\de \psi}\de \psi +\oint_{\p\Sigma} B[\eps,\de\psi,\psi]\,.
\end{align}
Now try to find boundary term $Q_\eps[\psi]=\oint_{\p\Sigma}\mathcal{Q}_\eps[\psi]$ such that 
\begin{align}\label{Q_xi def}
\de Q_\eps=-\oint B[\eps,\de\psi,\psi]\,.
\end{align}
Then the \textit{improved generator} $\tilde{G}_\eps$ is defined as
\begin{align}
\tilde{G}_\eps&=G_\eps+Q_\eps\,.
\end{align}
By construction, $\tilde{G}_\eps$ is differentiable under allowed variations, and generates the gauge symmetry with the parameter $\eps$. This is because the Poisson bracket is a local operator and adding boundary terms does not alter the role of $G_\eps$ as the generator of a gauge symmetry. 

The possibility of finding $Q_\eps$ with the property \eqref{Q_xi def} relies on the existence of  consistent boundary conditions. A good boundary  condition should be defined such that the generator of all allowed gauge transformations be differentiable and its corresponding charge be zero or finite.

However, the improved generator is not vanishing on-shell anymore. Its on-shell value gives the \textit{charge} $Q_\eps$ corresponding to the gauge parameter $\eps$. Especially the Hamiltonian will be 
\begin{align}
H&= H_c + Q_{\p_t}\,,
\end{align}
and hence the on-shell value of Hamiltonian is equal to the energy of the gauge system. Below we list important properties of $G_\eps$ and $\tilde{G}_\eps$

The generator $G_\eps$ is a combination of constraints. However, it is not a first class functional anymore. This is because the Poisson bracket of two such generators 
\begin{align}
\{G_{\eps_1},G_{\eps_2}\}&=\de_{\eps_2}G_{\eps_1}=G_{\eps_3}+\oint B[\eps_1,\de_{\eps_2}\psi,\psi]\,,
\end{align}
does not close the algebra. The existence of boundary for the spacetime, turns the generators from first class constraints into \textit{second class} constraints. 

Thinking in terms of the improved generator $\tilde{G}_\eps$, it was shown \cite{Brown:1986ed} that their algebra closes (up to a central extension \cite{Brown:1986nw}) but it should be noted that the improved generators are not pure constraints anymore. They involve boundary terms that are varying over the phase space. 

The subset of local symmetry transformations that correspond to generators with nontrivial charges form an algebra. They are called the \textit{asymptotic symmetries }of the system. The algebra of improved generators of asymptotic symmetries is a central extension of the Lie algebra of these local symmetry transformations. The important point is that field configurations obtained by acting these local symmetry transformations on a given field configuration, can be labeled by the charges associated to asymptotic symmetries. Therefore if we consider the charges as observables, then these ``diffeomorphic'' configurations are distinct physical states of the system. These facts are stated in the following proposition.
\begin{proposition}\label{prop boundary dynamics}
In the presence of boundaries for the spacetime, some of the gauge symmetries that correspond to nontrivial charges are not gauge transformations anymore. They produce new states of the system. We call these novel states, ``boundary degrees of freedom'' (boundary gravitons in GR, or boundary photons in EM, etc.).
\end{proposition}
We will see examples of these states in chapters \eqref{chapter-AdS3} and \eqref{chapter NHEG phase space}. Other examples in flat holography or in other theories of gravity can be found in \eg \cite{Afshar:2011qw,Afshar:2012nk,Afshar:2013bla}.

\section{Hamiltonian formulation of General Relativity}\label{sec GR Hamiltonian}

In order to define canonical momenta, we need to foliate the spacetime into a family of spacelike hypersurfaces, each of which representing an instant of time. This is usually called the 3+1 decomposition (in four dimensional spacetime) and has a long history dating back to the advent of general relativity. We first need a ``time function'' $t(x^\mu)$. The hypersurfaces $t=const$ determine the constant time surfaces. The function $t(x^\mu)$ should be such that the hypersurfaces be spacelike, \ie $n_\mu\propto \p_\mu t$ be a future directed timelike vector field.
On each $t=const$ hypersurface, one can define an independent coordinate system $y^a(x^\mu)$. However, in order to specify time evolutions of the system, one should relate these coordinate systems. That is, one should specify which point does the $y^a=const$ determines at each time. In order to do this systematically, we can define a timlike congruence of curves $\gamma$ intersecting each constant time hypersurface once. The parameter $t$ can be used to parametrize the curves. We postulate that the intersection of a specific curve with each constant time hypersurface determines the same coordinate $y^a$ on $\Sigma_t$. Hence defining a coordinate system on one hypersurface induces a coordinate system on all hypersurfaces. The set $(t,y^a)$ can be considered as a new coordinate system for the spacetime. The tangent vector to the curves is given by $t^\mu=\frac{dx^\mu}{dt}$ and we have
\begin{align}
t^\mu \p_\mu t&= \dfrac{dx^\mu}{dt}\dfrac{dt}{dx^\mu}=1\,.
\end{align}
On the other hand, each hypersurface has a normal vector $n_\mu$ as well as $d$ vectors tangent to each constant time hypersurface which we call $e_a, \; a=1,2,\cdots, d$
\begin{align}
(e_a)^\mu&=\dfrac{\p x^\mu}{\p y^a}\,.
\end{align}
Note that the curves $\gamma$ do not intersect $\Sigma_t$ orthogonally, hence we may write 
\begin{align}\label{lapse-shift}
t^\mu&=N n^\mu +N^a e_a^\mu\,.
\end{align}
$N,N^a$ are called the \textit{lapse} function and the \textit{shift} vector respectively. Now we can use the above construction to write the metric in the coordinate system $(t,y^a)$. Since $x^\mu=x^\mu(t,y^a)$, we can write
\begin{align}
dx^\mu&=\dfrac{dx^\mu}{dt}dt +\dfrac{dx^\mu}{dy^a}dy^a\nnr
&=t^\mu dt + e_a^\mu dy^a\nnr
&=(Ndt)n^a +(dy^a+N^a dt)e_a^\mu \,.
\end{align}
Now
\begin{align}
ds^2&=g_{\mn}dx^\mu dx^\nu \nnr
&= (N dt)^2 n\cdot n + (dy^a+N^a dt)(dy^b+N^b dt)e_a\cdot e_b +2 Ndt (dy^a+N^a dt) n\cdot e_a\,
\end{align}
where $n\cdot n=g_\mn n^\mu n^\nu$ and so on. However, by construction $n\cdot n=-1$ and $n\cdot e_a=0$, therefore
\begin{align}\label{3+1 decomposition}
ds^2&=-N^2 dt^2+h_{ab}(dy^a+N^a dt)(dy^b+N^b dt)
\end{align}
where $h_{ab}=g_\mn e_a^\mu e_b^\nu$ is the induced metric on $\Sigma_t$. Using \eqref{3+1 decomposition} one can show that $\sqrt{-g}=N \sqrt{h}$.

\subsection*{From action to Hamiltonian}
We start with the Einstein Hilbert action in $d+1$ dimensions accompanied by the Gibbons-Hawking-York boundary term. The boundary term leads to a well defined variational principle for the case of Dirichlet boundary conditions $\de g_{\mn}=0$ on the boundary\footnote{See section \ref{sec-AAdS} for the proper definition of Dirichlet boundary conditions for the case of AdS spacetime where the metric is degenerate on the boundary}. The action is   
\begin{align}
16\pi G S_G&=\int_{\cal M}d^{d+1} x \sqrt{-g}\,R  +2\oint_{\p \cal M}d^{d}y \sqrt{h}\,K 
\end{align}
where $K$ is the trace of the extrinsic curvature $K_\mn$ of the boundary. In order to define the Hamiltonian, we consider the fields $N,N^a,h_{ab}$ as the basic phase space coordinates and rewrite the action in terms of these variables and their derivatives. The conjugate momentum is defined as usual through $p_a=\frac{\p \cal L}{\p \dot{q}^a}$.   The conjugate momentum of $h_{ab}$ is 
\begin{align}
\pi^{ab}&= \frac{\p \mathcal{L}_G}{\p \dot{h}_{ab}}=\dfrac{1}{16\pi}(K^{ab}-K h^{ab})
\end{align}
where $\mathcal{L}_G$ is the volume part of the gravitational Lagrangian. Note that the boundary term is independent of $\dot{h}_{ab}$. On the other hand, since the action does not involve $\dot{N},\dot{N}^a$, their conjugate momenta are constrained to  zero. These are the \textit{primary constraints} of Einstein gravity. Namely
\begin{align}
\pi^{00}=\dfrac{\p\cL}{\p \dot{N}}&\approx 0,\qquad \pi^{0a}=\dfrac{\p\cL}{\p \dot{N}^a}\approx 0\,.
\end{align}
After performing the standard procedure, we find the Hamiltonian 
\begin{align}\label{generator Hamiltonian}
16\pi H_G&=\int_{\Sigma_t}(N \mathcal{H}+N_a \,\mathcal{H}^a)\sqrt{h}\,d^{n}x + 2\oint _{S_t} [N \mathcal{P}+N_a \mathcal{P}^a]\sqrt{\sigma}\,d^{n-1}x
\end{align}
where
\begin{align}
\mathcal{H}&=K^{ab}K_{ab}-K^2-{}^{(3)}R\,,\qquad \mathcal{H}^a=D_b(K^{ab}-K h^{ab})\,.
\end{align}
Here $K_{ab}$ is the extrinsic curvature of the constant time hypersurface and ${}^{(3)}R$ is the Ricci scalar of the induced metric $h_{ab}$. Also   
\begin{align}
\mathcal{P}&=k-k_0\,,\qquad \mathcal{P}^a=-r_b(K^{ab}-K h^{ab})
\end{align}
where $\sigma_{AB}$ is the induced metric on the sphere $S_t$ as the boundary of $\Sigma_t$ at infinity, and $k=\sigma^{AB}k_{AB}$ is the trace of extrinsic curvature of $S_t$ as embedded in $\Sigma_t$. $k_0$ is a regularization term which can be taken as the same quantity computed over a suitable background geometry. Also $r_a$ is the normal vector to the boundary.

Now the consistency conditions $\dot{p}^{0\mu}\approx 0$ implies the \textit{secondary constraints}
\begin{align}
\dot{\pi}^{00}=\{\pi^{00},H_G \}=\mathcal{H}\,,\qquad \dot{\pi}^{0a}=\{ \pi^{0a},H_G \}=\mathcal{H}^a\,.
\end{align}
Therefore $\mathcal{H}^\mu\equiv(\mathcal{H},\mathcal{H}^a)$ appearing in the Hamiltonian are called the \textit{Hamiltonian constraint} and \textit{momentum constraints} respectively which are  related to the $G_{00},G_{0a}$ components of the Einstein equations and hence vanishing on-shell. Note that there are no further constraints since $\dot{\mathcal{H}}=\dot{\mathcal{H}}^a=0$ automatically. Also the constraints are \textit{first class} since $\{\pi^{0\mu}, \cH^\nu\}=0$.

Now that we have the complete set of constraints, we can construct the generator of gauge transformations by the Castellani procedure as we explained in section \eqref{Castellani}. Since there are $n$ primary constraints, the most general form of the generator  involves $n$ parameters which we denote collectively by $\xi_\mu(x)$. Then the generator is
\begin{align}\label{generator GR}
G_{\xi^\mu}&=\int_\Sigma d^{n}x\,\sqrt{h}\, \Big(\dot{\xi}_\mu \pi^{0\mu}-\xi_\mu \cH^\mu\Big)
\end{align}
and one can check that 
\begin{align}
\{g_{\mn},G_\xi\}&=\nabla_{\mu}\xi_{\nu}+\nabla_{\nu}\xi_{\mu}=\mathcal{L}_\xi g_{\mn}
\end{align}
which is the gauge symmetry of General Relativity.

\subsection{Asymptotic symmetries and conserved charges }\label{sec asymptotic symmetries}
Since the choice of the time function, equivalently $n^\mu,e^\mu_a$, and also the vector congruence $t^\mu$ is arbitrary, therefore the values of $N,N^a$ can be arbitrarily deformed. Therefore this deformation can be seen as a \textit{surface deformation} if one deforms the constant time surfaces, or as a change of \textit{Hamiltonian flow} if one deforms the congruence. As we saw above, the on shell value of Hamiltonian is 
\begin{align}\label{charge-Hamiltonian}
Q_T&=H_G\Big\vert_{on shell}= \dfrac{1}{8\pi} \oint _{S_t} [N \mathcal{P}+N_a \mathcal{P}^a]\sqrt{\sigma}\,d^{n-1}x
\end{align}
which explicitly depends on the choice of $N,N^a$. Therefore associated to each choice of time direction (or equivalently associated to each surface deformation) allowed by the boundary conditions, one can associate a conserved charge.

Most famously, mass and angular momentum are defined by ADM using the above equation. Mass is defined when the Hamiltonian vector flow is an \textit{asymptotic time translation}. That is when $t^\mu$ asymptotically coincides with the normal to constant time surfaces. This is equivalent to $N=1,N^a=0$ due to \eqref{lapse-shift}. Using this in the above formula leads to the ADM mass
\begin{align}
M&=\dfrac{1}{8\pi} \oint _{S_t}(k-k_0)\sqrt{\sigma}\,d^{n-1}x\,.
\end{align}
On the other hand, the angular momenta are defined when the Hamiltonian flow is an \textit{asymptotic rotation}. That is when $t^\mu\to \phi^\mu\equiv \p x^\mu / \p \varphi$ where $\varphi$ is a rotation angle near the boundary. This is equivalent to $N=0,N^a=\p y^a / \p \varphi$. The corresponding angular momentum is then defined as
\begin{align}
J&= -\dfrac{1}{8\pi} \oint _{S_t}\phi_a r_b(K^{ab}-K h^{ab})\sqrt{\sigma}\,d^{n-1}x\,.
\end{align}
A set of boundary conditions, restricts the form of allowed metrics and accordingly restricts the choices of $N,N^a$ through \eqref{3+1 decomposition}. The set of allowed choices of $N,N^a$ are called the allowed surface deformations. Upon using \eqref{lapse-shift}, this specifies a set of allowed Hamiltonian vector flows $t^\mu$. Corresponding to each $t^\mu$ one can associate a conserved charge $Q_t$ using \eqref{charge-Hamiltonian}. These vectors determine the set of symmetries of the Hamiltonian phase space. Among these, some are associated with zero charge. These produce the \textit{pure gauge} transformations since their generator is a \textit{constraint}. However, other vectors associated with the nontrivial charges are the nontrivial symmetries of the phase space. These are sometimes called \textit{residual gauge transformations} or \textit{global symmetries of the phase space}. The set of allowed vectors $t^\mu$ quotiented by the set of pure gauge transformations, is called the \textit{asymptotic symmetry algebra} and the corresponding finite transformations are called the \textit{asymptotic symmetry group}. We will elaborate more on this in chapters \eqref{chapter-AdS3} and \eqref{chapter NHEG phase space}.

\chapter{Covariant phase space formulation of gauge theories and gravity}\label{chapter-covariant phase space}
\vspace{2cm}
\section{Introduction}
In previous chapter, we discussed the Hamiltonian approach to field theories. Clearly, the construction of Hamiltonian formulation involves an explicit choice of time direction. Therefore in a field theory one needs to perform a decomposition of spacetime into space and time. This breaks the covariant form of general relativity. Therefore it is very tempting to have a covariant version of Hamiltonian mechanics. Dirac stated in his lecture notes: ``From the relativistic point of view we are thus singling out one particular observer and making our whole formalism refer to the time for this observer. That, of course, is not really very pleasant to a relativist, who would like to treat all observers on the same footing. However, it is a feature of the present formalism which I do not see how one can avoid if one wants to keep to the generality of allowing the Lagrangian to be any function of the coordinates and velocities" \cite{dirac1976lecture}. Such a covariant formulation was later developed which is called the \textit{covariant phase space formalism} and is based on the symplectic structure of Hamiltonian mechanics. In this chapter, we will describe this approach in detail, and re-derive the results of previous chapter in the more clear language of covariant phase space.

\section{Symplectic Mechanics}
In this section we will describe the symplectic mechanics of finite dimensional systems like a set of particles established by Hamilton, Liouville, and others (see \eg \cite{arnold1989mathematical,vilasi2001hamiltonian}). The construction, however, can be generalized to the case of field theories as we will discuss in next section.

A \textit{symplectic manifold} (or a \textit{phase space} in physical terminology) is a manifold $\Gamma$ equipped with a symplectic form $\Omega$ with the following properties:
\begin{itemize}
	\item it is  a two form $\Omega_{ab}=-\Omega_{ba}$,
	\item it is closed $\mathrm{d}\Omega=0$,
	\item it is nondegenerate, \ie $\det \Omega_{ab}\neq 0$. Equivalently $\quad\Omega_{ab}X^b=0 \Rightarrow X^a=0$.  
\end{itemize}

We assume that the manifold is covered by the coordinate system $x^a$. A special coordinate system (called the Darboux chart) is the set $(q_i,p^i)$ in which the symplectic structure takes the form $\Omega=\sum_i dp^i\wedge dq_i$. Although it is always possible to bring the symplectic structure over a \textit{finite} dimensional manifold to this form (guaranteed by the Darboux theorem), but here we are going to build a \textit{covariant} formalism in which no special role is played by any choice of coordinate system.

Note the difference between symplectic manifolds and Riemannian manifolds which is instead equipped with a \textit{metric}, \ie a symmetric nondegenerate tensor $g_{ab}$. It turns out that symplectic manifolds are the natural framework for formulating Hamiltonian mechanics.

The nondegeneracy of the symplectic form implies that its inverse $\Omega^{ab}$ exists such that $\Omega^{ab}\Omega_{bc}=\delta^a_c$. The inverse can be used to define the Poisson bracket over $\Gamma$. The Poisson bracket between two scalar functions $f,g$ over the manifold is defined as
\begin{align}\label{Poisson bracket def}
\{f,g\}&\equiv\Omega^{ab}\p_a f \p_b g\,.
\end{align}
It can be checked that this definition satisfies the properties of a Poisson bracket. Specially the closedness of $\Omega$ ensures the Jacobi identity for the bracket. As a simple example we can check that for the dynamics of particles, the symplectic form $\Omega=\sum_i dp^i\wedge dq_i$ leads to the well known Poisson bracket \eqref{Poisson bracket Hamiltonian def}.

Now the dynamics of any observable $f$ is determined through the Poisson bracket, once a function $H(x)$ called the \textit{Hamiltonian} is given, so that 
\begin{align}\label{evolution}
\dfrac{df}{dt}&=\{f,H\}\,.
\end{align}
More precisely the Hamiltonian determines a vector field $T^a$ through
\begin{align}\label{Hamiltonian flow}
T^a&\equiv\Omega^{ab}\p_a H\,.
\end{align}
The vector field $T^a$ generates a congruence $\gamma(t)$ over the phase space such that $T^a=\frac{dx^a}{dt}$ where $t$ is a parameter along the congruence. Then \eqref{evolution} can be rephrased as
\begin{align}
\{f,H\}=\Omega^{ab}\p_a f \p_b g=T^a \p_a f=\dfrac{df}{dt}\,.
\end{align}
Given an initial configuration $ f (\bar{x}^a)$, the observable $f$ evolves as
\begin{align}
f(x)- f (\bar{x})&=\int_{t_0}^{t1} dt\dfrac{df}{dt}=\int_\gamma T^a \p_a f=\int_\gamma dx^a\p_a f \,.
\end{align}
Note that the above result is invariant under reparametrizations $t\to \tau(t)$. 
\subsection{Symplectic symmetries}
The symplectic form $\Omega$ can be used to define the notion of {symplectic symmetries} over the phase space. A vector field $X$ is called a \textit{symplectic symmetry} over the phase space if 
\begin{align}\label{symp-sym-def}
\mathcal {L} _X \Omega =0\,,
\end{align}
where $\cL_X$ is the Lie derivative along $X$\footnote{Note that $X$ is a vector tangent to the phase space and should not be confused with a ``spacetime'' vector. However, Lie derivative, exterior derivative, interior product, etc. are defined independent of ``metric'' and therefore apply in the phase space in the usual manner. }. A very useful identity that we will use throughout this thesis is stated in the following proposition. 
\begin{proposition}\label{prop Cartan identity}
\textbf{The Cartan identity.} For any vector field $X$ and any form $\sigma$, the following identity holds
\begin{align}\label{Cartan identity}
\mathcal{L}_X \sigma&=X\cdot d\sigma +d(X\cdot \sigma)\,,
\end{align}
where the interior product of a vector $X$ and a $p$-form $\omega$ is defined as a $p-1$-form 
\begin{align}
(X\cdot \omega)_{a_2\cdots a_p}&\equiv X^{a_1}\omega_{[a_1a_2\cdots a_p]}\,.
\end{align} 
\end{proposition}
\noindent Using the Cartan identity, we can expand \eqref{symp-sym-def}
\begin{align}
\mathcal{L}_X \Omega &= X\cdot d\Omega +d(X\cdot \Omega)\\
&=d(X\cdot \Omega)=0\,.
\end{align}
The fact that $X\cdot \Omega$ is closed implies according to the \textit{Poincar\'e lemma} that  the one form $X\cdot \Omega$ can be written as an exact form (at least locally), \ie there exists a function $H_X$ such that
\begin{align}
X\cdot\Omega &=d  H_X\,,
\end{align}
or in index notation $\Omega_{ab}X^b=\p_a H_X$. Multiplying this by the inverse $\Omega^{ca}$ implies that
\begin{align}\label{Hamiltonian vector field}
X^a&=\Omega^{ab}\p_b H_X\,.
\end{align}
$H_X$ is indeed the generator of evolution along the symmetry vector field $X$, through the Poisson bracket. To see this take any observable $f$, then
\begin{align}\label{Poisson-charge}
	\nn\{f,H_X\}&=\Omega^{ab}\p_a f\p_b H_X=(\Omega^{ab}\p_b H_X)\p_a f\\
	&=X^a\p_a f =\mathcal{L}_X f\,.
\end{align}
Hence the function (``functional'' in the case of field theory) $H_X$ is called the \textit{generator} of $X$ and its numerical value over a solution to the equations of motion is called the \textit{charge} of $X$ over that solution. 
Note that \eqref{Hamiltonian vector field} is a generalization of \eqref{Hamiltonian flow}. While  Hamiltonian is the generator of evolution in time, any symmetry direction of the symplectic form is produced by a generator. 
\subsection{Algebra of symplectic symmetries}
The set of symplectic symmetries form an algebra through the Lie bracket. Assume $X,Y$ are two symplectic symmetries, that is $\mathcal{L}_X \Omega=0=\mathcal{L}_Y \Omega$. Then the Lie  bracket $[X,Y]\equiv\mathcal{L}_X Y$ is also a symplectic symmetry. The reason is that
\begin{align}
\mathcal{L}_{[X,Y]} \Omega = \big(\mathcal{L}_X \mathcal{L}_Y  - \mathcal{L}_Y \mathcal{L}_X \big)\Omega =0\,.
\end{align}
An interesting result which is of great importance in gravity is described in the following theorem.
\begin{theorem}\label{theorem representation}
	The algebra of generators $H_X$ of symplectic symmetries through the Poisson bracket is the same as the algebra of symeplectic symmetries $X^a$ through the Lie bracket, up to a central extension. 
\end{theorem}
\begin{proof}
We start by computing 
\begin{align}
\nn\p_a H_{[X,Y]}&=\Omega_{ab}[X,Y]^b=\Omega_{ab}\cL_X Y^b\\
\nn&=\cL _X (\Omega_{ab}Y^b)=\cL _X (dH_Y)_a\\
\nn&=(d(X\cdot dH_Y))_a\\
&=\p_a(\Omega^{bc}(\p_bH_X)\, (\p_cH_Y))=\p_a(\{H_X,H_Y\})
\end{align}
where in the first line we have used the definitions. In the second line we have used the fact that $X$ is a symmetry of symplectic form (equation \eqref{symp-sym-def}) and also the definition of $H_Y$ (equation \eqref{Hamiltonian vector field}). In the third line we have used the Cartan identity \eqref{Cartan identity} to expand the Lie derivative and also the fact that $d^2=0$. In the last line we have used the definition of $H_X$ and the definition of the Poisson bracket \eqref{Poisson bracket def}. Therefore we have shown that 
\begin{align}\label{representation of algebra}
\{H_X,H_Y\}&=H_{[X,Y]}+C\,,
\end{align}
such that $dC=0$. This is sufficient to conclude that the Poisson bracket of $C$ with all of the charges is zero and therefore $C$ is a central extension of the algebra. Moreover, it implies that $C$ is constant over any connected patch of the phase space.
\end{proof}
The charge $H_X$ corresponding to a symplectic symmetry $X$ is a \textit{conserved charge} if it is preserved along time, \ie the Poisson bracket of $H_X$ and the Hamiltonian $H$ vanish
\begin{align}
\nn\{H_X,H\}&=\Omega^{ab}\p_a H_X \p_b H\\
&=T^a\p_a H_X =0\,.
\end{align} 
However, the above equation can alternatively be written as 
\begin{align}
\nn\{H_X,H\}&=\Omega^{ab}\p_a H_X \p_b H\\
&=X^a\p_a H =0\,.
\end{align}
Therefore conserved charges have two important aspects. On the one hand, they are \textit{constants of motion}, and on the other hand they determine symmetries of the Hamiltonian.
\section{Covariant phase space formulation of gauge theories}
In this section, we will describe how the symplectic mechanics explained in previous section can be adapted for field theories with local (gauge) degrees of freedom. The equations in the previous section are then applicable here, after the replacements \eqref{particle to field trans}. The construction presented here is developed in \cite{Lee:1990nz,Wald:1993nt,Iyer:1994ys,Wald:1999wa,Barnich:2001jy,Barnich:2007bf,Compere:2007az,Ashtekar:1990gc,Julia:1998ys,Julia:2000er,Julia:2002df}.

The first step is to determine the manifold on which the symplectic geometry is defined. Next, the symplectic structure has to be identified. We will investigate the role of local (gauge) symmetries, and show that generically local symmetries correspond to first class constraints in the phase space. Interestingly there is an exception: for spacetimes having a boundary, there might exist local symmetries that correspond to \textit{second} class constraints in phase space. We will then argue that this can lead to a ``lower dimensional dynamics" in the theory produced by \textit{surface degrees of freedom}. In later chapters, we will show that these surface degrees of freedom play an important role in the microscopic understanding of the entropy of BTZ black hole in 3 dimensions, as well as extremal black holes in higher dimensions.

 \subsection{Symplectic structure}\label{appendix-LW-review}
\textbf{Setup and notations.} The construction of covariant phase space as done in \cite{Lee:1990nz} proceeds by considering the space $\cal F$, of all field configurations satisfying a given initial/boundary conditions. Any field configuration $\Phi(x)$ in the spacetime corresponds to a point in $\cal F$ which we denote simply by $\Phi$. The field configurations do not need to satisfy the field equations, hence the set of \textit{on-shell} field configurations form a subspace denoted by $\bar{\mathcal{F}}$. An infinitesimal field perturbation $\de \Phi (x)$ over a configuration $\Phi(x) $ then corresponds to a \textit{vector} tangent to the phase space at $\Phi$. We denote this vector by $[\de\Phi]^A$, where the index $A$ referes to the components of the vector $[\de\Phi]$ in a chosen coordinate system on the phase space. Moreover, the variation operator $\de$ can be regarded as the ``exterior derivative'' on the phase space manifold once we postulate that the variation takes care of the anti-symmetrization, \ie
\begin{align}\label{d_V def}
\de F[\Phi,\de\Phi]\longleftrightarrow \de_1 F[\Phi,\de_2\Phi]-\de_2 F[\Phi,\de_1\Phi]\,.
\end{align}
Since the fields are assumed to be Bosonic, $\de^2\Phi=\de_1\de_2\Phi-\de_2\de_1\Phi=0$. Hence $\de$ really plays the role of exterior derivative. These are depicted schematically in figure \ref{fig:map}. 
\begin{figure}[!h]
\captionsetup{width=.8\textwidth}
\centering
\includegraphics[width=0.9\linewidth]{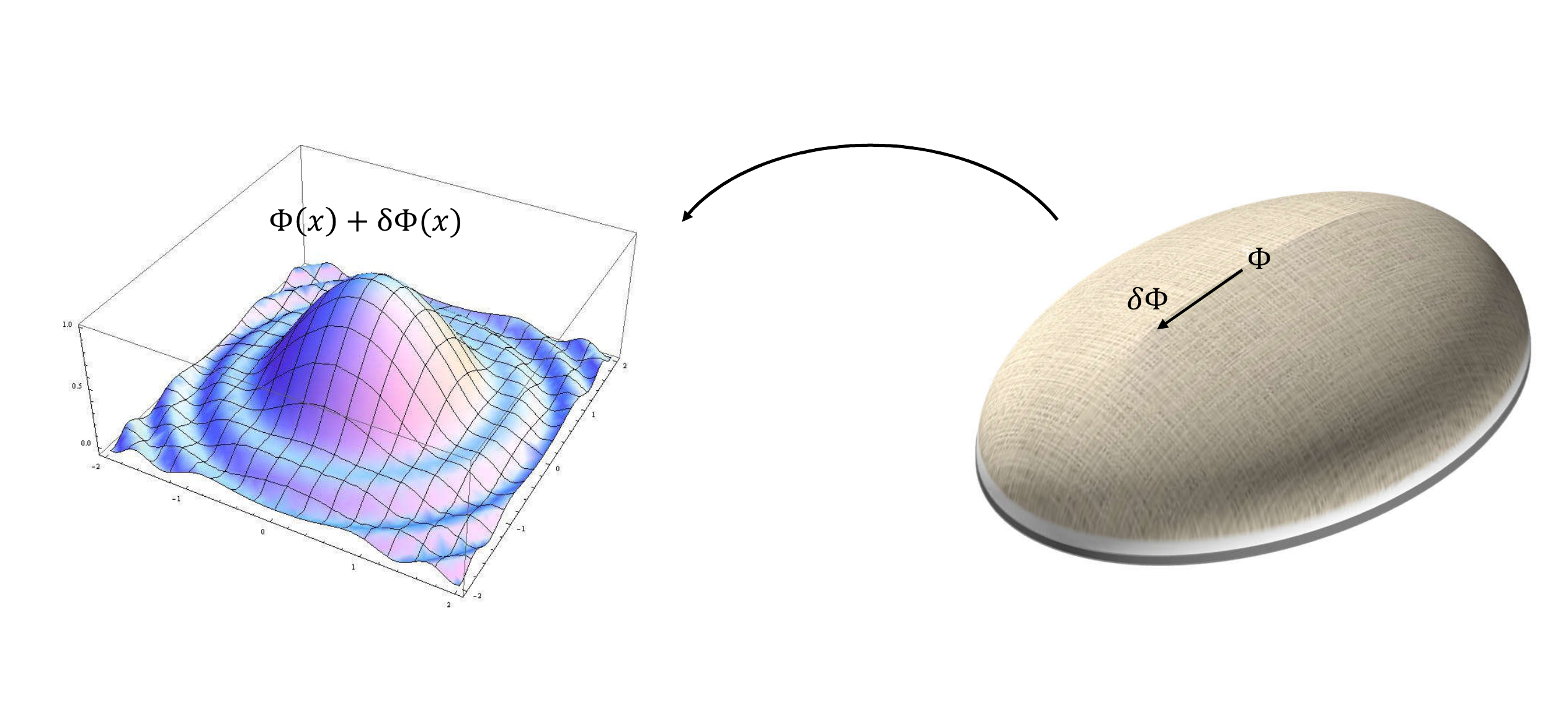}
\caption[Schematic relation between phase space and fields on spacetime]{Schematic relation between phase space and fields on spacetime. The right side is a Schematic depiction of the phase space. Each point on the phase space corresponds to a field configuration on spacetime. A field variation $\de\Phi$ corresponds to a vector tangent to the phase space. However, note that the phase space is an infinite dimensional manifold.}
\label{fig:map}
\end{figure}

\noindent\textbf{The symplectic form.} Now a quantity called the ``presymplectic structure'' is defined over $\cal F$. This structure satisfies the properties of symplectic structure except that it has degeneracy directions. This is because the space $\cal F$ is ``too large'' to serve as a symplectic manifold. Then a reduction over the degeneracy directions, or in other words, taking a symplectic quotient of $(\cal{F}$,$\Omega)$ (see \cite{Woodhouse:1992de}), produces a manifold $ \boldsymbol \Gamma$, on which there exists a consistent symplectic structure $\Omega$. Therefore $( \boldsymbol{\Gamma}, \Omega)$ serves as the suitable phase space of the theory. In the following, we will describe these issues in detail.
 
We assume that a field theory with a set of gauge symmetries is given through a Lagrangian. Let all dynamical fields in the theory be collectively denoted by $\Phi$. The Lagrangian $\mathbf{ {L}}[\Phi]$ (as a top form) is a function of fields and their derivatives up to finite number. Now we define the presymplectic potential $\boldsymbol \Theta [\de \Phi ,\Phi ]$ which is a $d-1$ form, via the  variation of the Lagrangian
 \begin{align}
 \de \mathbf{{L}} [\Phi]= \mathbf{E}[\Phi] \de \Phi +  d \boldsymbol \Theta [\de \Phi, \Phi ]\,.\label{deltaL} 
 \end{align}
 Here $\mathbf{E}_{\Phi}[\Phi] = \frac{\delta \bL}{\delta \Phi}$  are the Euler-Lagrange equations for the fields $\Phi$ and summation on all fields is understood. All fields are assumed to be bosonic (Grassmann-even)  which obey $\delta_1\delta_2\Phi -\delta_2 \delta_1 \Phi = 0$. Hence $\delta$ may be viewed as an exterior derivative operator on the space of field configurations, while $d$ is the exterior derivative operator on the spacetime. The operator $\de$ commutes with the total derivative operator $ d$. Therefore we say that $\boldsymbol \Theta$ is a $(d-1,1)$ form - which means that it is a $d-1$ spacetime form and a phase space 1 form. This notation can be implemented for other quantities to be defined in this chapter.
 
 The general solution of $\bTheta$ in \eqref{deltaL} has the following form:
 \begin{align} \label{Y ambiguity def}
 \boldsymbol \Theta[\delta \Phi,\Phi]  = \boldsymbol \Theta^{ref}[\delta \Phi,\Phi]  + d \mathbf{Y}[\delta \Phi,\Phi] \,,
 \end{align}
 where $\boldsymbol \Theta^{ref}$ is defined by the standard algorithm, which consists in integrating  by parts the variation of the Lagrangian or, more formally, by acting on the Lagrangian with Anderson's homotopy operator $\boldsymbol I^d_{\delta \Phi}$  \cite{Barnich:2001jy,Barnich:2000zw,Barnich:2007bf}, defined for second order theories as
\begin{align}
 \boldsymbol \Theta^{ref}=\boldsymbol{I}^d_{\delta \Phi}  \mathbf{L} \;,\qquad \boldsymbol I^d_{\delta \Phi}  \equiv \left( \delta \Phi \frac{\p}{\Phi_{\, ,\mu}} - \delta \Phi \p_\nu \frac{\p}{\Phi_{\, ,\nu\mu}} \right) \frac{\p}{\p (d x^\mu)}\,.\label{homot}
 \end{align}
In equation \eqref{deltaL} $\mathbf{Y}[\delta \Phi,\Phi] $ is an arbitrary $(d-2,1)$ form. However, it can be fixed by extra requirements depending on the physical problem. This is what we will do in next chapters. 
The presymplectic current $(d-1,2)$ form $ \boldsymbol \omega [\de_1 \Phi, \de_2 \Phi ,\Phi  ]$ is defined as the antisymmetrized variation of the presymplectic potential  \cite{Lee:1990nz}
 \begin{align}\label{LW}
 {\boldsymbol \omega}[\delta_1 \Phi ,\delta_2 \Phi,\Phi ] = \delta_1 {\boldsymbol  \Theta}[ \delta_2 \Phi,\Phi ] - \delta_2 {\boldsymbol \Theta}[\delta_1 \Phi, \Phi ] \,.
 \end{align}
 Under \eqref{Y ambiguity def} we find
 \begin{align}\label{omega with Y}
 {\boldsymbol \omega}[ \delta_1 \Phi ,\delta_2 \Phi, \Phi ] = {\boldsymbol \omega}^{ref} [ \delta_1 \Phi ,\delta_2 \Phi,\Phi ] +  d \left( \delta_1 {\mathbf{Y}}[ \delta_2 \Phi,\Phi] - \delta_2 {\mathbf{Y}}[\delta_1 \Phi,\Phi ]\right).
 \end{align}
The presymplectic current has the property that if $\Phi$ is a solution to the field equations, and $\de_1\Phi,\de_2\Phi$ are solutions to the linearized field equations around $\Phi$, then $d\bomega[\de_1\Phi,\de_2\Phi,\Phi]\approx 0$. This can be checked easily
\begin{align}\label{d omega}
d\bomega[\de_1\Phi,\de_2\Phi,\Phi]&=\delta_1 d{\boldsymbol  \Theta}[ \delta_2 \Phi,\Phi ] - (1\leftrightarrow 2)\nnr
&=\delta_1(\de_2\bL[\Phi]-\bE[\Phi] \de_2\Phi)- (1\leftrightarrow 2)\nnr
&=\delta_1\bE[\Phi] \de_2\Phi-\delta_2\bE[\Phi] \de_1\Phi\approx 0\,.
\end{align}
 The presymplectic form $\Omega_{AB}$ contracted with two  vectors $[\de_1 \Phi],[\de_2 \Phi]$ tangent to the phase space is defined as
 \begin{align}\label{symplectic form def}
 \Omega_{AB}\,[\de_1\Phi]^A [\de_2 \Phi]^B=\int_\Sigma \boldsymbol{\omega} [\de_1 \Phi,\de_2 \Phi,\Phi]
 \end{align}
 where the integral is defined over a spacelike surface $\Sigma$. The definition of presymplectic form a priori depends on $\Sigma$. However, let $\Sigma_2$ be obtained by a continuous deformation of $\Sigma$ when its boundaries are fixed. Then by making use of the Stokes' theorem, the difference is given by an integral of $d\boldsymbol{\omega}$  in the spacetime region between the two hypersurfaces, which is vanishing on shell according to \eqref{d omega}. Also we always assume that there is no symplectic flux at the boundary, \ie $\int_B \boldsymbol{\omega} [\de_1 \Phi,\de_2 \Phi,\Phi]=0$ over any region of the boundary. Hence we infer by the same argument that the presymplectic form (and accordingly the symplectic form) is the same for any hypersurface $\Sigma$. This result is necessary for the ``covariance'' of phase space construction.
 
\begin{figure}[!h]
\captionsetup{width=.8\textwidth}
\centering
\includegraphics[width=1\linewidth]{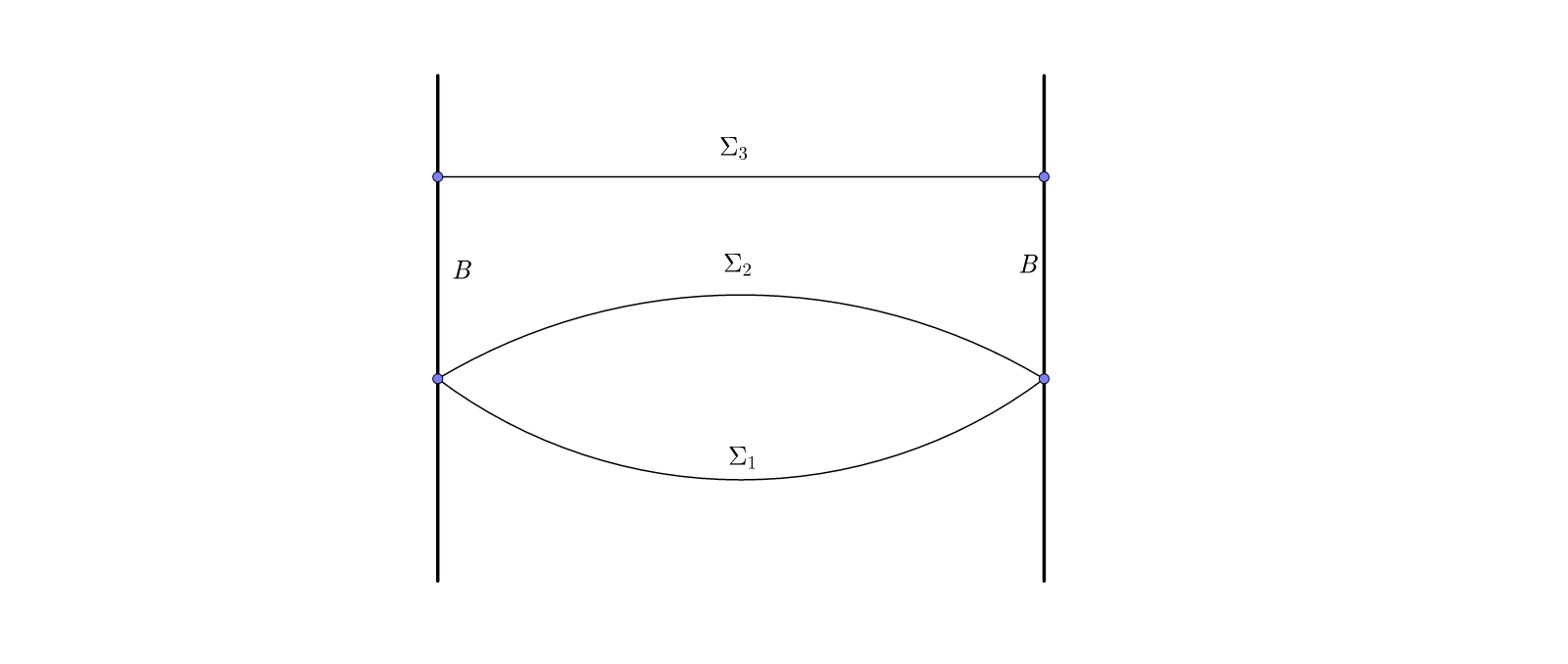}
\caption[Different integration surfaces of the symplectic form]{Different hypersurfaces lead to the same result for the presymplectic form if the symplectic flux at the boundary $B$ vanishes.}
\label{fig:hypersurfaces1}
\end{figure}
 
\section{Local symmetries and their generators}\label{section-conserved charges}
The covariant phase space formulation can be used to investigate the relation between the local symmetries in the space of field configurations, and the constraints on the phase space. This is the analogue of what we saw in sections \ref{sec Hamiltonian field theory} and \ref{sec conserved charges-Hamiltonian} but in the covariant formulation. The result is that local symmetries corresponding to vanishing charges correspond to constraints in the phase space, and hence they are unphysical redundant gauge degrees of freedom in the theory. However, more interestingly, there are \textit{residual} gauge transformations, that lead to nontrivial charges in the phase space and therefore they lead to inequivalent physical degrees of freedom in the theory. As we will see these can be interpreted as surface degrees of freedom which play an important role in the holographic description of gravity. This fact will be stressed specially in chapters \eqref{chapter-AdS3} and \eqref{chapter NHEG phase space}. 

In the following we will denote a gauge transformation in the fields by $\de_\chi\Phi$ where $\chi$ is the parameter of the gauge transformation. In case of Electrodynamics the gauge symmetry is $\de_\Lambda A= d\Lambda$ where $\Lambda$ is an arbitrary scalar function. In General Relativity, the gauge symmetry is $\de_\chi g_\mn\equiv \cL_\chi g_\mn=\nabla_{\mu}\chi_{\nu}+\nabla_{\nu}\chi_{\mu}$ where $\chi^\mu$ is an arbitrary vector field.

\subsection{Construction of generators}
Let $\delta_\chi \Phi$ denote an infinitesimal gauge transformation of the fields. We are interested to find the generator of this gauge transformation, \ie a function $H_\chi$ over the phase space satisfying the following relation
\begin{align}
	\nn\{\Phi,H_\chi\}&=\de_\chi \Phi\,.
\end{align}
The solution to the above equation is given in the following proposition
\begin{proposition}
The variation of the ``generator'' of a gauge transformation $\delta_\chi \Phi$ is given by 
\begin{align}\label{delta gen cov def}
\de H_\chi&=\int_\Sigma \boldsymbol{\omega} [\de \Phi,\de_\chi \Phi,\Phi]\,.
\end{align}
The generators $H_\chi$ are then obtained by an integration in the phase space
\begin{align}
H_\chi&=\int_{\bar{\Phi}}^{\Phi}\de H_\chi\,.
\end{align}
The on-shell value of the generator is  the ``charge'' of that gauge transformation.
\end{proposition}
\begin{proof}
Using the abstract notation \eqref{symplectic form def} we can rewrite \eqref{delta gen cov def} as
\begin{align}
\de H_\chi&= [\de\Phi]^A \dfrac{\de\,H_\chi}{[\de \Phi]^A}=\Omega_{AB}\,[\de\Phi]^A [\de_\chi \Phi]^B
\end{align}
which leads to 
\begin{align}
\dfrac{\de\,H_\chi}{[\de \Phi]^A}=\Omega_{AB}\, [\de_\chi \Phi]^B\,.
\end{align}
Using the above result in the definition of Poisson bracket, reveals that $H_\chi$ defined in \eqref{delta gen cov def} indeed generates the desired gauge transformation:
\begin{align}
\{F[\Phi],H_\chi\}&=\Omega^{AB}\dfrac{\de\,F}{[\de \Phi]^A}\dfrac{\de\,H_\chi}{[\de \Phi]^B}\nnr
&=\Omega^{AB}\Omega_{BC}\, [\de_\chi \Phi]^C\dfrac{\de\,F}{[\de \Phi]^A}\nnr
&=[\de_\chi \Phi]^A\dfrac{\de\,F}{[\de \Phi]^A}=\de_\chi F\,.
\end{align}
\end{proof}
As an example, let us find the generator of gauge symmetries of the Maxwell theory using the covariant phase space method. The Lagrangian density is given by $\cL =-\dfrac{1}{4}F_{\mn}F^{\mn}$. The symplectic potential is $\bTheta_{\mu_2\cdots \mu_n}[\de A,A]=\theta^\mu \eps_{\mu_1\cdots\mu_n}$ where $\eps_{\mu_1\cdots\mu_n}$ is the Levi Civita tensor and 
\begin{align}
\theta^\mu[\de A,A]=F^\mn \de A_\nu\,,
\end{align}
Accordingly, we find the presymplectic structure 
\begin{align}
\Omega[A,\de_1 A,\de_2 A]&=\int_\Sigma d\Sigma_\mu \Big(\de_1F^\mn\de_2A_\nu-\de_2F^\mn\de_1A_\nu\Big)\,.
\end{align}
Due to \eqref{delta gen cov def}, the generator of a gauge transformation $\de_\eps A_\mu=\p_\mu \eps(x)$  is
\begin{align}
\de H_\eps&=\int_\Sigma d\Sigma_\mu \Big(\de F^\mn\de_\eps A_\nu-\de_\eps F^\mn\de A_\nu\Big)\nnr
&=\int_\Sigma d\Sigma_\mu \,\de F^\mn\,\p_\nu \eps\,.
\end{align}
The last term in the right hand side of first line drops since $F_\mn$ is gauge invariant. Now we integrate by part to obtain
\begin{align}
\de H_\eps&=\int_\Sigma d\Sigma_\mu\, \eps \,\de\,\p_\nu F^\mn + \oint_{\p\Sigma}d\Sigma_\mn \,\eps \,\de F^\mn\,.
\end{align}
Assuming that the parameter $\eps$ of gauge transformation is field independent, we observe that the generator $H_\eps$ can be extracted directly
\begin{align}
H_\eps&=\int_\Sigma d\Sigma_\mu \,\eps \,\p_\nu F^\mn + \oint_{\p\Sigma}d\Sigma_\mn \,\eps \, F^\mn\,.
\end{align}
We observe that the generator of gauge transformation is composed of a bulk term which is a combination of equations of motion, and a surface contribution. The charge defined as the on-shell value of the generator, is hence
\begin{align}
Q_\eps&=\oint_{\p\Sigma}d\Sigma_\mn \,\eps \, F^\mn\,.
\end{align} 
In case where $\eps\to 1$ near the boundary, we find $Q=\oint_{\p\Sigma}d\Sigma_\mn  \, F^\mn$ which is exactly the electric charge of the system. Taking the surface $\Sigma$ to be the $t=const$ surface, we find $H_\eps=\int d^{n-1}x \,\eps \,\p_i E^i+ Q_\eps$ consistent with the result in Hamiltonian formulation \eqref{Maxwell generator} \footnote{As we are concerned in gauge transformation of \textit{dynamical} fields, the first term in \eqref{Maxwell generator} is irrelevant, since it generates only the variation of the non-dynamical field $A_0$, explicitly $\de_\Lambda A_0=\p_t\Lambda$.}  
\section{Explicit form of generators and charges in gravity}
In this section, we provide an explicit and covariant expression for the generators of gauge symmetries and related charges. We will show that the generators are given by two pieces: a bulk integral which is a combination of constraints, and an integral over a compact codimension 2 surface in the boundary of spacetime. After that, we discuss the three conditions that should be met so that a conserved charge could be defined over the phase space: the \textit{integrability condition} in section \ref{sec integrability}, and conservation in time as well as finiteness in \ref{sec conservation}.  
 
As we discussed before, the Lagrangian is invariant under a gauge transformation up to a total derivative
\begin{align}\label{delta chi L}
	\de_\chi \bL&=d\bM_\chi[\Phi]\,.
\end{align}
For gravitational theories, the gauge symmetry is a local coordinate transformation, generated infinitesimally by an arbitrary vector field $\chi$. The transformation in fields and accordingly the Lagrangian is given by a Lie derivative along $\chi$
\begin{align}
	\de_\chi \bL&=\cL_\chi \bL=\chi\cdot d\,\bL +d(\chi\cdot \bL)=d(\chi\cdot \bL)
\end{align}
by making use of the Cartan identity (see proposition \eqref{prop Cartan identity}), and the fact that $\bL$ is a top form and hence $d\bL=0$. Therefore for gravity $\bM_\chi=\chi\cdot \bL$. For Maxwell theory $\bM_\eps=0$ since the Lagrangian is invariant under gauge transformations $\de_\eps A=d\eps$. 

 The \textit{Noether current} for a gauge transformation parameterized by $\chi$ is defined as \cite{Iyer:1994ys}
 \begin{align}\label{Noether-current}
 \mathbf{J}_\chi&=\bTheta[\de_\chi\Phi,\Phi ]-\bM_\chi[\Phi]
 \end{align}
 The exterior derivative of $\mathbf{J}_\chi$ is 
 \begin{align}\label{d J}
 d	\mathbf{J}_\chi&= d\bTheta[\de_\chi\Phi,\Phi ]-d\bM_\chi[\Phi]\nnr
 &=(\de_\chi \bL-\bE[\Phi]\, \de_\chi\Phi)-\de_\chi \bL=-\bE[\Phi]\, \de_\chi\Phi
 \end{align}
 by using \eqref{deltaL} and \eqref{delta chi L}. Therefore we see that the Noether current is closed once the field equations are satisfied, \ie $d \mathbf{J}_\chi\approx 0$. Note that this is only an on-shell equality, but what we need is a stronger result that we obtain by using \textit{Noether's second theorem}.
 
 Let us focus on the quantity $\bE[\Phi]\, \de_\chi\Phi$. If one tries to remove all derivatives on $\chi$ by integrating by parts, one would obtain
 \begin{align}
 	\bE[\Phi]\, \de_\chi\Phi&=\chi \cdot \boldsymbol{N}(\bE[\Phi],\Phi)+d\boldsymbol{S}_\chi(\bE[\Phi],\Phi)\,.
 \end{align}
 Noether's second theorem states that $\boldsymbol{N}(E[\Phi],\Phi)=0$ strongly (\ie without use of equations of motion) \cite{Noether:1918zz,Avery:2015rga,Compere:2009dp}. Therefore using the Noether's second theorem in equation \eqref{d J} we find that the combination $\mathbf{J}_\chi+\boldsymbol{S}_\chi$ is closed \textit{off-shell} and hence exact by the Poinca'e lemma
 \begin{align}
 	\mathbf{J}_\chi&=-\boldsymbol{S}_\chi(\bE[\Phi]) + d\bQ_\chi\,.
 \end{align}
The $d-2$ form $\mathbf{Q}_\chi$ is the \textit{Noether charge} density associated with $\chi$. 
 The fundamental identity of the covariant phase space formulation of gravity is the following \cite{Wald:1993nt,Barnich:2001jy,Wald:1999wa,Compere:2009dp}.
 \begin{theorem}\label{theorem-charge}
 	
 	The presymplectic current contracted with a gauge transformation $\de_\chi\Phi$, is equal to a bulk term proportional to the field equations and a boundary term
 	\begin{align}\label{omega off shell}
 	\bomega [\delta \Phi , \delta_\chi \Phi,\Phi ] = \boldsymbol{G}_\chi(\bE[\Phi],\de\bE[\Phi],\Phi)+  d \,{\boldsymbol k}_\chi [\delta\Phi,\Phi]\,.
 	\end{align}
 	The bulk term given by 
 	\begin{align}\label{G chi def}
 	\boldsymbol{G}_\chi(\bE[\Phi],\de\bE[\Phi],\Phi)&=\de \boldsymbol{S}_\chi(\bE[\Phi],\Phi)+	\chi\cdot (\bE[\Phi]\, \de\Phi)
 	\end{align}
 	vanishes provided that the fields $\Phi$ satisfy the equations of motion and the field variations $\delta\Phi$ satisfy the linearized equations of motion around $\Phi$. The boundary $d-2$ form ${\boldsymbol k}_\chi [\delta\Phi,\Phi]$ is given by 
 	\begin{align}\label{charge variation identity}
 	{\boldsymbol k}_\chi [ \delta\Phi,\Phi ] = \delta \mathbf{Q}_\chi[\Phi ] - \chi \cdot \bm{\Theta} [\delta\Phi,\Phi ]+ d(\cdot) .
 	\end{align}
 where $d(\cdot)$ refers to possible boundary terms which cancel upon integration over a closed surface. 
\end{theorem} 	

\begin{proof}
We start by taking a variation of \eqref{Noether-current} and assume that $\chi$ is field independent, \ie does not depend on dynamical fields $\Phi$. The case where the vector is field dependent is considered in appendix \ref{sec field dependent vec charges}. 
\begin{align}\label{deltaJ}
\de \mathbf{J}_\chi&=\de \bTheta[\de_\chi\Phi,\Phi ]-\chi\cdot \de \mathbf{L}\,.
\end{align}
Now using \eqref{deltaL} and assuming that the field equations are satisfied, we have $\de \mathbf{L}\approx \mathrm{d} \bTheta[\de\Phi,\Phi ]$. Then using the Cartan identity we have 
\begin{align}
\chi\cdot \de \mathbf{L}&=\chi\cdot \mathrm{d} \bTheta[\de\Phi,\Phi ]+\chi\cdot (\bE[\Phi]\, \de\Phi)\\
& = \mathcal{L}_\chi \bTheta[\de\Phi,\Phi ]-\dd\,(\chi\cdot \bTheta[\de\Phi,\Phi ])+\chi\cdot (\bE[\Phi]\, \de\Phi)\,.
\end{align}
Replacing this in \eqref{deltaJ} and using $\de \mathbf{J}_\chi= -\de \boldsymbol{S}_\chi+\dd(\de \bQ_\chi) $ on the left hand side, we obtain 
\begin{align}
-\de \boldsymbol{S}_\chi+\dd\,(\de \bQ_\chi)&=\Big(\de \bTheta[\de_\chi\Phi,\Phi ]- \mathcal{L}_\chi (\bTheta[\de\Phi,\Phi ]) \Big)+\dd\,(\chi\cdot \bTheta[\de\Phi,\Phi ])-\chi\cdot (\bE[\Phi]\, \de\Phi)\,.
\end{align}
The term in parantheses on the right hand side is indeed the symplectic current $\bomega(\de\Phi,\de_\chi\Phi,\Phi)$. Therefore
\begin{align}\label{omega-dk}
\bomega(\de\Phi,\de_\chi\Phi,\Phi)&=\Big(-\de \boldsymbol{S}_\chi+\chi\cdot (\bE[\Phi]\, \de\Phi)\Big)+\dd\,\Big(\de \bQ_\chi-\chi\cdot \bTheta[\de\Phi,\Phi ]\Big)\,.
\end{align}
\end{proof}

Now we can read off the generator of gauge transformations and their corresponding charges by combining equation \eqref{delta gen cov def} and the above theorem
\begin{align}
\de H_\chi&=\int_\Sigma \boldsymbol{\omega} [\de \Phi,\de_\chi \Phi,\Phi]= \int_{\Sigma} \boldsymbol{G}_\chi(\bE[\Phi],\de\bE[\Phi],\Phi)+  \oint_{\p\Sigma}{\boldsymbol k}_\chi [\delta\Phi,\Phi]\,.
\end{align}
Since we are interested in the generator of a gauge transformation over a solution, we may use $\bE[\Phi]=0$ to simplify the result. Note however, that for the generator we cannot use the linearized field equations since the Poisson bracket $\{f,g\}$ defined in \eqref{Poisson bracket def}  involves an arbitrary variation in its arguments not only the variations tangent to the solution space. Hence  we can drop the second term in \eqref{G chi def} and the result is 
\begin{align}\label{generator var def}
\de H_\chi&\approx -\int_{\Sigma} \de \boldsymbol{S}_\chi (\bE[\Phi])+  \oint_{\p\Sigma}{\boldsymbol k}_\chi [\delta\Phi,\Phi]\,.
\end{align}
However, as we discussed before, the charge of $\chi$ is defined as the on-shell value of its generator, through an integration of the above equation on a path in solution space from a reference field configuration $\bar{\Phi}$ to an arbitrary field configuration $\Phi$. Since the path is tangent to the solution space, the bulk term in \eqref{generator var def} drops and we find the covariant expression of the charge associated with the gauge transformations $x\to x+\chi$ in a diffeomorphism invariant theory
\begin{align}\label{delta charge}
\de Q_\chi&\equiv\oint_{\p\Sigma}{\boldsymbol k}_\chi [\delta\Phi,\Phi]\,.
\end{align}
\noindent For later use, we note that integrating \eqref{omega-dk} over a hypersurface $\Sigma$ bounded by two surfaces $S_1,S_2$ yields a relation between charges defined at $S_1,S_2$:
\begin{align}\label{charge difference}
\de (Q_\chi\Big\vert_{S_2}-Q_\chi\Big\vert_{S_1})&\approx \int_{\Sigma} \bomega (\de\Phi,\de_\chi\Phi,\Phi)
\end{align}
when $\Phi$ solves the field equations and $\de\Phi$ solves the linearized field equations.
\subsection{Explicit charges for Einstein gravity}

 For Einstein theory which will be the context of next chapters 
 \be
 \mathbf{L}_{Einstein} =\frac{1}{16 \pi G} R \boldsymbol{\eps} , 
 \ee
 and
 \begin{align}
 (16\pi G) (\star\bTheta^{ref})^{\,\mu} &= \nabla_\nu h^{\nu\mu} - \nabla^\mu h ,\qquad
 (16\pi G)(\star \bQ)^{\mu\nu}_\chi   = \nabla^\nu \chi^\mu - \nabla^\mu \chi^\nu , 
 \end{align}	
 { where we denoted $h_{\mu\nu} \equiv \delta g_{\mu\nu}$, $h^{\mu\nu}=g^{\mu\alpha}h_{\alpha\beta}g^{\beta \nu}$, $h=g^{\mu\nu}h_{\mu\nu}$. } Therefore
 \begin{align}\label{kgrav}
 \hspace*{-8mm}{\boldsymbol k}_\chi^{ref}&\equiv\delta \mathbf{Q}_\chi[\Phi ] - \chi \cdot \bm{\Theta}^{ref} [\delta\Phi,\Phi ]\cr
 &=\dfrac{(d^{d-2}x)_{\mu \nu}}{8 \pi G} \left(\chi^\nu\nabla^\mu h
 -\chi^\nu\nabla_\sigma h^{\mu\sigma}
 +\chi_\sigma\nabla^{\nu}h^{\mu\sigma}
 +\frac{1}{2}h\nabla^{\nu} \chi^{\mu}
 -h^{\rho\nu}\nabla_\rho\chi^{\mu}\right). 
 \end{align}
 Considering the effect of the ambiguity $\bY$ in \eqref{Y ambiguity def}, the surface charge is explicitly given by
\begin{align}
 {\boldsymbol k}_\chi [ \delta\Phi,\Phi ] &= {\boldsymbol k}_\chi^{ref} +\Big(\delta \mathbf{Y}[\delta_\chi \Phi,\Phi ] - \delta_\chi \mathbf{Y} [\delta \Phi,\Phi ] \Big). \label{sc}
\end{align}
Also the bulk term of the generator \eqref{generator var def} is  the variation of \cite{Compere:2009dp}
 \begin{align}
\int_\Sigma \boldsymbol{S}(\bE[\Phi])&=2\int_\Sigma d\Sigma_\mu G^\mu_{\;\nu} \,\xi^\nu \,,
 \end{align}
which is a combination of constraint equations of Einstein gravity $\cH_\mu\equiv 2 n_\nu G^\nu_{\;\mu}$ reproducing exactly the results in \eqref{generator Hamiltonian}.
\subsection{Spatially compact manifolds} 
The hypersurface $\Sigma$ over which the presymplectic form was defined in \eqref{symplectic form def} can be either compact or noncompact. In the former case, its boundary  $\p\Sigma$ is vanishing and accordingly the charges \eqref{delta charge} defined on $\p\Sigma$ vanish
\begin{align}
\de Q_\chi&\approx 0\,.
\end{align}
We can rewrite the above result as 
\begin{align}
\de Q_\chi&= \mu_A T^A=0\,,
\end{align}
where $\mu_A={\Omega}_{AB}[\de_\xi\Phi]^B$ and $T^A$ any vector field tangent to the solution submanifold $\bar{\Gamma}$. Note that the symplectic form is nondegenerate and hence the one form $\mu_A$ is nonvanishing. Therefore the relation $\mu_A T^A=0$ implies that $\mu_A$ is ``normal'' to any vector tangent to  $\bar{\Gamma}$ in $\Gamma$. This means that the codimension of $\bar \Gamma$ in $\Gamma$ is at least one. Each local symmetry $\de_\xi\Phi$ increases the codimension of $\bar \Gamma$ in $\Gamma$ by 1. Since $\Gamma$ represents the set of all ``kinematically possible" whereas $\bar\Gamma$ represents the set of all ``dynamically possible" states, each such increase in the codimension of $\bar{\Gamma}$ corresponds to a constraint in the system. Moreover, all these constraints are first class, since their algebra is closed.

\subsection{Spacetimes with a boundary}
The discussions of previous section is valid for spacetimes in which $\Sigma$ is either compart (having no boundary), or otherwise the boundary conditions are such that no spatial boundary terms arise from applying Stokes' theorem. In this case, we showed that all gauge symmetries correspond to constraints in the phase space. However, the situation differs when the spacetime has a boundary with suitable boundary conditions such that nontrivial boundary terms appear due to Stokes' theorem. In this case, the set of gauge symmetries break into three classes
\begin{itemize}
	\item gauge transformaitons which are not allowed by the boundary conditions,
	\item allowed gauge transformaitons corresponding to vanishing charges,
	\item allowed transformations corresponding to nontrivial charges.
\end{itemize}
A gauge transformation $\eta$ of the second type corresponds to a first class constraint (as we discussed in previous section) and is usually called a \textit{trivial} gauge symmetry. The third class however corresponds to a second class constraint in the bulk and a conserved charge at the boundary.
In this case the one form $\mu_A=\Omega_{AB}\de_\chi\Phi^B$ is not a constraint, but its value on the constraint surface $\bar{\Gamma}$ is the variation of charge corresponding to $\chi$ according to \eqref{delta gen cov def} and \eqref{symplectic form def}, \ie
\begin{align}\label{charge-omega}
\Omega_{AB}[\de_\chi\Phi]^B&\approx [\de Q_\chi]_A
\end{align}
provided that an integrability condition holds which will be discussed in the following.
\subsection{Integrability of charges}\label{sec integrability}
It is not obvious that the infinitesimal charge defined in \eqref{delta charge} is really a variation of a function $Q_\chi$ over the phase space. This is similar to the usual thermodynamics in which for example heat transfer $\centernot\de Q$ is not an exact variation and the total amount of heat transfer depends on the path traveled by the system. On ther other hand $\Delta S=\int_{\gamma}\frac{\de Q}{T}$ defines entropy which is independent of the path, and therefore a function $S$ can be defined over the phase space. We say that $\de S$ is integrable while $\de Q$ is not. The necessary and sufficient condition for a charge variation to be integrable is that 
\begin{align}\label{Integrability-general}
\de_1\de_2 Q_\chi -\de_2\de_1 Q_\chi&=0
\end{align}
for any two variations $\de_1,\de_2$. Therefore the integrability condition is 
	\bea
	\mathcal I [\delta_1\Phi,\delta_2\Phi,\Phi] \equiv \delta_1 \oint  {\boldsymbol k}_{\chi}[\delta_2 \Phi ; \Phi ] - (1 \leftrightarrow 2) = 0\label{integrability Wald}
	\eea
	for arbitrary variations $\delta_1\Phi, \delta_2\Phi$ tangent to the phase space at any arbitrary point in the phase space $\Phi$. If the integrability condition holds,  then $\oint  {\boldsymbol k}_{\chi}[\delta \Phi ; \Phi ]$ is an exact variation. In other words, there exist a function $Q_\chi$ on phase space satisfying 
	\begin{align}
	\de Q_\chi&=\oint  {\boldsymbol k}_{\chi}[\delta \Phi ; \Phi ].
	\end{align} 
	To compute $Q_\chi$, one can choose any path $\gamma$ in the phase space between a reference configuration $\bar\Phi$ (which can be the background field configuration) and the field of interest $\Phi$ and define the canonical  charge as
	\bea
	Q_\chi [\Phi,\bar\Phi]= \int_\gamma \de Q_\chi  + N_\chi[\bar\Phi]\,.
	\eea
	Note that the integral $\int_{\gamma}$ is an integral over a one parameter family of field configurations in the phase space. Here, $N_\chi[\bar\Phi]$ is the freely chosen charge of the reference configuration $\bar\Phi$.\footnote{In the covariant phase space formalism, this reference charge is arbitrary. If a holographic renormalization scheme exists, one would be able to define this reference charge from the first principles, as it is done e.g. in asymptotically AdS spacetimes.} 
	Using \eqref{charge variation identity} and the fact that $\de \mathbf{Q}_\chi$ is an exact variation, we find the simple integrability condition \cite{Wald:1999wa},
	\begin{align}\label{LW-integrability}
		\mathcal{I} [\delta_1\Phi,\delta_2\Phi,\Phi]&\equiv  - \oint  \chi \cdot \bomega[\delta_1 \Phi, \delta_2 \Phi,\Phi] = 0,
	\end{align}
	for arbitrary variations $\delta_1\Phi, \delta_2\Phi$ and for the $\chi$ of interest.
\subsection{Integrability and symplectic symmetries} 
Not all nontrivial gauge transformations correspond to integrable charges. There is a subalgebra of nontrivial gauge transformations that are integrable which form the algebra of symplectic symmetries of the phase space. This point is less noticed in the literature of asymptotic symmetries. This is what we will show in this subsection.

\begin{proposition}
A local symmetry associated with integrable charge, corresponds to a symplectic symmetry of the phase space.
\end{proposition}
\begin{proof}
As we discussed, an infinitesimal charges is integrable if and only if \eqref{Integrability-general} hold. On the other hand the infinitesimal charge is given by \eqref{charge-omega}. Here we use a notation that treats $\de$ as the exterior derivative operator in phase space  which takes care of antisymmetrization (see \eqref{d_V def}). Upon using \eqref{delta gen cov def} in \eqref{Integrability-general}, the integrability condition becomes
\begin{align}
\de(\de Q_\chi)&=\de(\de_\chi\Phi\cdot \Omega)=0\,.
\end{align}
Note that $\de_\chi\Phi$ is a vector tangent to the phase space and $\Omega$ is a two form. Now using the Cartan identity \eqref{Cartan identity}, we have
\begin{align}
\mathcal{L}_{\de_\chi\Phi}\Omega - \de_\chi\Phi\cdot \de\Omega =0\,.
\end{align}
However, since by \eqref{symplectic form def}, $\Omega=\de\Theta$, therefore $\de\Omega=\de^2\Theta=0$ and accordingly the integrability condition corresponds to 
\begin{align}
\mathcal{L}_{\de_\chi\Phi}\Omega&=0\,.
\end{align}
This is nothing but the definition of a symplectic symmetry \cf \eqref{symp-sym-def}. This proves the proposition. Also since symplectic symmetries form a closed algebra, therefore the set of integrable charges form a closed algebra under the Poisson bracket. Equivalently, if \eqref{integrability Wald} holds for $\chi_1,\chi_2$, then it holds for $[\chi_1,\chi_2]$ as well.
\end{proof}
\subsection{Conservation and Finiteness}\label{sec conservation}
The charges \eqref{delta charge} are defined over the boundary of a spacelike hypersurface $\Sigma$. For example $\Sigma$ can be considered as the surfaces of constant time in some coordinate system. Then the charge will be \textit{conserved} if its value does not depend on the chosen hypersurface $\Sigma$. We will obtain a necessary and sufficient condition for the conservation of charges.

Take the hypersurfaces $\Sigma_1, \Sigma_3$ in figure \eqref{fig:hypersurfaces1}. According to the Stokes' theorem and the fact that $d\bomega\approx 0$ we find
\begin{align}
\de (H_\chi\Big\vert_{\Sigma_3}-H_\chi\Big\vert_{\Sigma_1})&\approx\int_{\mathcal{B}} \bomega (\de\Phi,\de_\chi\Phi,\Phi)\,.
\end{align}
Therefore the conservation of charge is fulfilled if the flux of the symplectic current through the boundary $\cal B$ is vanishing. 

Moreover, the charge $Q_\chi$ associated with each nontrivial gauge transformation should be finite otherwise the boundary conditions is considered as inconsistent.

 \section{Discussion}
In this chapter, we discussed the covariant phase space formulation of gauge theories. We showed how the symplectic form can be obtained from the second variation of the Lagrangian. Then we discussed the local symmetries and their generators in phase space. We obtained a covariant form for the generators and charges in  diffeomorphism invariant theories like General Relativity. 

\subsection*{Surface degrees of freedom}
The phase space of a gauge theory over a spacetime with boundaries, may contain field configurations related by \textit{residual} gauge transformations, \ie those associated with nontrivial charges. Using the bulk observables one cannot distinguish between these states, since by construction, the theory is invariant under gauge transformations. However, as we showed, the corresponding surface charges can indeed distinguish between such states. Moreover, charges are the only observables that can do this separation. In this sense, these states are called \textit{surface degrees of freedom} or \textit{boundary gravitons} in the context of gravity. Surface degrees of freedom play an important role in  microscopic description of black hole entropy as well as holography.
 
Here we discussed surface degrees of freedom in the context of Hamiltonian formulation. Interestingly, surface degrees of freedom also appear in other approaches like the path integral formulation of gauge theories \cite{Esposito:1995zf}. Also they have appeared in the context of loop quantum gravity \cite{Baez:1995jn,Smolin:1995vq}. 
 Also in the non-gravitational context it is  well known that Chern-Simons theory on a manifold with boundary induces a  dynamical Wess-Zumino-Witten (WZW) theory on the boundary \cite{Witten:1988hf,Elitzur:1989nr,Ogura:1989gn,Carlip:1991zm,Balachandran:1991dw,Balachandran:1992yh}, whose degrees of freedom correspond to surface degrees of freedom as we described. Since gravity in 3 dimensions has a Chern-Simons description, which does not possess any local dynamics, the only dynamics is related to surface gravitons. The proposal followed by Carlip , and also Balachandran \etal is that these boundary gravitons are the origin of BTZ black hole entropy. \cite{Carlip:1994gy,Carlip:1995qv,Carlip:1998wz,Carlip:1999cy,Carlip:2005zn,Balachandran:1994up,}. 
Interestingly similar phenomenon happens in Quantum Hall effect in condensed matter physics \cite{Balachandran:1994ik}.
\subsection*{Spacetime with disconnected boundaries}
Regarding the covariant phase space method described above, it seems that there is still a shortage in the construction. We assumed that the spacetime has only one connected boundary in the asymptotic region. However, many interesting examples involve spacetimes with \textit{two} disconnected boundaries. Black holes are a good examples. A Cauchy surface in black hole geometry has two disjoint asymptotics. Even if we restrict the study to the region outside black hole (region (I) in the Penrose diagram below), there will be still an internal boundary  on the ``horizon''. 

\begin{figure}[!h]
\captionsetup{width=.8\textwidth}
\centering
\includegraphics[width=0.7\linewidth]{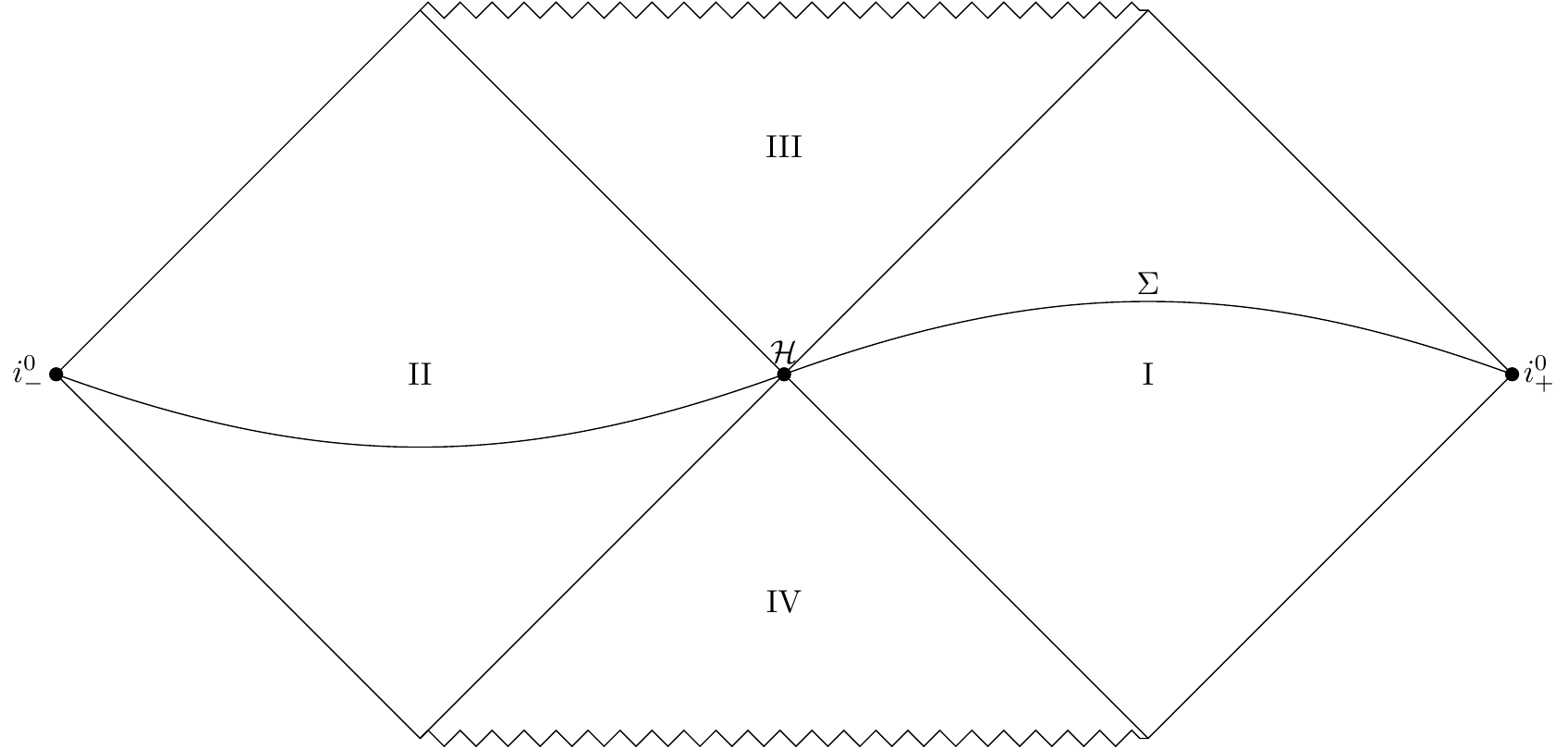}
\caption[Disconnected boundaries and charge definition]{The surface $\Sigma$ is used to define charges. This hypersurface has two boundaries at $(i_+^0,i_-^0)$. If we restrict $\Sigma$ to region I, there are still two boundaries at $(i_+^0,\cH)$. }
\label{fig:Schd}
\end{figure}

Accordingly the charges, as defined in \eqref{delta charge}, will get contributions from both boundaries
\begin{align}\label{charges-two boundaries}
\de Q_\chi&=\oint_{\p\Sigma}{\boldsymbol k}_\chi [\delta\Phi,\Phi]=\oint_{+\infty}{\boldsymbol k}_\chi [\delta\Phi,\Phi]-\oint_{-\infty}{\boldsymbol k}_\chi [\delta\Phi,\Phi]
\end{align}
where $+\infty,-\infty$ symbolically refer to the closed surfaces at two boundaries of $\Sigma$.  However, the charges are usually computed at only \textit{one} compact codimension 2 surface. One resolution is that one assumes a strong boundary condition over one of the boundaries such that all boundary terms coming from that surface vanishes. This seems to be the assumption of \cite{Wald:1993nt,Iyer:1994ys}. But this may not be the appropriate choice in some physical problems. Examples are AdS$_2$ and the NHEG geometries that we will discuss in \ref{chapter-NHEG}. As we will see in chapter \ref{chapter NHEG phase space}, if the charges are computed as \eqref{charges-two boundaries}, they will be all identically zero, since the contribution from all surfaces are the same. Therefore we have defined charges only on one compact codimension 2 surface. This is what is done in \cite{Guica:2008mu}, and all related papers.

Indeed there is an alternative definition of charges. In this chapter, we defined charges from \eqref{delta gen cov def} and we used Stokes' theorem to convert it to an integral on a codimension two surface. However, it is legitimate to \textit{define} charges from the $d-2$ form $\bk_\chi$ in \eqref{omega off shell}, that is 
\begin{align}\label{charge- one surface}
\de Q_\chi&\equiv\oint_{\infty}{\boldsymbol k}_\chi [\delta\Phi,\Phi]\,,
\end{align}
and the bracket of charges can be simply defined using the above equation as 
\begin{align}\label{charge bracket}
\{Q_\chi,Q_\eta\}&\equiv \de_\eta Q_\chi\,.
\end{align}
However, there is a deficiency in this approach. We loose the connection between charge and the \textit{generator} of a gauge symmetry. The charge itself is not the generator of a gauge symmetry, since it does not include the constraint bulk term. The definition of charges in spacetimes with disconnected boundaries, seems to be a well defined problem that to our knowledge is not addressed completely in the literature. We will leave this problem to later works, although we expect that \eqref{charge- one surface} is the correct definition, since it works for both case. Also in chapter \ref{chapter-AdS3}, \ref{chapter NHEG phase space}, we will use \eqref{charge- one surface} to compute charges.

\chapter{AdS$_3$ phase space and its surface degrees of freedom}\label{chapter-AdS3}
\vspace{3cm}
\section{Introduction and outline}

It is well known that  Einstein gravity in three dimensions admits no propagating degrees of freedom (``bulk gravitons''),  but still admits different interesting solutions like black holes \cite{Banados:1992wn, Banados:1992gq}, particles \cite{Deser:1983tn, Deser:1983nh}, wormholes  \cite{Brill:1995jv,Brill:1998pr,Skenderis:2009ju} and a novel boundary dynamics \cite{Brown:1986nw,Ashtekar:1996cd,Barnich:2006av}.   Moreover, it can arise as a consistent subsector of higher dimensional matter-gravity theories, see e.g. \cite{Aharony:1999ti, SheikhJabbaria:2011gc}.  Therefore, three-dimensional gravity can be viewed as a simplified and fruitful setup to analyze and address issues related to the physics of black holes and quantum gravity.

In three dimensions the Riemann tensor is completely specified in terms of the Ricci tensor,  and hence the equations of motion force the geometry to be locally maximally symmetric. However, there is still the possibility of having solution with nontrivial topology that can be obtained by taking a descrete quotient of a globally maximally symmetric solution. This was used to construct black hole solutions in \ads known as BTZ black holes \cite{Banados:1992wn,Banados:1992gq}. Also solutions with a conical singularity can be ontained similarly that represent a ``particle" in 3 dimensions \cite{Deser:1983nh}. Besides these, a new type of dynamics can arise due to the existence of a boundary for the spacetime, known as boundary degrees of freedom that was first discussed in context of \ads in the seminal work of Brown and Henneaux \cite{Brown:1986nw}. There, it was pointed out that one may associate nontrivial conserved  charges, to diffeomorphisms which preserve prescribed (Brown-Henneaux) boundary conditions. The surface charges formed two copies of the Virasoro algebra. It was realized that the Virasoro algebra should be interpreted in terms of a holographic dictionary with a conformal field theory \cite{Strominger:1997eq}. These ideas found a more precise and explicit formulation within the celebrated AdS$_3$/CFT$_2$ dualities in string theory \cite{Kraus:2006wn}. Finding a conformal field theory dual to asymptotic \ads geometries is still an open problem, although many advances have occured\cite{Maldacena:1998bw,Skenderis:2002wp,Maloney:2007ud,Witten:2007kt,Banados:1998gg,Rooman:2000ei,Ryu:2006ef,Balasubramanian:2009bg,SheikhJabbaria:2011gc,Barnich:2012aw,Li:2013pra,Sheikh-Jabbari:2014nya, Maloney:2015ina, Kim:2015qoa}.

In this chapter, we revisit the Brown-Henneaux analysis from the phase space point of view, and show that the surface charges and the associated algebra and dynamics can be defined not only on the circle at spatial infinity, but also on any closed curve inside the bulk obtained by a smooth deformation which does not cross any geometric defect or topological obstruction. This was previously known in the Chern-Simons formulation of 3d gravity in which the radial direction appears as a ``gauge'' direction and drops out of all charge computations. However, the point we stress here is that the notion of asymptotic symmetries can be extended into the bulk and form the set of (local) symplectic symmetries of the \ads\,  phase space.

We start with the set of Ba\~nados geometries \cite{Banados:1998gg} which constitute all locally AdS$_3$ geometries with Brown-Henneaux boundary conditions. We show that the invariant presymplectic current \cite{Barnich:2007bf} (but not the Lee-Wald presymplectic form \cite{Lee:1990nz}) vanishes on-shell in the entire bulk spacetime. The charges can hence be defined over any closed curve in the bulk. We generalize this phenomenon by introducing the notion of \textit{local symplectic symmetries} and investigate their properties. Local symplectic symmetries were also observed in the near-horizon region of extremal black holes \cite{Compere:2015bca, Compere:2015mza}.

Furthermore, we will study in more detail the extremal  sector of the phase space. Boundary conditions are known in the decoupled near-horizon region of the extremal BTZ black hole which admit a chiral copy of the Virasoro algebra \cite{Balasubramanian:2009bg}. Here, we extend the notion of decoupling limit to more general extremal metrics in the Ba\~nados family and show that one can obtain this (chiral) Virasoro algebra as a limit of the bulk symplectic symmetries, which are defined from the asymptotic AdS$_3$ region all the way to the near-horizon region. Quite interestingly, the vector fields defining the Virasoro symmetries are distinct from all previous ansatzes for near-horizon symmetries \cite{Carlip:1998uc,Guica:2008mu,Balasubramanian:2009bg,Compere:2015bca, Compere:2014cna,Compere:2015mza}.

Ba\~nados geometries in general have (at least) two global $U(1)$ Killing vectors \cite{Sheikh-Jabbari:2014nya}. We will study the conserved charges  $J_\pm$ associated with these two Killing vectors and show that these charges commute with the Virasoro charges associated with symplectic symmetries. We then discuss how the elements of the phase space may be labeled using the $J_\pm$ charges. We then review the coadjoint representations of Virasoro algebra (see \cite{Balog:1997zz, Witten:1987ty}) and show that the phase space classifies into (a direct product of two) coadjoint representations of Virasoro algebra. The charges $J_\pm$ then turn out to be invariants on the coadjoint orbits and can be used to label the orbits.  .  We also discuss briefly that for geometries having a Killing horizon \cite{Sheikh-Jabbari:2014nya}, the entropy is another invariant of the orbit, which together with $J_\pm$, satisfies a first law of thermodynamics.

\subsection{Outline}
In section \ref{sec-AAdS}, we introduce the notion of asymptotically AdS spacetimes in arbitrary dimensions in a rigorous way and stress the role of boundary conditions. We then restrict to 3 dimensions and introduce Ba\~nados geometries. In section \ref{symp-sym-sec} and \ref{sec-AdS3-symmetries}, we establish that the family of locally AdS$_3$ geometries with Brown-Henneaux boundary conditions form a phase space with two copies of the Virasoro algebra as symplectic symmetries.  In section \ref{Killing-sec}, we show that each metric in the phase space admits two $U(1)$ Killing vectors which commute with the vector fields generating the symplectic symmetries, once we use the appropriately ``adjusted (Lie) bracket'' \cite{Barnich:2010eb,2010AIPC.1307....7B}. We show that the charge associated with these two Killing vectors are integrable over the phase space and commute with the generators of the Virasoro symplectic symmetries. In section \ref{Banados-Orbits-sec}, we discuss how the phase space falls into Virasoro coadjoint orbits and how the Killing charges may be attributed to each orbit and discuss the first law of thermodynamics on the black hole orbits. In section \ref{extremal-NH-sec}, we focus on a chiral half of the phase space which is obtained through decoupling limit over the extremal geometries. We show that this sector constitutes a phase space with symplectic symmetries of its own.

\section{Asymptotically AdS spacetimes}\label{sec-AAdS}
In this section, we briefly review the definition of asymptotic AdS geometries in arbitrary dimensions. This is a well established subject \cite{Skenderis:2000in,Skenderis:2002wp}  as it is the first step towards AdS/CFT duality. 

An asymptotically AdS spacetime is by definition \textit{a conformally compact Einstein geometry} \cite{Penrose:1986ca,Skenderis:2002wp}. This means that the spacetime $\cal M$ possess a boundary and the metric has a second oder pole at the boundary. Therefore the metric does not induce a metric on the boundary. However, there is a defining function, say $r(x)$ such that $r(x)^2$ smoothly extends to the boundary and $g_{(0)}=r^2 g \vert_{\p\cal M}$ is nondegenerate. If in addition, the metric satisfies Einstein equations with a negative cosmological constant, then this geometry is called an asymptotically AdS geometry (AAdS). Note that if $r(x)$ is a suitable defining function, then also is $re^w$. Due to the arbitrariness in the choice of defining function, the metric only defines a conformal structure on the boundary.  

In a $d$ dimensional spacetime, there are $d$ gauge degrees of freedom in the metric. Therefore one can choose $d$ gauge conditions to fix the form of metric. A suitable choice of coordinate system is Gaussian normal coordinates emanating from the boundary in which the metric of an AAdS geometry can be written in the form
\begin{align}
	ds^2&=\dfrac{1}{z^2}(dz^2+g_{ij}(z,x^i)dx^idx^j)
\end{align}
where $z=0$ is the location of the boundary, and $x^i$ label the boundary. By definition, $g_{ij}(z,x^i)$ must smoothly extend to the boundary, hence
\begin{align}
	g_{ij}(z,x^i)&={g_{(0)}}_{ij}+z {g_{(1)}}_{ij} + z^2 {g_{(2)}}_{ij}+\cdots
\end{align}
where ${g_{(0)}}_{ij}$ is a nondegenerate metric on the boundary. Einstein's equations become algebraic and can be solved order by order in the $z$ variable \cite{Graham:1999jg}.It turns out that in an $(d+1)$ dimensional AAdS spacetime, all coefficients multiplying odd powers of $r$ vanish up to the order $r^d$ and the metric can be written as 
\begin{align}
	g_{ij}&={g^{(0)}}_{ij}+z^2 {g^{(2)}}_{ij} +\cdots + z^d g^{(d)}_{ij}+h^{(d)}_{ij} z^d \log z^2 +\cdots
\end{align}
in which the logarithmic term only appears in $d=$even and is related to the conformal anomaly. Einstein equations uniquely determine $g_{(2)},\cdots, g_{(d-2)},h_{(d)}$ as well as trace and covariant divergence of $g_{(d)}$ (see appendix A of \cite{deHaro:2000xn} for explicit expressions). Physically $g_{(d)}$ is related to the quasilocal stress tensor of the geometry. Notably in 3 dimensions, the above expansion truncates and we can explicitly express the solutions with suitable boundary conditions \cite{Skenderis:1999nb,Banados:1998gg}. This is what we will do next.

We use the redefinition $r=\dfrac{1}{z}$ so that the boundary is located at $r=\infty$. Choosing the Gaussian normal coordinates, then corresponds to the Fefferman-Graham gauge conditions
\begin{align}\label{F-G gauge}
	g_{rr}=\dfrac{1}{r^2},\qquad g_{ra}=0
\end{align}
and the metric becomes
\begin{align}
	ds^2&=\dfrac{dr^2}{r^2}+\gamma_{ab}(r,x^i)\,dx^a\,dx^b\,.
\end{align}
The boundary is located at $r\to \infty$. Being asymptotic AdS implies that
\begin{align}
	\gamma_{ab}&=r^2 \Big(g^{(0)}_{ab}(x^a)+\dfrac{1}{r} g^{(1)}_{ab}(x^a)+\cdots\Big)\,.
\end{align}
Now comes the choice of boundary conditions. Dirichlet boundary conditions amounts to identifying $g_{(0)}$ with a prescribed boundary metric. The famous Brown Henneaux boundary conditions \cite{Brown:1986nw} are indeed Dirichlet boundary conditions with a flat boundary metric 
\bea\label{BCBH}
g^{(0)}_{ab}dx^a dx^b = -dx^+ dx^-,
\eea
together with the periodic identifications $(x^+,x^-) \sim (x^++2\pi ,x^--2\pi)$ which identify the boundary metric with a flat cylinder (the identification reads $\phi \sim \phi +2 \pi$ upon defining $x^\pm = t/\ell \pm \phi$). Other relevant Dirichlet boundary conditions include the flat boundary metric with no identification (the resulting solutions are usually called ``Asymptotically Poincar\'e AdS$_3$''), and the flat boundary metric with null orbifold identification $(x^+,x^-) \sim (x^++2\pi,x^-)$ which is relevant to describing near-horizon geometries \cite{Coussaert:1994tu,Balasubramanian:2009bg, Sheikh-Jabbari:2014nya}.\footnote{Other boundary conditions which lead to different symmetries were discussed in \cite{Compere:2008us,Compere:2013bya,Troessaert:2013fma}.}

Until now, we have not used the Einstein equation to determine the set of solutions. It turns out that in pure Einstein gravity with a negative cosmological constant in 3 dimensions, with the following action and field equations
\be\label{action-eom}
S=\frac{1}{16\pi G}\int d^3x \sqrt{-g} (R+\frac{1}{\ell^2} ),\qquad R_{\mu\nu}=-\frac{2}{\ell^2}g_{\mu\nu}.
\ee
the set of all asymptotic \ads solutions with flat boundary metric take the form \cite{Banados:1998gg}
\begin{align}\label{Banados}
	ds^2&=\ell^2\dfrac{dr^2}{r^2}-\Big(rdx^+-\ell^2\dfrac{L_-(x^-)dx^-}{r}\Big)\Big(rdx^--\ell^2\dfrac{L_+(x^+)dx^+}{r}\Big)
\end{align}
where $L_\pm$ are two single-valued arbitrary functions of their arguments. This set contains interesting solutions. The constant $L_\pm$ cases correspond to better known geometries \cite{Deser:1983nh,Banados:1992wn,Banados:1992gq}: $L_+=L_-=-1/4$ case corresponds to AdS$_3$ in global coordinates, $-1/4< L_\pm< 0$ case corresponds to conical defects (particles on AdS$_3$), $L_-=L_+=0$ case corresponds to massless BTZ and generic positive values of $L_\pm$ correspond to generic BTZ black holes \cite{Banados:1992wn} of mass and angular momentum $(L_++L_-)/(4G)$ and $\ell(L_+-L_-)/(4G)$ respectively.  The selfdual orbifold of AdS$_3$ \cite{Coussaert:1994tu} belongs to the phase space with null orbifold identification and $L_-=0, L_+\neq 0$.

\section{AdS$_3$ phase space}\label{symp-sym-sec}


In this section we will show that the set of Ba\~nados metrics \eqref{Banados} forms a well-defined on-shell phase space. To this end, we need to define a symplectic structure over this phase space. The symplectic structure is a two-form acting on vectors tangent to the phase space. Since we will be interested specially in the form of on-shell perturbations, we will discuss them first. Given that the set of all solutions are of the form \eqref{Banados}, the on-shell tangent space is clearly given by metric variations {of the form} 
\be\label{generic-perturbations}
\delta g= g(L+\delta L)- g(L)\,,
\ee
where $\delta L_\pm$ are arbitrary single-valued  functions.
The vector space of all on-shell perturbations $\delta g$ {can be written as} the direct sum of two types of perturbations: those which are generated by diffeomorphisms and those which are not, and that we will refer to as \textit{parametric perturbations}.

There are two known definitions for the presymplectic form: the one $\boldsymbol \omega^{LW}$ by Lee-Wald \cite{Lee:1990nz} (see also Crnkovic and Witten \cite{Crnkovic:1986ex}) and invariant presymplectic form $\boldsymbol \omega^{inv}$ as defined in \cite{Barnich:2007bf}. The invariant presymplectic form is determined from field equations, while the Lee-Wald presymplectic form is determined from the action, see \cite{Compere:2007az} for details. In this chapter we choose to work with $\boldsymbol \omega^{inv}$ since it turns out that it has the interesting property that it exactly vanishes \emph{on-shell} on the phase space, that is,
\bea\label{noo}
\boldsymbol\omega^{inv}[\de_1 g , \de_2 g ; g] \approx 0 .
\eea
This is not the case for the Lee-Wald presymplectic form.

The fact that the invariant presymplectic form  vanishes on-shell illustrates the fact that there are no propagating bulk degrees of freedom in three dimensional vacuum Einstein gravity. Nevertheless, this does not exclude the existence of a lower dimensional dynamics as we will show next.

\section{Symplectic symmetries and charges}\label{sec-AdS3-symmetries}

As we mentioned earlier, a field variation tangent to the phase space takes the form \eqref{generic-perturbations}. In this section, we are interested in field variations in the above class which are produced by an infinitesimal coordinate transformation along a vector field $\chi$. In order that this variation respects the conditions \eqref{F-G gauge}, the vector $\chi$ should take the form \cite{Compere:2008us}
\begin{align}\label{Allowed}
\chi^r=r \,\sigma(x^a),\qquad
\chi^a=\eps^a(x^b)-\ell^2\p_b \,\sigma \int_r^\infty \dfrac{dr'}{r'} \gamma^{ab}(r',x^a)
\end{align}
Moreover, by requiring the boundary metric to be invariant (Dirichlet boundary conditions that we have imposed) the functions $\sigma(x^a)$ and $\eps^a(x^b)$ are constrained by the condition $\de g^{(0)}_{ab}\equiv\mathcal{L}_{\vec{\eps}}\,g^{(0)}_{ab}+2\sigma g^{(0)}_{ab}=0$.
That is, $\vec{\eps} \equiv (\eps_+(x^+),\eps_-(x^-))$ should be a conformal Killing vector of the flat boundary metric and $\sigma$ is defined as the Weyl factor  in terms of $\vec{\eps}$.

Solving the above integral for a given Ba\~nados metric,  we arrive at
\begin{align}\label{symplectic symmetries}
\chi&=-\dfrac{r}{2}(\eps_+'+\eps_-')\p_r + \Big(\eps_+ +\dfrac{\ell^2r^2 \eps_-''+\ell^4L_- \eps_+''}{2(r^4-\ell^4L_+L_-)}\Big)\p_+ +\Big(\eps_- +\dfrac{\ell^2r^2 \eps_+''+\ell^4L_+ \eps_-''}{2(r^4-\ell^4L_+L_-)}\Big)\p_-,
\end{align}
where $\eps_\pm$ are two arbitrary single-valued periodic functions of $x^\pm$ and possibly of the fields $L_+(x^+),\, L_-(x^-)$, and the \emph{prime} denotes derivative w.r.t. the argument. As we see,
\begin{enumerate}
	\item $\chi$ is a \emph{field-dependent} vector field. That is, even if the two arbitrary functions $\eps_\pm$ are field independent, it has explicit dependence upon $L_\pm$: $\chi=\chi(\eps_\pm; L_\pm)$.
	\item The vector field $\chi$ is defined in the entire coordinate patch spanned by the Ba\~nados metric, not only asymptotically.
	\item Close to the boundary, at large $r$, $\chi$ reduces to the Brown-Henneaux asymptotic symmetries \cite{Brown:1986nw}. Also, importantly, at large $r$ the field-dependence of $\chi$ drops out if one also takes $\eps_\pm$ field-independent.
\end{enumerate}

The variation in metric generated by an infinitesimal diffeomorphism along a vector  $\chi({\eps}_\pm)$ given by \eqref{symplectic symmetries} takes the form 
\begin{align}
g_{\mn}(L_+,L_-)+\delta_\chi g_{\mn}&=g_{\mn}(L_++\de_{\eps_+} L_+, L_-+\de_{\eps_-} L_-),
\end{align}
where
\be\begin{split}\label{delta-chi-L}
	\de_{\eps_+} L_+&= \eps_+ \p_+ L_+ +2L_+ \p_+ \eps_+ -\dfrac{1}{2}\p_+^3\eps_+,\cr
	\de_{\eps_-} L_-&= \eps_- \p_- L_- +2L_- \p_- \eps_- -\dfrac{1}{2}\p_-^3\eps_-.
\end{split}\ee
A diagonal matrix with components $L_+,L_-$ hence transforms exactly in the same way as the energy-momentum tensor of a 2 dimensional conformal field theory under generic infinitesimal conformal transformations. This is probably the first entry of the dictionary of AdS$_3$/CFT$_2$ correspondence \cite{Aharony:1999ti, Kraus:2006wn}. Remarkably, the last term related to the central extension of the Virasoro algebra is a quantum anomalous effect in the CFT side, while in the gravity side appears classically.
\subsection{Local symplectic symmetries}\label{local symplectic symmetries}
Here we review the notion of \emph{local symplectic symmetries} introduced in \cite{Compere:2015bca,Compere:2015mza} (see also \cite{Compere:2014cna} for earlier observations) and show that the vector field $\chi$ determines local symplectic symmetries of \ads phase space.
\begin{definition}
\setlength{\parindent}{1cm} A vector $\chi$ is called a \textit{local symplectic symmetry} of phase space if the presymplectic current contracted with the Lie derivative of the metric with respect to $\chi$ vanishes on-shell everywhere in spacetime. Explicitly it reads
\begin{align}\label{symplectic condition}
\boldsymbol \omega [ \Phi;\de \Phi,\mathcal{L}_{\chi} \Phi ] & \approx 0, 
\end{align}
for all $\Phi$ solving the equations of motion and all $\de\Phi$ solving the linearized field equations around $\Phi$. Moreover the associated charge should be finite and conserved.
\end{definition}

The above definition is a local condition and hence very tight. However, as we will show, they can be realized in interesting gravitational situations. Below we list the nice properties of local symplectic symmetries in a proposition. We postpone the proof of these properties to appendix \eqref{appendix proofs}.
\begin{proposition}\label{prop local symplectic symmetries}
The set of local symplectic symmetries have the following properties:
\begin{enumerate}
\item they form a closed algebra,
\item their corresponding charge is integrable,
\item their corresponding charge can be computed over any codimension 2 closed surface in the bulk that can be continuously deformed from the asymptotics.
\end{enumerate}
\end{proposition}

The vectors $\chi$ defined in \eqref{symplectic symmetries} clearly satisfy the condition \eqref{symplectic condition} as a result of \eqref{noo}. Then according to the above properties, the charges associated to the symplectic symmetries of \ads phase space is integrable and can be defined over any closed codimension two surface $\mathcal{S}$ (circles in 3d) anywhere in the bulk. Moreover, as we discussed before, the symplectic symmetries \eqref{symplectic symmetries} coincide with Brown-Henneaux asymptotic symmetries in the asymptotic region. Therefore, the concept of ``local symplectic symmetry'' extends the notion of ``asymptotic symmetry'' inside the bulk.

\subsection{Charges associated with local symplectic symmetries}
A direct computation gives the formula for the infinitesimal charge one-forms as 
\bea
\boldsymbol k_\chi [\delta g ; g]= \boldsymbol{\hat{ k}}_\chi[\delta g ; g] + d \boldsymbol B_\chi[\delta g ; g],
\eea
where
\bea\label{chk}
\boldsymbol{\hat{ k}}_\chi [\delta g ; g]= \frac{\ell}{8\pi G} \left(\eps_+(x^+) \,\delta L_+ (x^+) dx^+ -  \eps_-(x^-) \,\delta L_-(x^-) dx^- \right),
\eea
is the expected result and $$\boldsymbol B_\chi =\dfrac{\ell(\eps_+' + \eps_-')(L_+\delta L_- - L_-\delta L_+)}{32\pi G(r^4 - L_+L_-)},$$ is an uninteresting boundary term which drops after integration on a circle. Note also that since $\eps_\pm$ are assumed to be field independent $\boldsymbol{\hat{ k}}_\chi [\delta g ; g]= \de \cH_\chi$ where
\begin{align}\label{cal H def}
	\cH_\chi&=\frac{\ell}{8\pi G} \left(\eps_+(x^+) \, L_+ (x^+) dx^+ -  \eps_-(x^-) \, L_-(x^-) dx^- \right)
\end{align}
as required in \eqref{LSS condition}. The integrability of charges is also guaranteed by the above property. In the case of periodic identifications leading to a cylindrical boundary, we are then led to the standard Virasoro charges
\bea\label{ch1}
Q_\chi [g] = \oint_S \boldsymbol \cH_\chi [g]  = \frac{\ell}{8\pi G}\int_0^{2\pi} d\phi \left(\eps_+(x^+) L_+ (x^+) + \eps_-(x^-) L_-(x^-) \right),
\eea
where $\phi \sim \phi +2 \pi$ labels the periodic circle $S$. The charges are manifestly defined everywhere in the bulk in the range of the Ba\~nados coordinates. 
Note that the charges are normalized to zero for the zero mass BTZ black hole $\bar g$ for which $L_\pm=0$ \footnote{As we will discuss in section \eqref{Banados-Orbits-sec}, the zero mass BTZ can only be used as a reference to define charges over a patch of phase space connected to it. For other disconnected patches, one should choose other reference points. }. In AdS$_3$/CFT$_2$, the functions $L_+, L_-$ (with the right numerical factor) are interpreted as components of the dual stress-energy tensor of CFT.

	\subsection{Charge algebra and adjusted bracket}\label{modified-bracket-sec}
	
According to \eqref{Poisson-charge}, the algebra of conserved charges is defined as
	\bea
	\{ Q_{\chi_1}, Q_{\chi_2} \} = - \delta_{\chi_1} Q_{\chi_2},
	\eea
Let us denote the charge associated with the vector $\chi^+_n = \chi( \eps_+ = e^{i n x^+}, \eps_- = 0)$ by $L_n$ and the charge associated with the vector $\chi^-_n = \chi( \eps_+ = 0 , \eps_- = e^{i n x^-})$ by $\bar L_n$.
From the definition of charges \eqref{ch1} and the transformation rules \eqref{delta-chi-L}, we directly obtain the charge algebra
	\bea\label{cea}
	\{ L_m , L_n \} &=& (m-n)L_{m+n}+ \frac{c}{12}m^3 \delta_{m+n,0}, \cr
	\{ \bar L_m, L_n \} &=& 0,\\
	\{\bar L_m , \bar L_n \} &=& (m-n) \bar L_{m+n}+ \frac{c}{12}m^3 \delta_{m+n,0}, \nonumber
	\eea
	where
	\be
	c=\frac{3\ell}{2G},
	\ee
	is the Brown-Henneaux central charge. These are the famous two copies of the Virasoro algebra. A term proportional to $m$ can be arbitrarily added to or removed from the central term by a constant shift in the generator $L_0$ of the algebra. We have fixed the above form by requiring the charges to be zero for the massless BTZ black hole.
	
	We discussed in chapter \ref{chapter-covariant phase space}, theorem \eqref{theorem-charge}, that the algebra of charges represents the algebra of symplectic symmetries up to a central extension. However computing commutator of symplectic symmetries through the Lie bracket, we do not find even a closed algebra. This apparent inconsistency stems from the fact that the symplectic symmetry vectors \eqref{symplectic symmetries} are \textit{field dependent}, \ie depend on $L_\pm$. It turns out that in the case of field dependent vectors, one  should  ``adjust'' the Lie bracket by subtracting off the terms coming from the variations of fields within the $\chi$ vectors \cite{2010AIPC.1307....7B}. Explicitly,
	\begin{align}\label{modified-bracket}
	\big[\chi(\eps_1;L),\chi(\eps_2;L)\big]_*&\equiv\big[\chi(\eps_1;L),\chi(\eps_2;L)\big]_{L.B}-\Big(\de^L_{\eps_1} \chi(\eps_2;L)-\de^L_{\eps_2} \chi(\eps_1;L)\Big),
	\end{align}
	where the variations $\de^L_\eps$ are defined as
	\begin{align}\label{field var}
	\de^L_{\eps_1} \chi(\eps_2;L)&=\de_{\eps_1} L\;\dfrac{\p}{\p L}\chi(\eps_2;L).
	\end{align}
	It can be checked that the \emph{adjusted bracket} $[\,,\,]_*$ satisfies all the properties of a bracket. Interestingly this is precisely the bracket which lead to the representation of the algebra by conserved charges in the case of field-dependent vector fields. Here the field dependence is stressed by the notation $\chi(\eps;L)$. Here and in the following we use a compressed notation by merging the left and right sectors into single symbols, $\eps = (\eps_+,\eps_-)$ and $L = (L_+,L_-)$.
	
	Using the adjusted bracket, it can be checked that the vectors $\chi(\eps)$ form a closed algebra
	\begin{align}\label{bracket}
	\big[\chi(\eps_1;L),\chi(\eps_2 ;L)\big]_*&=\chi( \eps_1  \eps_2' - \eps_1'  \eps_2 ;L).
	\end{align}
	Upon expanding in modes $\chi^\pm_n$, one  obtains {two copies of} the Witt algebra
	\bea\label{Witt}
	\big[\chi^+_m,\chi^+_n  \big]_*  &=& (m-n)\chi^+_{m+n}, \cr
	\big[\chi^+_m,\chi^-_n  \big]_* &=& 0, \\
	\big[ \chi^-_m, \chi^-_n \big]_{*}  &=& (m-n) \chi^-_{m+n}, \nonumber
	\eea
	which is then represented by the conserved charges as the centrally extended algebra \eqref{cea}.

\subsection{Finite form of symplectic symmetry transformations}\label{finite-transf-sec}
	
	We discussed in the previous subsections that the phase space of Ba\~nados geometries admits a set of non-trivial  perturbations generated by the vector fields $\chi$. Then, there exists finite coordinate transformations (obtained by  ``exponentiating the $\chi$'s'') which map a Ba\~nados metric to another one. That is, there are coordinate transformations
	\be\label{Finite-diff}
	x^\pm \to X^\pm=X^\pm(x^\pm, r)\,,\qquad r\to R=R(x^\pm,r),
	\ee
	with $X^\pm,R$ such that the metric $\tilde g_{\mu\nu}=g_{\alpha\beta}\frac{\partial x^\alpha}{\partial X^\mu} \frac{\partial x^\beta}{\partial X^\nu}$ is a Ba\~nados geometry with appropriately transformed $L_\pm$. Such transformations change the physical charges. They are not gauge transformations but are instead  solution or charge generating transformations.
	
	Here, we use the approach of Rooman-Spindel \cite{Rooman:2000ei}. We start by noting that the technical difficulty in ``exponentiating'' the $\chi$'s arise from the fact that $\chi$'s are field dependent and hence their form  changes as we change the functions $L_\pm$, therefore the method discussed in section 3.3 of \cite{Compere:2015mza} cannot be employed here. However, this feature disappears in the large $r$ regime. Therefore, if we can find the form of \eqref{Finite-diff} at large $r$ we can read how the $L_\pm$ functions of the two transformed metrics should be related. Then, the subleading terms in $r$ are fixed  such that the form of the Ba\~nados metric is preserved. This is guaranteed to work as a consequence of Fefferman-Graham's theorem \cite{FGpaper:1985fg}. From the input of the (flat) boundary metric and first subleading piece (the boundary stress-tensor), one can in principle reconstruct the entire metric.
	
	{It can be shown that the finite coordinate transformation preserving \eqref{F-G gauge} is }
	\bea
	&&\xp\to X^+=h_+(\xp) + \dfrac{\ell^2}{2r^2} \frac{h_-''}{h_-'}\frac{h_+'}{h_+}+{\cal O}(r^{-4}),\nn\\
	&&r\to R=\frac{r}{\sqrt{h_+' h_-'}}+{\cal O}(r^{-1}),\\
	&&\xn\to X^-=h_-(x^-)+\dfrac{\ell^2}{2r^2} \frac{h_+''}{h_+'}\frac{h_-'}{h_-}+{\cal O}(r^{-4}),\nn
	\eea
	where $h_\pm(x^\pm +2\pi)=h_\pm(x^\pm) \pm 2\pi$, $h_\pm$ are monotonic ($h_\pm' > 0$) so that the coordinate change is a bijection. At leading order (in $r$), the functions $h_\pm$ parametrize a generic conformal transformation of the boundary metric.
	
	Acting upon the metric by the above transformation one can read how the functions $L_\pm$ transform:
	\bea\label{finite-L-transf}
	L_+(x^+) &\to \tilde L_+={h_+'}^2 L_+ - \dfrac{1}{2} S[h_+;x^+],\\
	L_-(x^-) &\to \tilde L_-={h_-'}^2L_- - \dfrac{1}{2} S[h_-;x^-],
	\eea
	where $S[h;x]$ is the Schwarz derivative
	\be
	S[h(x);x]=\frac{h'''}{h'}-\frac{3h''^2}{2h'^2}.
	\ee
	It is readily seen that in the infinitesimal form, where $h_\pm(x)=x^\pm+\eps_\pm(x)$, the above reduce to \eqref{delta-chi-L}.
	It is also illuminating to explicitly implement the positivity of $h'_\pm$ through
	\be
	h_\pm'= e^{\Psi_\pm},
	\ee
	where $\Psi_\pm$ are two real single-valued functions. In terms of $\Psi$ fields the Schwarz derivative takes a simple form and the expressions for $\tilde L_\pm$ become
	\be\label{Liouville-form}
	\tilde L_+[\Psi_+,L_+]= e^{2\Psi_+} L_+(x^+)+\frac{1}{4}\Psi_+'^2-\frac{1}{2}\Psi_+'',\qquad
	\tilde L_-[\Psi_-,L_-]= e^{2\Psi_-} L_-(x^-)+\frac{1}{4}\Psi_-'^2-\frac{1}{2}\Psi_-''.
	\ee
	This reminds the form  of a Liouville stress-tensor and dovetails with the fact that AdS$_3$ gravity with Brown-Henneaux boundary conditions may be viewed as a Liouville theory \cite{Coussaert:1995zp} (see also  \cite{Kim:2015qoa} for a recent discussion).

We finally note that not all functions $h_\pm$ generate new solutions. The solutions to $\tilde L_+ = L_+$, $\tilde L_- = L_-$ are coordinate transformations which leave the fields invariant: they are finite transformations whose infinitesimal versions are generated by the isometries. There are therefore some linear combinations of symplectic symmetries which do not generate any new charges. These ``missing'' symplectic charges are exactly compensated by the charges associated with the Killing vectors that we will discuss in section \eqref{Killing-sec}.

\section{The two Killing symmetries and their charges}\label{Killing-sec}

So far we discussed the symplectic symmetries of the phase space. These are associated with non-vanishing metric perturbations which are degenerate directions of the on-shell presymplectic form. A second important class of symmetries are the Killing vectors which are associated with vanishing metric perturbations. In this section we analyze these vector fields, their charges and their commutation relations with the symplectic symmetries. We will restrict our analysis to the case of asymptotically globally AdS$_3$ where $\phi$ is $2\pi$-periodic. We use Fefferman-Graham coordinates for definiteness but since Killing vectors are geometrical invariants, nothing will depend upon this specific choice.

\subsection{Global Killing vectors}\label{Global-Killings-sec}

Killing vectors are vector fields along which the metric does not change. All diffeomorphisms preserving the Fefferman-Graham coordinate system are generated by the vector fields given in \eqref{symplectic symmetries}. Therefore, Killing vectors have the same form as $\chi$'s, but with the extra requirement that $\delta L_\pm$ given by \eqref{delta-chi-L} should vanish. Let us denote the functions $\eps_\pm$ with this property by $K_\pm$ and the corresponding Killing vector by $\zeta$ (instead of $\chi$). Then, $\zeta$ is a Killing vector if and only if
\begin{align}\label{stabilizer}
K_+'''-4L_+ K_+'-2K_+ L_+' =0,\qquad K_-'''-4L_- K_-'-2K_- L_-' =0.
\end{align}
These equations were thoroughly analyzed in \cite{Sheikh-Jabbari:2014nya}  and  we only provide a summary of the results relevant for our study here.
The above linear third order differential equations have three linearly independent solutions and hence Ba\~nados geometries in general have six (local) Killing vectors which form an $sl(2,\mathbb{R})\times sl(2,\mathbb{R})$ algebra, as  expected. The three solutions take the form $K_+ = \psi_i \psi_j$, $i,j=1,2$ where $\psi_{1,2}$ are the two independent  solutions to the second order Hill's equations
\bea
\psi'' = L_+(x^+) \psi
\eea
where $L_+(x^++2\pi)=L(x^+)$. Therefore, the function $K_+$ functionally depends upon $L_+$  but not on $L_+'$, i.e. $K_+=K_+(L_+)$. This last point will be crucial for computing the commutation relations and checking integrability as we will shortly see. The same holds for the right moving sector. In general, $\psi_i$ are not periodic functions under $\phi \sim \phi+2\pi$ and therefore not all six vectors described above are global Killing vectors of the geometry. However, Floquet's theorem \cite{magnus2013hill} implies that the combination $\psi_1 \psi_2$ is necessarily periodic. This implies that Ba\~nados geometries have at least two global Killing vectors.
Let us denote these two global Killing vectors by $\zeta_\pm$,
\bea
\zeta_+=\chi(K_+, K_-=0;L_\pm),\qquad  \zeta_-=\chi(K_+=0, K_-; L_\pm),
\eea
where $\chi$ is the vector field given in \eqref{symplectic symmetries}. These two vectors define two global $U(1)$ isometries of Ba\~nados geometries.

The important fact about these global $U(1)$ isometry generators is that they commute with \emph{each} symplectic symmetry generator $\chi$ \eqref{symplectic symmetries}: Since the vectors are field-dependent, one should use  the adjusted bracket \eqref{modified-bracket} which reads explicitly as
\begin{align}
\big[\chi(\eps;L),\zeta(K;L)\big]_*&=\big[\chi(\eps;L),\zeta(K;L)\big]_{L.B.}-\Big(\de^L_\eps \zeta(K;L)-\de^L_{K} \chi(\eps;L)\Big),\nn
\end{align}
where the first term on the right-hand side is the usual Lie bracket. Since $K=K(L)$, the adjustment term reads as
\begin{align}
\de^L_\eps \zeta(K(L);L)&=\de_\eps L\;\dfrac{\p}{\p L}\zeta(K;L) +\zeta(\de^L_\eps K;L), \label{eq:99}\\
\de^L_{K} \chi(\eps;L) &= \de_K L\;\dfrac{\p}{\p L} \chi(\eps;  L) = 0
\end{align}
where we used the fact that $\zeta,\chi$ are linear in their first argument as one can see from \eqref{symplectic symmetries} and we used Killing's equation. We observe that we will get only one additional term with respect to the previous computation \eqref{bracket} due to the last term in \eqref{eq:99}. Therefore,
\begin{align}
\big[\chi(\eps;L),\zeta(K(L);L)\big]_*&= \zeta(\eps\, K' - \eps' \,K;L) -\zeta(\de^L_\eps K;L).
\end{align}
Now the variation of Killing's condition \eqref{stabilizer} implies that
$$
(\de K)'''-4L(\de K)'-2L'\de K=4\delta L K'+2(\delta L)'K.
$$
Then, recalling \eqref{delta-chi-L} and using again \eqref{stabilizer} we arrive at
\begin{align}
\de^L_\eps K&=\eps\, K' - \eps' \,K\label{deltaK},
\end{align}
and therefore
\be\label{chi-zeta-commute}
\big[\chi(\eps;L),\zeta(K(L);L)\big]_*=0.
\ee
The above may be intuitively understood as follows. $\zeta$ being a Killing vector field does not transform $L$, while a generic $\chi$ transforms $L$. Now the function $K$ is a specific function of the metric, $K=K(L)$. The adjusted bracket is defined such that it removes the change in the metric and only keeps the part which comes from Lie bracket of the corresponding vectors as if $L$ did not change.

It is interesting to compare the global Killing symmetries and the symplectic symmetries. The symplectic symmetries are given by \eqref{symplectic symmetries} and determined by functions $\eps_\pm$. The functions $\eps_\pm$ are field independent, so that they are not transformed by moving in the phase space. On the other hand, although the Killing vectors have the same form \eqref{symplectic symmetries}, their corresponding functions $\eps_\pm$ which are now denoted by $K_\pm$, are field dependent as a result of \eqref{stabilizer}. Therefore the Killing vectors differ from one geometry to another. Accordingly if we want to write the Killing vectors in terms of the symplectic symmetry Virasoro modes $\chi_n^\pm$ \eqref{Witt}, we have
\begin{align}
\zeta_+=\sum_n c^+_n(L_+) \chi_n^+,\qquad \qquad 	\zeta_-=\sum_n c^-_n(L_-) \chi_n^-.
\end{align}	
For example for a BTZ black hole, one can show using \eqref{stabilizer} that the global Killing vectors are $\zeta_\pm=\chi_0^\pm$ while for a ``BTZ Virasoro descendant'', which is generated by the coordinate transformations in section \eqref{finite-transf-sec}, it is a complicated combination of different Virasoro modes. For the case of global AdS$_3$ with $L_\pm=-\frac{1}{4}$ (but not for its descendants), \eqref{stabilizer} implies that there are six global Killing vectors which coincide with the subalgebras $\{\chi^+_{1,0,-1}\}$ and $\{\chi^-_{1,0,-1}\}$ of symplectic symmetries.

\subsection{Conserved charges associated with the $U(1)$ Killing vectors}\label{Jpm-sec}

Similarly to the Virasoro charges \eqref{chk}, the infinitesimal charges associated to Killing vectors can be computed using \eqref{kgrav}, leading to
\be\label{delta-J}
\delta J_+  =\frac{\ell}{8\pi G} \int_0^{2\pi}  d\phi\ K_+(L_+) \delta L_+,\qquad \delta J_-  =\frac{\ell}{8\pi G} \int_0^{2\pi}  d\phi\ K_-(L_-) \delta L_-.
\ee
\paragraph{Integrability of Killing charges.}
Given the field dependence of the $K$-functions, one may inquire about the integrability of the charges $J_\pm$ over the phase space.
However, in the present case, the integrability of $J_\pm$ can be directly checked as follows
\begin{align}\label{ci}
\delta_1 (\delta_2 J) &= \frac{\ell}{8\pi G}\oint \delta_1 K(L) \; \delta_2 L=\frac{\ell}{8\pi G} \oint \dfrac{\p K}{\p L}\; \delta_1 L\; \delta_2 L,
\end{align}
and therefore  $\delta_1 (\delta_2 J )- \delta_2 (\delta_1 J )=0$.

Having checked the integrability, we can now proceed with finding the explicit form of charges through an integral along a suitable path over the phase space connecting a reference field configuration to the configuration of interest. However, as we will see in section \eqref{Banados-Orbits-sec}, the Ba\~nados phase space is not simply connected and therefore one cannot reach any field configuration through a path from a reference field configuration. As a result, the charges should be defined independently over each connected patch of the phase space. In section \eqref{Banados-Orbits-sec} we will give the explicit form of charges over a patch of great interest, i.e. the one containing BTZ black holes and their descendants. We then find a first law relating the variation of entropy to the variation of these charges.

\paragraph{Algebra of Killing and symplectic charges.} We have already shown in section \eqref{Global-Killings-sec} that the adjusted bracket between generators of respectively symplectic and Killing symmetries vanish. If the charges are correctly represented, it should automatically follow that the corresponding charges $L_n, J_+$ (and $\bar L_n, J_-$) also commute:
\bea\label{J-L-commute}
\{ J_\pm , L_n \} = \{ J_\pm , \bar L_n \} = 0.
\eea
Let us check \eqref{J-L-commute}. By definition we have
\bea
\{ J_+ , L_n \}  = -\delta_K L_n,
\eea
where one varies the dynamical fields in the definition of $L_n$ with respect to the Killing vector $K$. Since $K$ leaves the metric unchanged, we have $\delta_K L_+(x^+) = 0$ and therefore directly $\delta_K L_n = 0$. Now, let us also check that the bracket is anti-symmetric by also showing
\bea\label{J-fixed-on-orbit}
\{ L_n , J_+ \} \equiv  - \delta_{L_n} J_+ = 0.
\eea
This is easily shown as follows:
\begin{align}\label{Jvar}
\delta_{L_n} J_+ &= \frac{\ell}{8\pi G} \int_0^{2\pi} d\phi\ K_+ \delta_{\eps^+_n} L
= \frac{\ell}{8\pi G} \int_0^{2\pi} d\phi\ K_+ (\eps^+_n L_+' +2 L_+ \eps^+_n{'}-\frac{1}{2}\eps^+_n{'''} )\nn \\
&=  \frac{\ell}{8\pi G} \int_0^{2\pi} d\phi\ (-L_+' K_+ -2L_+ K_+' +\frac{1}{2}K_+''')\eps_{+,n}
= 0
\end{align}
after using \eqref{delta-chi-L}, integrating by parts and then using \eqref{stabilizer}. The same reasoning holds for $J_-$ and $\bar L_n$.

In general, the Ba\~nados phase space only admits two Killing vectors. {An exception} is the descendants of the vacuum AdS$_3$ which admit six globally defined Killing vectors. In that case, the two $U(1)$ Killing charges are $J_\pm = -\frac{1}{4}$ and the other four $\frac{SL(2,\mathbb R)}{U(1)} \times \frac{SL(2,\mathbb R)}{U(1)}$ charges are identically zero. In the case of the decoupled near-horizon extremal phase space defined in section \eqref{extremal-NH-sec} we will have four global Killing vectors with the left-moving $U(1)_+$ charge $J_+$ arbitrary, but the $SL(2,\mathbb R)_-$ charges all vanishing $J^a_-=0$, $a=+1,0,-1$.

\section{Phase space as Virasoro coadjoint orbits}\label{Banados-Orbits-sec}

The symplectic symmetry vectors $\chi$ form an adjoint representation of (centerless) Virasoro algebra. On the other hand, the conserved charges \eg \eqref{cal H def} indeed define an interior product between the vectors $\eps(x)\p_x$ and a form density $L(x)dx^2$ on a circle defined by $x$ ($x$ here can be either $x_+$ or $x_-$). Therefore $L$ actually forms the \textit{coadjoint} representation of Virasoro algebra. Equation \eqref{delta-chi-L} is then the coadjoint action of the Virasoro algebra.
As we discussed, each element of the phase space can be labeled by $L_+(x^+),L_-(x^-)$. Accordingly, the \ads phase space is isomorphic to the direct product of two coadjoint representations of Virasoro algebra. Starting from a \textit{representative} function $L$, \eg $L=const$, the elements obtained by the coadjoint action \eqref{delta-chi-L} form an \textit{orbit} of the Virasoro algebra. More precisely. an orbit is defined by the coadjoint action of Virasoro generators quotiented by the stabilizer subalgebra \eqref{stabilizer}. This is a well-established concept in the literature, see e.g. \cite{Balog:1997zz,Witten:1987ty} and references therein.
What we have shown here is that the phase space is composed of distinct Virasoro coadjoint orbits and that the Killing charges $J_\pm$ are constants along each orbit, and hence label the orbits. In the language of a dual 2d CFT, each orbit may be viewed as a primary operator together with its conformal descendants.


\subsection{Classification of Virasoro orbits}\label{Virasoro-orbits-sec}

Let us summarize some key results from \cite{Balog:1997zz}. We will focus on the orbits of a single Virasoro algebra, say the $+$ sector (which we refer to as left-movers).
The orbits in general fall into two classes: those orbits with a constant representative, and those which do not contain any constant element. The constant $L_+$ representatives correspond to the better studied geometries, e.g. see \cite{Garbarz:2014kaa,Sheikh-Jabbari:2014nya}  for a review. They fall into four categories:
\begin{itemize}
	\item
	\emph{Exceptional orbits} ${\cal E}_n$ with representative $L=-n^2/4,\  n=0,1, 2,3,\cdots$. The  zero massless BTZ is a constant representative of the orbit ${\cal E}_0 \times {\cal E}_0$. (The $n=0$ case coincides with the hyperbolic orbit ${\cal B}_0(0)$, see below.) The ${\cal E}_1 \times {\cal E}_1$ orbit admits global AdS$_3$ as a representative. For $n \geq 2$, ${\cal E}_n \times {\cal E}_n$ is represented by an $n$-fold cover of global AdS$_3$.
	\item \emph{Elliptic orbits} ${\cal C}_{n, \nu}$, with representative $L=-(n+\nu)^2/4$, $n=0,1,\dots$ and $0<\nu<1 $.  The geometries with both $L_\pm$ of this form with $n=0$, correspond to conic spaces which can be regarded as particles on AdS$_3$ \cite{Deser:1983nh}, and geometries in this orbit may be viewed as ``excitations'' (descendants) of particles on the AdS$_3$.
	\item \emph{Hyperbolic orbits} ${\cal B}_0 (b)$, with representative $L=b^2/4$, where $b$ is a real non-negative number $b\geq 0$. The $b=0$ case coincides with the ${\cal E}_0$ orbit. The geometries with both $L_\pm=b_\pm^2/4$ are BTZ black holes. The extremal BTZ corresponds the case where one sector is frozen $b_-=0$ .
	\item \emph{Parabolic orbit} ${\cal P}_0^+$, with representative $L=0$. The geometries associated with ${\cal P}_0^+\times {\cal B}_0 (b)$ orbits correspond to the self-dual orbifold \cite{Coussaert:1994tu} which may also be obtained as the near horizon geometry of extremal BTZ black holes. In particular, ${\cal P}_0^+\times {\cal B}_0 (0)$ corresponds to null selfdual orbifold \cite{Balasubramanian:2003kq}. The
	${\cal P}_0^+\times{\cal P}_0^+$ orbit corresponds to AdS$_3$ in the Poincar\'e patch and its descendants, which in the dual 2d CFT corresponds to vacuum Verma module of the CFT on 2d plane.
	
\end{itemize}
The non-constant representative orbits, come into three categories, the generic \emph{hyperbolic} orbits ${\cal B}_n(b)$ and two \emph{parabolic} orbits ${\cal P}^\pm_n$, $n\in \mathbb{N}$. Geometries associated with these orbits are less clear and understood.


\subsection{ Killing charges as orbit labels }\label{subsec: Killing labels}

As shown in \eqref{Jvar}, all the geometries associated with the same orbit have the same  $J_\pm$ charges. In other words, $J_\pm$ do not vary as we make coordinate transformations using $\chi$ diffeomorphisms \eqref{symplectic symmetries}; $J_\pm$ are ``orbit invariant'' quantities. One may hence relate them with the labels on the orbits, explicitly, $J_+$ should be a function of $b$ or $\nu$ for the hyperbolic or elliptic orbits associated to the left-moving copy of the Virasoro group and $J_-$ a similar function of labels on the right-moving copy of the Virasoro group.

The Ba\~nados phase space has a rich topological structure. It consists of different disjoint patches. Some patches (labeled by only integers) consist of only one orbit, while some consist of a set of orbits with a continuous parameter. On the other hand, note that the conserved charges in covariant phase space methods are defined through an integration of infinitesimal charges along a path connecting a reference point of phase space to a point of interest. Therefore, the charges can be defined only over the piece of phase space simply connected to the reference configuration. For other patches, one should use other reference points. In this work we just present explicit analysis for the  ${\cal B}_0(b_+)\times  {\cal B}_0(b_-)$ sector of the phase space. Since this sector corresponds to the family of BTZ black holes of various mass and angular momentum and their descendants, we call it the BTZ sector. Note that there is no regular coordinate transformation respecting the chosen boundary conditions, which  moves us among the orbits. In particular for the BTZ sector, this means that there is no regular coordinate transformation which relates BTZ black hole geometries with different mass and angular momentum, i.e. geometries with different $b_\pm$.

We now proceed with computing the charges $J_\pm$ for an arbitrary field configuration in the BTZ sector of the phase space. Since the charges are integrable, one can choose any path from a reference configuration to the desired point. We fix the reference configuration to be the massless BTZ with $L_\pm = 0$. We choose the path to pass by the constant representative $L_\pm$ of the desired solution of interest $\tilde L_\pm(x^\pm)$. Let us discuss $J_+$ (the other sector follows the same logic). Then the charge is defined as
\begin{align}
J_+ & =\int_\gamma \delta J_+ =\int_{0}^{\tilde L_+} \delta J_+  = \int_{0}^{L_+} \delta J_+ + \int_{L_+}^{\tilde L_+} \delta J_+.
\end{align}
We decomposed the integral into two parts: first the path \emph{across} the orbits, between constant representatives $L_+ = 0$ and $L_+$ and second the path along (within) a given orbit with representative $L_+$. Since the path along the orbit does not change the values $J_\pm$ ($\delta_\chi J_\pm=0$), the second integral is zero.  Accordingly, the charge is simply given by
\begin{align}
J_+&= \frac{\ell}{8\pi G}\int_{0}^{L_+} \oint d\varphi K_+(L)\de L
\end{align}
where $L_+$ is a constant over the spacetime. Solving \eqref{stabilizer} for constant $L_\pm$ and assuming periodicity of $\phi$, we find that $K_\pm=const$. Therefore the Killing vectors are $\p_\pm$ up to a normalization constant, which we choose to be 1. Hence $K_+(L)=1$, and
\begin{align}\label{defJp}
J_+&=\frac{\ell}{4 G} L_+ ,\qquad J_-=\frac{\ell}{4 G} L_-.
\end{align}
Therefore the Killing charges are a multiple of the Virasoro zero mode of the constant representative.

\subsection{Thermodynamics of Ba\~nados geometries}\label{Thermodynamics-sec}

Since the BTZ descendants are obtained through a finite coordinate transformation from the BTZ black hole, the descendants inherit the causal structure and other geometrical properties of the BTZ black hole. We did not prove that the finite coordinate transformation is non-singular away from the black hole Killing horizon but the fact that the Virasoro charges are defined all the way to the horizon gives us confidence that there is no singularity between the horizon and the spatial boundary. The geometry of the Killing horizon was discussed in more detail in \cite{Sheikh-Jabbari:2014nya}.

The area of the outer horizon defines a geometrical quantity which is invariant under diffeomorphisms. Therefore the BTZ descendants admit the same area along the entire orbit. The angular velocity and surface gravity are defined geometrically as well, given a choice of normalization at infinity. This choice is provided for example by the asymptotic Fefferman-Graham coordinate system which is shared by all BTZ descendants. Therefore these chemical potentials $\tau_\pm$ are also orbit invariant and are identical for all descendants and in particular are constant. This is the zeroth law for the BTZ descendant geometries.

One may define more precisely $\tau_\pm$ as the chemical potentials conjugate to $J_\pm$ \cite{Sheikh-Jabbari:2015nn}. Upon varying the parameters of the solutions we obtain a linearized solution which obey the first law
\be\label{first-law}
\delta S= \tau_+ \delta J_++\tau_- \delta J_- .
\ee
This first law is an immediate consequence of the first law for the BTZ black hole since all quantities are geometrical invariants and therefore independent of the orbit representative. In terms of $L_\pm$, the constant representatives of the orbits in the BTZ sector, one has {\eqref{defJp} and \cite{Kraus:2006wn} }
\be
\tau_\pm=\frac{\pi}{\sqrt{L_\pm}}
\ee
and the entropy takes the usual Cardy form
\bea
S=\frac{\pi}{3} c(\sqrt{L_+}+\sqrt{L_-}) .
\eea
One can also write the Smarr formula in terms of orbit invariants as
\bea
S = 2(\tau_+ J_++\tau_- J_-).
\eea
The only orbits which have a continuous label (necessary to write infinitesimal variations) and which admit a bifurcate Killing horizon are the hyperbolic orbits \cite{Sheikh-Jabbari:2014nya, Sheikh-Jabbari:2015nn}.
The extension of the present discussion to generic hyperbolic orbits (and not just for the BTZ sector) will be postponed to \cite{Sheikh-Jabbari:2015nn}.

\section{Extremal phase space and decoupling limit}\label{extremal-NH-sec}

We define the ``extremal phase space'' as the subspace of the set of all Ba\~nados geometries (equipped with the invariant presymplectic form) with the restriction that the right-moving function $L_-$ vanishes identically. The Killing charge $J_-$ is therefore identically zero. Also, perturbations tangent to the extremal phase space obey $\delta L_- =0$ but $\delta L_+$ is an arbitrary left-moving function.

A particular element in the extremal phase space is the extremal BTZ geometry with $M \ell=J$. It is well-known that this geometry admits a decoupled near-horizon limit which is given by the self-dual spacelike orbifold of AdS$_3$ \cite{Coussaert:1994tu}
\bea\label{sdo}
ds^2 = \frac{\ell^2}{4} \left(-r^2 dt^2 + \frac{dr^2}{r^2} + \frac{4 |J|}{k} (d\phi - \frac{r}{2\sqrt{|J|/k}}dt )^2 \right),\qquad \phi \sim \phi + 2\pi,
\eea
where $k \equiv \frac{\ell}{4G}$. A Virasoro algebra exists as asymptotic symmetry in the near-horizon limit and this Virasoro algebra has been argued to be related to the asymptotic Virasoro algebra defined close to the AdS$_3$ spatial boundary \cite{Balasubramanian:2009bg}. Since these asymptotic symmetries are defined at distinct locations using boundary conditions it is not entirely obvious that they are uniquely related. Now, using the concept of symplectic symmetries which extend the asymptotic symmetries to the bulk spacetime, one deduces that the extremal black holes are equipped with one copy of Virasoro hair. The Virasoro hair transforms under the action of the Virasoro symplectic symmetries, which are also defined everywhere outside of the black hole horizon.

One subtlety is that the near-horizon limit is a decoupling limit obtained after changing coordinates to near-horizon co-moving coordinates. We find two interesting ways to take the near-horizon limit. In Fefferman-Graham coordinates the horizon is sitting at $r=0$ and it has a constant angular velocity $1/\ell$ \emph{independently of the Virasoro hair}. Therefore taking a near-horizon limit is straightforward and one readily obtains the near-horizon Virasoro symmetry. It is amusing that the resulting vector field which generates the symmetry differs from the ansatz in \cite{Balasubramanian:2009bg}, as well as the original Kerr/CFT ansatz \cite{Guica:2008mu} and the newer ansatz for generic extremal black holes \cite{Compere:2014cna, Compere:2015mza}. The difference is however a vector field which is pure gauge, i.e. charges associated with it are zero.

A second interesting way to take the near-horizon limit consists in working with coordinates such that the horizon location depends upon the Virasoro hair. This happens in Gaussian null coordinates. Taking the near-horizon limit then requires more care. This leads to a yet different Virasoro ansatz for the vector field which is field dependent. After working out the details, a chiral half of the Virasoro algebra is again obtained, which also shows the equivalence with the previous limiting procedure.

\subsection{Decoupling limit of the extremal sector} \label{sec decoupling sector}

The general metric of the extremal phase space of AdS$_3$ Einstein gravity with Brown-Henneaux boundary conditions and in the Fefferman-Graham coordinate system is given by
\begin{align}\label{extremal Banados}
ds^2=\dfrac{\ell^2}{r^2}dr^2-r^2dx^+dx^- + \ell^2L(x^+){dx^+}^2,\qquad x^\pm = t/\ell \pm \phi,\qquad \phi \sim \phi+2\pi
\end{align}
{where we dropped the $+$ subscript, $L_+=L$}. It admits two global Killing vectors: $\p_-$ and $\zeta_+$ defined in subsection \eqref{Global-Killings-sec}. In the case of the  extremal BTZ orbit, the metrics \eqref{extremal Banados} admit a Killing horizon at $r=0$ which is generated by the Killing vector $\p_-$ \cite{Sheikh-Jabbari:2014nya}.

One may readily see that a diffeomorphism $\chi(\epsilon_+,\epsilon_-=0)$ defined from \eqref{symplectic symmetries} with  arbitrary $\epsilon_+(x^+)$, namely
\begin{align}\label{chi-ext}
\chi_{ext}=\dfrac{\ell^2\epsilon_+''}{2r^2} \p_-  + \epsilon_+ \p_+ - \half r \epsilon_+' \p_r,
\end{align}
is tangent to the phase space. Indeed, it preserves the form of the metric \eqref{extremal Banados}. Remarkably, the field dependence, i.e. the dependence on $L_+$, completely drops out in $\chi_{ext}$. Note however that although
$\chi_{ext}$ is field independent, the Killing vector $\zeta_+$ is still field dependent.
From previous discussions, it follows straightforwardly that $\chi_{ext}$ generates local symplectic symmetries.

One may then take the decoupling limit
\bea
t \rightarrow \frac{\ell\,\tilde t}{\lambda},\qquad \phi \rightarrow \phi + \Omega_{ext} \frac{\ell\,\tilde t}{\lambda},{\qquad r^2\rightarrow 2\ell^2\, \lambda \tilde r},\qquad \lambda \rightarrow 0
\eea
where $ \Omega_{ext}=-1/\ell$ is the constant angular velocity at extremality. As a result $x^+ \rightarrow \phi$ and {$x^- \rightarrow 2 \frac{\tilde t}{\lambda} - \phi $}. Functions periodic in $x^+$ are hence well-defined in the decoupling limit while functions periodic in $x^-$ are not. Therefore, the full Ba\~nados phase space does not admit a decoupling limit. Only the extremal part of the Ba\~nados phase space does. Also, since $\frac{\tilde t}{\lambda}$ is dominant with respect to $\phi$ in the near-horizon limit, the coordinate $x^-$ effectively decompactifies in the limit while $x^+$ remains periodic. Since $-dx^+ dx^-$ is the metric of the dual CFT, this leads to the interpretation of the decoupling limit as a discrete-light cone quantization of the dual CFT \cite{Balasubramanian:2009bg}.

In this limit the metric \eqref{extremal Banados} and symplectic symmetry generators \eqref{chi-ext} become
\begin{align}
\dfrac{ds^2}{\ell^2}&=\dfrac{d\tilde{r}^2}{4\tilde{r}^2}-{4\tilde{r}} d\tilde td\phi+L(\phi){d\phi}^2\label{NH-metric-Banados}\\
\chi_{ext} &=\dfrac{\epsilon''(\phi)}{8\tilde{r}} \p_{\tilde t} -  \tilde{r} \epsilon'(\phi) \p_{\tilde{r}}  + \epsilon(\phi) \p_\phi \label{NH-chi-Banados},
\end{align}
{where we dropped again the $+$ subscript, $\eps_+=\eps$.} As it is standard in such limits, this geometry acquires an enhanced global $SL(2,\mathbb{R})_- \times U(1)_+$ isometry \cite{Sheikh-Jabbari:2014nya, Li:2013pra}. The $sl(2,\mathbb{R})_-$ Killing vectors are given as
\begin{align}
\xi_1=\dfrac{1}{2}\p_{\tilde t},\;\;\;\;\xi_2=\tilde t \p_{\tilde t} - \tilde{r} \p_{\tilde{r}},\;\;\;\;\xi_3=[({2\tilde t^2} + \dfrac{L}{8\tilde{r}^2})\p_{\tilde t} + \dfrac{1}{2\tilde{r}}\p_\phi - {4\tilde  t \tilde{r}} \p_{\tilde{r}}].
\end{align}
and obey the algebra
\begin{align}\label{sl2R-algebra}
[\xi_1,\xi_2]=\xi_1,\qquad
[\xi_1,\xi_3]=2\xi_2,\qquad
[\xi_2,\xi_3]=\xi_3,
\end{align}
The $u(1)_+$ is still generated by $\zeta_+$.

As it is explicitly seen from the metric \eqref{NH-metric-Banados}, absence of Closed Timelike Curves (CTC) requires $L(\phi) \geq 0$. This restricts the possibilities for orbits which admit a regular decoupling limit. The obvious example is the extremal BTZ orbit for which the decoupling limit is a near-horizon limit. Representatives of these orbits are the extremal BTZ black holes with $L_+ \geq 0$ constant and the near-horizon metric  \eqref{NH-metric-Banados} is precisely the self-dual orbifold \eqref{sdo} after recognizing $J = \frac{\ell}{4G}L = \frac{c}{6}L$ and setting $\tilde t=\sqrt{L_+} t/4$ and $\tilde r=r$.\footnote{For the case of the massless BTZ, one should note that there are two distinct near-horizon limits; the first leads to null self-dual orbifold of AdS$_3$ and the second to the pinching AdS$_3$ orbifold
	\cite{deBoer:2010ac}.}

From the analysis provided in \cite{Balog:1997zz} one can gather that all orbits other than the hyperbolic ${\cal B}_0(b)$  and the parabolic ${\cal P}^+_0$ orbits, admit a function $L(\phi)$ which can take negative values. The corresponding geometries therefore contain CTCs. The only regular decoupling limit is therefore the near-horizon limit of generic extremal BTZ (including massless BTZ \cite{deBoer:2010ac}).  Therefore, the near-horizon extremal phase space is precisely the three-dimensional analogue of the phase space of more generic near-horizon extremal geometries discussed in \cite{Compere:2015bca, Compere:2015mza}. In other words, geometries of the form \eqref{NH-metric-Banados} which are free of CTCs are in  ${\cal P}^+_0\times  {\cal P}^+_0$ or  ${\cal P}^+_0\times {\cal B}_0(b),\ b\geq 0 $ orbits.

Under the action of $\chi_{ext}$ above, one has
\begin{align}
\delta_\chi L(\phi)= \epsilon L(\phi)' + 2 L(\phi) \epsilon' - \half \epsilon'''
\end{align}
in the decoupling limit. With the mode expansion $\epsilon=e^{i n \phi}$, one may define the symplectic symmetry generators $l_n$ which satisfy the Witt algebra,
\begin{align}
i[l_m,l_n]=(m-n)l_{m+n}.
\end{align}
The surface charge is integrable and given by
\be
H_\chi[\Phi]=\dfrac{\ell}{8 \pi G}\oint d\phi \, \epsilon(\phi) L(\phi).
\ee
Moreover, one may show that the surface charges associated to the $SL(2,\mathbb{R})_-$ Killing vectors, $J_-^a$, vanish. Interestingly, we find that the $\tilde t$ and $\tilde{r}$ components of $\chi_{ext}$ \eqref{NH-chi-Banados} do not contribute to the surface charges. The various ansatzes described in \cite{Balasubramanian:2009bg,Guica:2008mu,Compere:2014cna, Compere:2015mza} which differ precisely by the $\p_{\tilde t}$ term are therefore physically equivalent to the one in  \eqref{NH-chi-Banados}.

One may also work out the algebra of charges $H_n$ associated with $\epsilon=e^{i n \phi}$:
\be
\{H_m,H_n\}=(m-n)H_{m+n}+\dfrac{c}{12}m^3\delta_{m+n,0},
\ee
where $c$ is the usual Brown-Henneaux central charge.

The charge $J_+$ associated with the Killing vector $\zeta_+$ commutes with the $H_n$'s, as discussed in general in section \eqref{Jpm-sec}. Following the analysis of section \eqref{Thermodynamics-sec}, one may associate an entropy $S$ and chemical potential $\tau_+$ which satisfy the first law and Smarr relation
\be\label{first-law-Smarr-extremal}
\delta S=\tau_+ \delta J_+\,,\qquad S=2\tau_+ J_+.
\ee

\section{Microscopic counting of entropy}
We showed that corresponding to a BTZ black hole, there is a family of solutions obtained by specific nontrivial coordinate transformations of BTZ. We called these solutions, the \textit{descendants} of BTZ. What we naively have in mind, is that this family of solutions can be mapped to \textit{states} in the Fock space of a conformal field theory. Then BTZ is represented by a \textit{primary} field in CFT and its entropy is obtained by a counting on the number of its descendants. Although the number of descendants is infinite, the result can regularized by defining an appropriate measure or a suitable cutoff. We postpone this analysis to future works.

Let us briefly review a closely related computation in the literature (see \eg \cite{Banados:1998gg}) which is qualitatively different from the one presented above. In this picture, a BTZ black hole is considered as a thermal ensemble in the dual field theory and the vacuum corresponds to a massless BTZ black hole. The descendants are  obtained by acting Virasoro generators $L_{-n}$ (with negative index) with the energy
\begin{align}
L_0|n_1,\cdots n_r\rangle &=\Delta |n_1,\cdots n_r\rangle \,,\qquad \Delta=\sum_{i=1}^{r}n_i\,.
\end{align}
Therefore the density of states with energy $\Delta$ is equal to the number of possible partitions of $\Delta$, a well known problem in combinatorics. BTZ is then considered as the ensemble of states with energy $\Delta=L_0$. For large values of $L_0,\bar{L}_0$, the number of states is approximated by Ramanujan formula
\begin{align}
\rho(L_0,\bar{L}_0)&=\exp\left(2\pi\sqrt{L_0/6}+2\pi\sqrt{\overline{L}_0/6}\right)\,.
\end{align}
However, this computation leads to 
\begin{align}
S&=\log \rho = c^{-1/2}\dfrac{A}{4G}
\end{align}
where $c=3\ell/2G$ is the central charge of the Virasoro algebra. As we see, this computation fails to reproduce the entropy of BTZ black hole. Despite from this fact, we think that this does not present the correct holographic map. The reason is that action of charges generates an infinitesimal diffeomorphism in the bulk and hence the  geometry is moved in its orbit. Accordingly, the state $|n_1,\cdots n_r\rangle$ represents a geometry in the orbit of the massless BTZ, while we know that a BTZ black hole with nonvanishing mass and angular momentum is in another orbit. In other words, there is no diffeomorphism mappling a BTZ black hole to a different one.

The above proposals are based on the identification of microscopic states. However, there is a successful microstate \textit{counting} without trying to identify them. In this approach, one assumes that the symmetry algebra is not the basic algebra of underlying fields, but only a representation thereof. Then based on the assumption that the dual CFT$_2$ is modular invariant, and the Cardy formula is applicable, one can show that starting from the partition function, one can reproduce exactly the entropy of BTZ black hole \cite{Strominger:1997eq,Carlip:1998wz}. The partition function is given by
\begin{align}
Z(\tau)&=\mathrm{Tr}\left[ \exp\left(2\pi i \tau L_0-2\pi i \bar{\tau} \bar{L}_0\right)\right]
\end{align}
where $\tau,\bar{\tau}$ are modular parameters given by a combination of horizon temperature and angular velocity. Modular invariance implies
\begin{align}
Z(\tau')&=Z(\tau)\,,\qquad \tau'=\dfrac{a\tau+b}{c\tau+d}
\end{align}
for integer values of $a,b,c,d$. Then it follows that the number of states is given by \cite{Cardy:1986ie}
\begin{align}
\rho(L_0,\bar{L}_0)&=\exp\left(2\pi\sqrt{c L_0/6}+2\pi\sqrt{c\overline{L}_0/6}\right)
\end{align}
correctly reproducing the entropy of BTZ black hole.

\section{Discussion}
In this chapter we studied the set of asymptotic \ads spacetimes with flat boundary metric. We showed that this collection considered as a manifold, can be equipped with a symplectic structure. Therefore we obtained the \ads phase space. We showed that the Brown Henneaux asymptotic symmetries can be extended into the bulk such that they transform a configuration of the phase space into another nearby configuration. We demonstrated that these symmetry transformations satisfy in addition properties that makes them ``local symplectic symmetries'' (as explained in proposition \eqref{prop local symplectic symmetries}) of the \ads phase space. The exponentiation of these infinitesimal symmetry transformations generates one parameter family of configurations of the phase space. All configurations connected to a \textit{representative} configuration forms an \textit{orbit} of the Virasoro algebra. Not all configurations in the \ads phase space are related by a diffeomorphism, hence the \ads phase space breaks into different orbits of the Virasoro algebra. This gives a good understanding of the \ads phase space. The manifold is not simply connected, but it consists of many connected patches. We also investigated the extremal sector of the phase space and studied its decoupling limit, which leads to a phase space of its own that corresponds to a boundary flat metric with an identification on a null circle. We showed that the generator of symmetries of the near horizon region can be read off from the Brown Henneaux asymptotic symmetries extended all the way down to the horizon by requiring them to be local symplectic symmetries.

\chapter{Near horizon geometries of extremal black holes}\label{chapter-NHEG}
\section{Introduction}
A black hole with zero Hawking temperature is called an extremal black hole for which the two inner and outer horizons coincide. Meanwhile they have non vanishing horizon area and hence entropy. By the third law of black hole mechanics, an extremal black hole cannot be formed through a physical process (like a collapse) in a finite time. Therefore one might expect that these are not interesting objects to study. However from the theoretical point of view, it can be argued that in a theory of quantum gravity, an extremal black hole can be regarded as a ``frozen state'' in which the excitations related to Hawking radiation are not present. The fact that the entropy of an extremal black hole is non vanishing, then implies that the microscopic theory has many degenerate states corresponding to the extremal black hole. If one can understand the microscopics of extremal black holes, then it would probably open the way to understand the resolution for generic non-extremal black holes. 

The near horizon geometry of an extremal black hole (to be defined precisely) has many universal and interesting properties. independent of theory and dimensions. The study of these geometries (also known as \textit{throat geometries}) has different advantages. When the thermodynamic properties of black holes are concerned, it can be shown that almost all the information is encoded in the near horizon geometry. Also different approaches show that the entropy of a black hole resides on or at least near its horizon. Therefore it is proposed that for a statistical description of black hole entropy, it is enough to study the throat geometry \cite{Strominger:1997eq,Guica:2008mu,Compere:2015bca,Compere:2015mza}. It is expected that the universal properties of near horizon geometries can reveal a universal statistical description of extremal black hole entropy. This will be the context of the next chapter.

In a different context, the classification of near horizon geometries is a first step towards the classification of black hole geometries in higher dimensions. In 4 dimensions, there are uniqueness theorems that answer the classification of black holes in Einstein Maxwell theory( see \cite{Chrusciel:2012jk} for a review). However in 5 dimensions, within pure Einstein gravity, there are 2 known black hole solutions which are asymptotically flat; namely the Myers-Perry black hole, and the black rings. Spatial sections of the horizon in the former case have spherical topology, while the latter is $S^2\times S^1$. Although there are constraints on the topology and symmetries of black holes in 5 and higher dimensions, the classification of black holes in dimensions higher than four remains an open problem. Therefore a thorough study of near horizon geometries can be helpful in this direction as well \cite{Kunduri:2007vf,Kunduri:2008rs,Kunduri:2013ana,Hollands:2009ng}.

In this chapter, we will switch to study the problem of black hole entropy in arbitrary dimensions $d\geq 4$. The cost is that we need to restrict to a special but interesting type of black holes, namely the \textit{extremal} black holes. The strategy is to focus on the geometry near the horizon of the black hole. We will show that zooming on the horizon reveals an independent geometry which contains all the relevant information, and is called the \textit{near horizon geometry of extremal black hole} (NHEG). These are interesting geometries having enhanced isometries compared to the black hole geometry, and hence more manageable dynamics. 

In this chapter, we first introduce these geometries and study their geometrical properties. Specially we show that they have bifurcate Killing horizon structures which is generated by a specific Killing vector of the spacetime. This is interesting since the original extremal black hole geometry does not possess bifurcate Killing horizon. Next we study the conserved charges corresponding to the global isometries of the geometry and show that they satisfy laws similar to that of black hole mechanics. 
In section \ref{Kerr CFT}, we study the Kerr/CFT correspondence \cite{Guica:2008mu} that proposed a duality between the dynamics over the near horizon geometry of extremal Kerr (NHEK) and a dual chiral CFT, and its extensions. We then discuss the problems and challenges in this proposal. Specifically we discuss the perturbations over NHEG and show the so called `No dynamics'' property of NHEK \cite{Amsel:2009ev,}that extends to other NHEGs by physical assumptions \cite{Hajian:2014twa}. Based on these results, we will construct the ``NHEG phase space'' in the next chapter.

\section{Extremal black holes and near horizon limit}
The most general form of the metric of a stationary and axisymmetric black hole can be written in the ADM form as
\begin{align}\label{BH geometry}
ds^2&=-f d\tau^2+g_{\rho\rho}d\rho^2+\tilde g_{\alpha\beta}d\theta^\alpha d\theta^\beta +g_{ij}(d\psi^i-\omega^i d\tau)(d\psi^j-\omega^j d\tau)\,,\cr
\end{align}
where $f, g_{\rho\rho}, \tilde g_{\alpha\beta}, g_{ij}, \omega^i$ are functions of $\rho,\theta^\alpha$ and $i,j=1,2,\cdots, n$ and $p=n+1,\cdots, N$. The horizons of black hole are at the roots of $g^{\rho\rho}$,
\begin{align}\label{Delta}
g_{\rho\rho}=\dfrac{1}{D^2(\rho,\theta^\alpha)\Delta(\rho)}\,,\qquad \Delta=\prod_m (\rho-r_m)\times (\rho-r_+)^2\,,
\end{align}
where we assume the function $D$ to be analytic and nonvanishing everywhere. $r_+$ is the radius of  the outer horizon and is a degenerate root of $\Delta$. This implies that the black hole is an extremal black hole. In four dimensions the black hole has at most two horizons (e.g. see \cite{Chrusciel:2012jk}) and $\Delta=(\rho-r_+)^2$. Due to the smoothness of metric on the horizons $f$ can always be written in the following form:
\begin{align}
f=C^2(\rho,\theta)\Delta(\rho)\,.
\end{align}
and the fact that surface gravity is constant on the horizon implies that $CD=const$ on the horizon. 
\subsection{Near horizon limit of extremal black holes}\label{sec6.1-NH-limit}
To take the near horizon limit let us first make the coordinate transformations \cite{Bardeen:1999px}
\begin{align}
\rho &=r_e(1+\lambda r)\,,\qquad \tau =\dfrac{\alpha r_e t}{\lambda}\,,\qquad
\varphi^i =\psi^i-\Omega^i \tau\,
\end{align}
where $\Omega^i=\omega^i(r_e)$ is the horizon angular velocity. In the first equation we scale $\rho-r_e$ by the constant $\lambda$ to zoom on the horizon. Then in second equation we scale $\tau$ inversely to cancel divergences appearing in the transformed metric, and $\alpha$ is a suitable constant to get the most simple form for the near horizon metric. The shift in $\psi^i$ takes us to the frame co-rotating with the black hole. Finally we take the limit $\lambda\to 0$.  The set of these transformations and limit is called the near horizon limit, and the resulting near horizon geometry becomes
\begin{align}
ds^2
=\dfrac{1}{D^2}\left[- r^2dt^2+\dfrac{dr^2}{r^2}+D^2\tilde g_{\alpha\beta}d\theta^\alpha d\theta^\beta
+D^2g_{ij}(d\varphi^i+(\Omega^i-\omega^i) d\tau)(d\varphi^j+(\Omega^j-\omega^j) d\tau)\right]\,,
\end{align}
where we used the fact that $CD=const$ on the horizon and chose
\begin{align}
\alpha r_e^2=\dfrac{1}{CD}\,.
\end{align}
Recalling that $\Omega^i=\omega^i\vert_{r_e}$, we arrive at the general form:
\begin{align}\label{NH-geometry}
ds^2&=\Gamma\left[-r^2dt^2+\dfrac{dr^2}{r^2}+g_{\alpha\beta}d\theta^\alpha d\theta^\beta
+\gamma_{ij}(d\varphi^i +k^irdt)(d\varphi^j +k^j rdt)\right]\,,
\end{align}
in which
\begin{align}\label{parameters}
 \Gamma(\theta^\alpha) =\dfrac{1}{D^{2}}\bigg|_{\rho=r_e},\quad\gamma_{ij}(\theta^\alpha) =D^2 g_{ij}\bigg|_{\rho=r_e}\,,\quad
k^i =-\dfrac{1}{CD}\dfrac{\partial \omega^i}{\partial \rho}\bigg|_{\rho =r_e}\,.
\end{align}
It can be shown \cite{Hajian:2013lna} that smoothness of black hole geometry \eqref{BH geometry} forces $\partial_\rho\omega^i$ to be constant on the horizon, and $k^i$ are hence constants in the NHEG (see also the appendix of \cite{Figueras:2008qh} for a detailed proof). The near horizon limit of an extremal black hole is a solution to the theory in which the original black hole was a solution. In this sense, this can be called a \textit{decoupling limit}. Note that taking the near horizon limit of a nonextremal black hole does not necessarily lead to a solution. For example the near horizon limit of a Schwarzschild black hole (as a solution to vacuum Einstein gravity) is a product of a Rindler  space and a sphere which is clearly not a solution to vacuum Einstein gravity.

\section{Quick Review on NHEG}\label{sec-NHEG-review}

The near horizon extremal geometries (NHEG) are generic classes of geometries with at least $SL(2,\mathbb{R}) \times U(1)$ isometry group. These geometries, as the name suggests, may appear in the near horizon limit of extremal black holes, while they may also be viewed as independent classes of geometries. Here we will mainly adopt the latter viewpoint. Also for concreteness and technical simplicity, we will focus on a special class of the NHEG which are Einstein vacuum solutions in generic $d$ dimensions with $SL(2,\mathbb{R})\times U(1)^{d-3}$ isometry group. The general metric for this class of NHEG is 
\begin{align}\label{NHEG-metric}
	{ds}^2&=\Gamma(\theta)\left[-r^2dt^2+\frac{dr^2}{r^2}+d\theta^2+\sum_{i,j=1}^{d-3}\gamma_{ij}(\theta)(d\varphi^i+k^irdt)(d\varphi^j+k^jrdt)\right]
\end{align}
where
\begin{align}\label{ranges}
t\in (-\infty,+\infty),\qquad r\in \{r<0 \} \text{ or } \{r>0\},\qquad \theta\in[0,\theta_{Max} ], \qquad \varphi^i\sim \varphi^i+2\pi,
\end{align}
and $k^i$ are given constants. The geometry is a warped fibred product over an AdS$_2$ factor, spanned by $t,r$, with a Euclidean smooth and compact codimension two surface $\mathcal{H}$, covered by $\theta,\varphi^i$; i.e. $\mathcal{H}$ are constant $t,r$ surfaces. Notably, due to the $SL(2,\mathbb{R})$ isometry of the background, constant $t=t_{\mathcal H},r=r_{\mathcal H}$ surfaces for any value of $t_{\mathcal H}, r_{\mathcal H}$, all give isometric surfaces $\mathcal{H}$.

\begin{figure}[htb]	
	\captionsetup{width=.8\textwidth}
	\centering
		\centering
		\begin{tikzpicture}[scale=0.8]
		\path
		( -2,0)  coordinate (I)
		(2,-4) coordinate (II)
		(2, 4) coordinate (III)
		(-2,8)  coordinate (IV)
		(2,8) coordinate (V)
		(-2,-4) coordinate (VI)
		;
		\draw (IV) --node[midway, below, sloped] {$r=-\infty$} (I);
		\draw (V) -- (III);
		\draw [ultra thick, draw=black, fill=gray!20] (I) -- node[midway, above, sloped] {$r=0$} (II) -- node[midway, above, sloped] {$r=\infty$} (III) --  node[midway, below, sloped] {$r=0$} (I);
		\draw (III) -- (IV);
		\draw (I) -- (VI);
		\node at (.5,0) {${\mathbf{I}}$};
		\node at (-0.5,4) {${\mathbf{II}}$};
		\draw[>=stealth,->] (-2,5.1)--(-2,5);
		\draw[>=stealth,->] (-2,1.1)--(-2,1);
		\draw[>=stealth,->] (-2,-2.9)--(-2,-3);
		\draw[>=stealth,->] (2,4.9)--(2,5);
		\draw[>=stealth,->,ultra thick] (2,0.9)--(2,1);
		\draw[>=stealth,->, ultra thick] (2,-3.1)--(2,-3);		
		\end{tikzpicture}
	\caption[Penrose diagram of NHEG geometry]{Penrose diagram for NHEG, suppressing the $\theta,\varphi^i$ directions. The positive and negative $r$ values of the coordinates used in \eqref{NHEG-metric} respectively cover ${\mathbf{I}}$ and ${\mathbf{II}}$ regions in the above figure. The two boundaries are mapped onto each other by an $r$--$\vec{\varphi}$ inversion symmetry \eqref{r-phi-inversion}. {The arrows on the boundaries shows the flow of time $t$. Note also that flow of time  is reversed between regions I and II.}}	\label{fig:penrose diagram}
\end{figure}
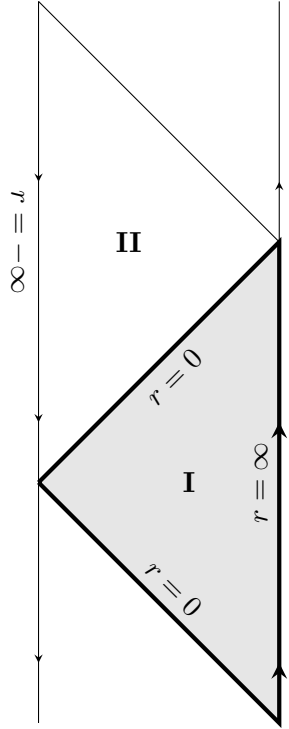

Let us discuss the Penrose diagram of the geometries by ignoring the directions on the compact surface $\cH$. The first two terms of the metric \eqref{NHEG-metric} form an AdS$_2$ in the Poincar\'{e} patch; $r=0$ is the Poincar\'{e} horizon. This coordinate system covers a triangle in the penrose diagram of global AdS$_2$. The metric however can be extended beyond the horizon by allowing negative values for $r$. Then $r=+\infty,-\infty$ cover segments of the two disjoint boundaries of AdS$_2$ as depicted in figure \ref{fig:penrose diagram}(see also \cite{Aharony:1999ti}). 
The range of the $\theta$ coordinate is fixed requiring that $\mathcal{H}$ is a smooth and compact manifold. Note that $\mathcal{H}$ can take various topologies \cite{Hollands:2009ng}. Requiring the geometry to be smooth and Lorentzian implies  $\Gamma(\theta)>0$ and  the eigenvalues of $\gamma_{ij}$ to be real and nonnegative. Moreover, smoothness and absence of conical singularity of $\mathcal{H}$  implies that: (1) At most one of the eigenvalues of $\gamma_{ij}(\theta)$ matrix can be vanishing around a given $\theta=\theta_0$ coordinate; (2) if at $\theta_0$ we have a vanishing eigenvalue, it should behaves as $(\theta-\theta_0)^2+\mathcal{O}(\theta-\theta_0)^3$. { Note that the coefficient of $(\theta-\theta_0)^2$ should be exactly one.}

The geometry is completely determined by the functions $\Gamma(\theta), \gamma_{ij}(\theta)$ and the $d-3$ constants $k^i$ which are determined through the Einstein field equations. A complete classification of vacuum near horizon geometries was given in \cite{Hollands:2009ng}. It was shown that there are $(d-2)(d-3)/2$ independent continuous parameters and two discrete parameters that specify a given NHEG. The discrete parameters specify the topology of compact surface $\cH$ which can be either $S^2 \times T^{d-4}$, or $S^3 \times T^{d-5}$, or quotients thereof and $L(p,q) \times T^{d-5}$ where $L(p,q)$ is a Lens space. In four dimensions, there is only one continuous parameter which is the entropy or angular momentum (remember that $k=1$ in that case) and the topology is $S^2$. In five dimensions, there are three possible topologies $S^2 \times S^1$, $S^3$ and $L(p,q)$ and three continuous parameters. 

\paragraph{NHEG isometries.} The NHEG background \eqref{NHEG-metric} enjoys $SL(2,\mathbb{R})\times U(1)^{d-3}$ isometry. 
The  $SL(2,\mathbb{R})$ isometries are generated by Killing vectors $\xi_a$ with $ a\in \{-,0,+\}$, 
\be\label{xi1-xi2}
\begin{split}
\xi_- &=\partial_t\,,\qquad \xi_0=t\partial_t-r\partial_r,\qquad	\xi_+ =\dfrac{1}{2}(t^2+\frac{1}{r^2})\partial_t-tr\partial_r-\frac{1}{r}{k}^i{\p}_{\varphi^i},
\end{split}
\ee
and the  $U(1)^{d-3}\;$ isometries  by Killing vectors $\mathrm{m}_i$ with $i\in \{1,\cdots, d-3 \}$, 
\begin{align}\label{U(1)-generators}
\mathrm{m}_i=\p_{\varphi^i}.
\end{align}
The isometry algebra is then
\begin{align}\label{commutation relation}
[\xi_0,\xi_-]=-\xi_-,\qquad [\xi_0,\xi_+]=\xi_+, \qquad [\xi_-,\xi_+]=\xi_0\,,\qquad [\xi_a,\mathrm{m}_i]=0.
\end{align}
That is, if we view $\xi_0$ as the scaling operator, $\xi_-,\xi_+$ are respectively lowering and raising operators in $SL(2,\mathbb{R})$. We also note that $\xi_-,\xi_0$ form a two dimensional subalgebra of $SL(2,\mathbb{R})$. For further use we define the structure constants $f_{ab}^{\;\;\; c}$ from $[\xi_a,\xi_b] = f_{ab}^{\;\; \;c}\xi_c$. 

\textbf{Notations:} \emph{Hereafter, we will denote all the $d-3$ indices by vector sign; e.g. $k^i$ will be denoted by $\vec{k}$, $\varphi^i$ by $\vec{\varphi}$, ${\p}_{\varphi^i}$ by $\vec{\partial}_{\varphi}$  and when there is a summation over $i$-indices it will be denoted by dot-product; e.g. ${k}^i{\p}_{\varphi^i}=\vec{k}\cdot\vec{\partial}_{\varphi}=\vec{k}\cdot\vec{\mathrm{m}}$.}

The NHEG also enjoys various $Z_2$ { isometries}. The two which will be relevant for our later analysis are $r$--$\vec{\varphi}$ and $t$--$\vec{\varphi}$-inversions. The $t$--$\vec{\varphi}$-inversion, 
\be\label{t-phi-inversion}
(t,\varphi^i) \ \ \to\ \ (-t,-\varphi^i).
\ee
is reminiscent of the similar symmetry in the (extremal) black hole (see \cite{Schiffrin:2015yua} for a recent discussion) whose near horizon limit leads to the NHEG. One may readily check that under the above $Z_2$, 
$\xi_0$ do not change while $\xi_-, \xi_+, \vec{\mathrm{m}}$ change sign.
Another $Z_2$ isometry is the $r$--$\vec{\varphi}$-inversion,
\be\label{r-phi-inversion}
(r,\varphi^i) \ \ \to\ \ (-r,-\varphi^i).
\ee
This $Z_2$ exchanges the two boundaries of AdS$_2$ (\emph{cf.} Fig. \ref{fig:penrose diagram}). Under the $r$--$\vec{\varphi}$-inversion \eqref{r-phi-inversion}, the $SL(2,\mathbb{R})$ Killing vectors \eqref{xi1-xi2} remain invariant.

\paragraph{NHEG examples in $4d$ and $5d$.} As some examples of NHEG, let us consider the near horizon geometry of extremal Kerr black hole (NHEK) in four dimensions \cite{Bardeen:1999px} and extremal Myers-Perry black hole in five dimensions \cite{Kunduri:2007vf,Kunduri:2008rs}. For NHEK we have
\begin{align}\label{NHEK-solution}
\Gamma = J\frac{1+\cos^2\theta}{2},\qquad \gamma_{11}=\left(\dfrac{2\sin\theta}{1+\cos^2\theta}\right)^2,\qquad k=1,
\end{align}
where $J$ is a constant equal to the angular momentum of the geometry and also the angular momentum of the corresponding black hole. The range of polar coordinate is $\theta\in [0,\pi]$. Near the roots of $\gamma_{11}$ which occur at $\theta=0,\pi$, it clearly satisfies the smoothness condition and the compact surface $\mathcal{H}$, whose area is $4\pi J$, is topologically a two-sphere. 

For the $5d$ extremal Myers-Perry near-horizon geometry we have
\begin{equation}
\begin{split}
\Gamma &= \frac{1}{4} (a+b) (a\cos^2\frac{\theta}{2}+b\sin^2\frac{\theta}{2}),\qquad k^1=\frac{1}{2}\sqrt{\frac{b}{a}},\qquad k^2=\frac{1}{2}\sqrt{\frac{a}{b}},\\\\
\gamma_{ij}&=\dfrac{4}{(a\cos^2\frac{\theta}{2}+b\sin^2\frac{\theta}{2})^2}
\begin{pmatrix}
 a (a+b\sin^2\frac{\theta}{2}) \sin ^2\frac{\theta}{2} &  a b \cos ^2\frac{\theta}{2} \sin ^2\frac{\theta}{2} \\ \ \ & \ \ \\
 a b \cos ^2\frac{\theta}{2} \sin ^2\frac{\theta}{2} &  b \cos ^2\frac{\theta}{2} (b+a\cos^2\frac{\theta}{2}) \\
\end{pmatrix},
\end{split}
\end{equation}
where $a>0,b>0$ are constants related to the angular momenta, and $\theta\in [0,\pi]$. Note that $k^1k^2=\frac14$ and hence $k^1$ and $k^2$ are not independent. One can compute the eigenvalues $\lambda_{1,2}(\theta)$ of the matrix $\gamma_{ij}$. Then we observe that one of the eigenvalues is always positive, while the other eigenvalue (say $\lambda_2$) vanishes at $\theta=0,\pi$. Near these poles we find
\begin{align}
\lambda_2=\theta^2+\mathcal{O}(\theta^3),\qquad \lambda_2=(\pi-\theta)^2+\mathcal{O}((\pi-\theta)^3)
\end{align}
satisfying the regularity condition. The surface $\mathcal{H}$ is topologically $S^3$ and it is area is $2\pi^2\cdot \sqrt{ab}(a+b)^2$.

\subsection{Bifurcate Killing horizons}\label{sec-Killing-horizon}

The Petrov classification has been extended to higher dimensions \cite{Coley:2004jv}. NHEG is a Petrov type D spacetime \cite{Godazgar:2009fi}. It has two real principal null directions which turn out to be congruences of torsion,  expansion and twist free geodesics \cite{Durkee:2010ea}. They are generated by
\begin{align}\label{ells}
\begin{split}
\ell_+&=\left(\dfrac{1}{r}\partial_t +r\partial_r-\vec{k}\cdot\vec{\partial}_{\varphi}\right)\,, \\
\ell_-&=\left(\dfrac{1}{r}\partial_t -r\partial_r-\vec{k}\cdot\vec{\partial}_{\varphi}\right)\,.
\end{split}
\end{align}
These vector fields are respectively normal to the hypersurfaces, 
\begin{equation}
\begin{split}
\mathcal{N}_+ \; : \quad v\equiv t+\dfrac{1}{r}= const\equiv t_{\mathcal H}+\dfrac{1}{r_{\mathcal H}}= v_{\mathcal H}\, ,\\
\mathcal{N}_- \; :\quad  u\equiv t-\dfrac{1}{r}= const \equiv t_{\mathcal H}-\dfrac{1}{r_{\mathcal H}}= u_{\mathcal H}\,.
\end{split}
\end{equation}
One may readily see that $\ell_+\cdot dv=\ell_-\cdot du=0$ and that $\mathcal{N}_\pm$ are therefore null hypersurfaces. Intersection of these two hypersurfaces is a $d-2$ dimensional compact surface $\mathcal{H}$, identified by $t=t_{\mathcal H}, r=r_{\mathcal H}$. Note that both $\ell_\pm$ are  normal to $\mathcal{H}$ and its binormal tensor is 
\begin{align}\label{binormal} 
\boldsymbol{\eps}_\perp=\Gamma dt \wedge dr =\frac{\Gamma}{2}  r^2d v \wedge du,
\end{align}
normalized such that $\eps^\perp_\mn \eps_\perp^\mn=-2$. We note that under the $t$--$\vec{\varphi}$-inversion or $r$--$\vec{\varphi}$-inversion symmetries \eqref{t-phi-inversion}-\eqref{r-phi-inversion}, $\ell_\pm\leftrightarrow -\ell_{\mp}$.

In what follows we prove that each surface $t_{\mathcal H}, r_{\mathcal H}$ is the bifurcation point of bifurcate Killing horizon \cite{Hajian:2013lna,Hajian:2014twa}.  (Similar arguments can be found in \cite{Bengtsson:2005zj} for warped $AdS_3$ geometries.) 

\paragraph{Killing Horizon Generator.} By definition, a Killing horizon $\mathcal{N}$ is a null hypersurface  generated by a Killing vector field $\zeta$, provided that the vector $\zeta$ is normal to $\mathcal{N}$. 

Let us now consider the Killing vector $\zeta_{ H}$ \cite{Hajian:2014twa}
\begin{align}\label{zeta-H}
\zeta_{ H}&
=n_{ H}^a\xi_a - \vec{k}\cdot\vec{\mathrm{m}},
\end{align}
where $n_{H}^a$ are given by the following functions computed at the constant values $t=t_{\mathcal H},\,r=r_{\mathcal H}$ 
\begin{align}\label{n-a-r-t}
n^-=-\frac{t^2r^2-1}{2r}\,, \qquad n^0=t\,r\,,\qquad n^+=-r.
\end{align}
We will discuss the derivation and interesting properties of this vector in appendix \ref{sec vec na}. What is crucial here is that
\begin{align}
\zeta_{H}\big{|}_{\mathcal{N}_\pm}=\frac{r-r_{ H}}{r} \ell_\pm.
\end{align}
Since $\ell_\pm$ are null vectors normal to ${\cal N}_\pm$, the hypersurface $\mathcal{N}=\mathcal{N}_+ \cup \mathcal{N}_-$  is the ``Killing horizon" of $\zeta_{H}$, and $\mathcal{H}$ is its bifurcation surface. Note also that $\zeta_{H}$ vanishes at the bifurcation surface $\mathcal{H}$. Therefore, $\mathcal{N}$ . The choice of $t_{H},r_{H}$ is arbitrary in the above argument, so there are \textit{infinitely many} Killing horizons, bifurcating at any compact surface determined by $t_{ H},r_{H}$.

\begin{figure}[!h]
\centering
\includegraphics[scale=0.36, angle=45]{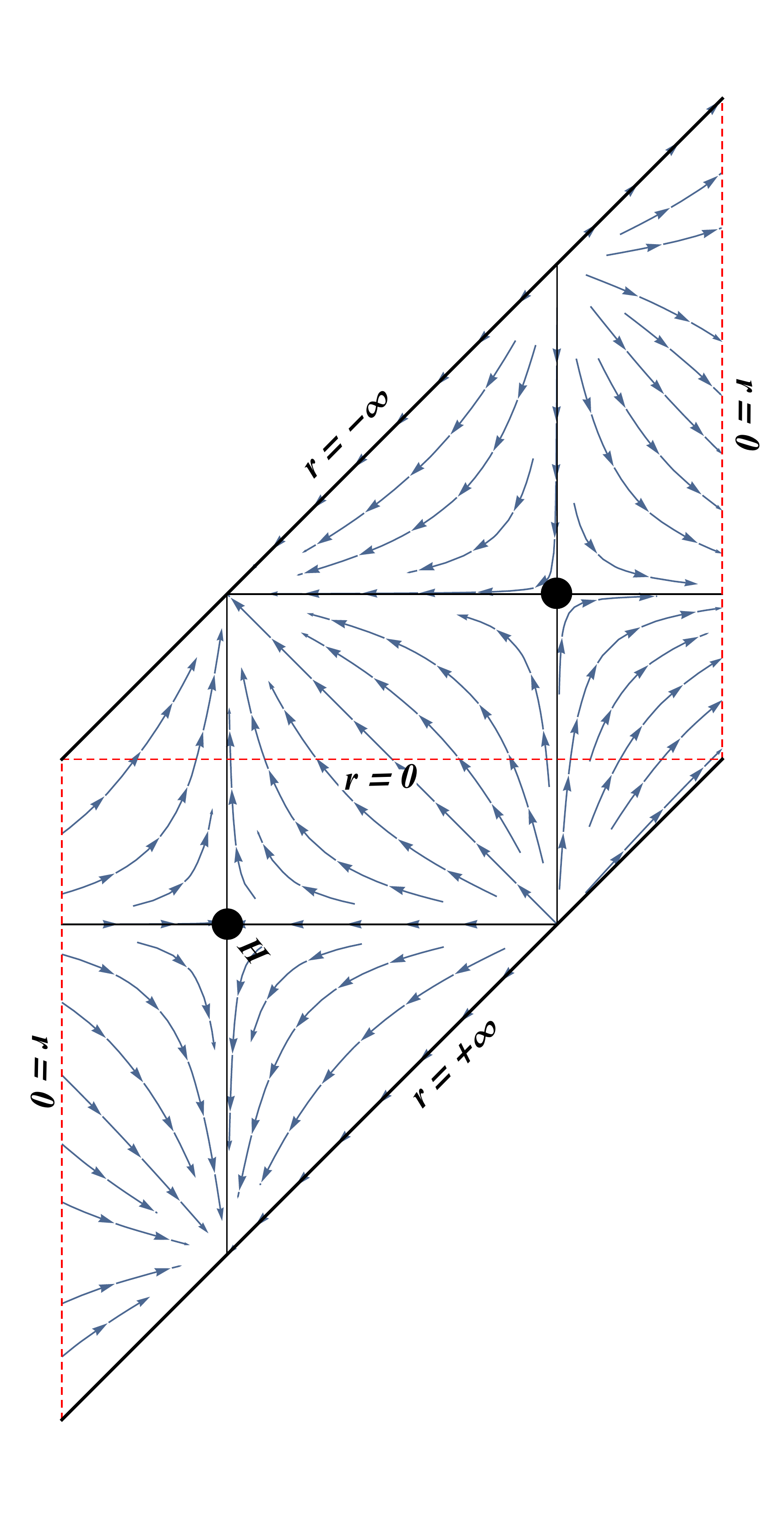}
\caption[Bifurcate Killing horizon in NHEG]{The flow of vector $\zeta_H$ generates a bifurcate Killing horizon.}
\label{fig:Flow}
\end{figure}

It is important to note that although the extremal black hole does not possess any bifurcate Killing horizon, the corresponding near horizon geometry has an infinite number of them. This has to do with the enhancement of symmetries in the near horizon geometry. Note that in construction of the vector $\zeta_{\mathcal H}$ we used \textit{all} the Killing vectors of the geometry.

Another important feature about the vector $\zeta_{H}$ is that on $\mathcal{H}$, 
\bea\label{propH}
\nabla_{[\mu}{\zeta_{ H}}_{\nu]}=\eps^\perp_{\mu\nu} 
\eea
where $\boldsymbol\eps_\perp$ is the binormal tensor \eqref{binormal}. We can use this fact to compute the \textit{surface gravity} on the bifurcation surface of the Killing horizon:
\begin{align}\label{kappa-NHEG}
\kappa^2&=-\dfrac{1}{2}|\nabla{\zeta_{H}}|^2=1.
\end{align}
The above gives the value of $\kappa^2$. As in the usual black hole cases, $\zeta_{H}$ is the generator of a bifurcate Killing horizon with future and past oriented branches.\footnote{In the black hole terminology, the future (past) oriented branch of horizons corresponds to the black (white) hole. However here there is no event horizon.} One can then show that the value of $\kappa$ is $+1$ for the future oriented branch and $-1$ for the past oriented branch.
As  a consequence of $SL(2,\mathbb{R})$ invariance the surface gravity is a constant and independent of $t_{H}$ and $r_{H}$. As in the Rindler space, one can associate an Unruh-type temperature \cite{Unruh:1976db} to the Killing horizons. This temperature is simply $\frac{\hbar}{2\pi}$ and constant over $\mathcal{H}$. We call this the \textit{zeroth law of NHEG}.

\section{Global charges and first law of NHEG mechanics}
We showed that the set of NHEG geometries have the isometry group \sltruod. Also we showed that NHEG possess Bifurcate Killing horizons bifurcating at any radius (except $r=0$ which is a degenerate horizon) and hence one can associate them  a well defined entropy which is nothing but the Wald entropy. Now we can run the machinery developed in chapter \ref{sec-cov-phase-space} to compute charges associated with symmetries and the relation among them. This way we find a \textit{first law} for NHEG geometries similar to first law of black holes derived by Wald.

Consider a set of field configurations that include NHEG geometries \eqref{NHEG-metric} as solutions to the field equations. We assume that this set accompanied by a consistent symplectic structure forms a phase space. We do not intend to specify the phase space more carefully since for the first law to be proved below, we only require minimal properties to hold which we mention during the argument. We start by computing equation \eqref{omega-dk} on a specific NHEG solution for a variation $\delta_{\zeta_H}\Phi$ where $\zeta_H$ is the Killing vector generating the horizon, and another arbitrary perturbation $\de \Phi$ tangent to the phase space
\begin{equation}\label{conservation-zetaH}
\boldsymbol{\omega}(\Phi_0,\delta \Phi,\delta_{\zeta_H}\Phi)  =\mathrm{d}\Big(\delta  \mathbf{Q}_{\zeta_H}- {\zeta_H} \! \cdot \! \mathbf{\Theta} (\Phi_0,\delta \Phi)\Big)\,.
\end{equation}
We integrate the above ``conservation equation'' over a timelike hypersurface $\Sigma$ bounded between two radii  $r=r_H,\;r=\infty$ where $r_H$ is the radius of the bifurcate horizon generated by $\zeta_H$. The hypersurface $\Sigma$ can be simply chosen as a constant time surface $t=t_H$. Since $\zeta_H$ is a Killing vector and the symplectic current is linear in perturbations, we have $\bomega(\Phi_0,\delta \Phi,\delta_{\zeta_H}\Phi) =0.$ Therefore we obtain 
\begin{align}\label{omega vs Q}
\nonumber0 &= \oint_{\partial\Sigma}\Big(\delta  \mathbf{Q}_{\zeta_H}- {\zeta_H} \! \cdot \! \mathbf{\Theta} (\Phi_0,\delta \Phi)\Big)\\
&=\oint_{\infty}\Big(\delta  \mathbf{Q}_{\zeta_H}- {\zeta_H} \! \cdot \! \mathbf{\Theta} (\Phi_0,\delta \Phi)\Big)-\oint_{H}\delta  \mathbf{Q}_{\zeta_H}
\end{align}
where in the first line we have used the Stokes theorem to convert the integral over $\Sigma$ to an integral over its boundary $\partial\Sigma$ and in the second line, we used the fact that $\zeta_H=n_H^a \xi_a-k^im_i$ vanishes on $H$. Since the charge perturbation $\delta\mathbf{Q}_{\zeta_H}$ is linear in the vector $\zeta_H$, one can expand the first term on RHS of \eqref{omega vs Q}
\begin{align}
n_H^a\oint_{\infty}\Big(\delta  \mathbf{Q}_{a}- {\xi_a} \! \cdot \! \mathbf{\Theta}\Big)-k^i\oint_{\infty}\Big(\delta  \mathbf{Q}_{m_i}- {m_i} \! \cdot \! \mathbf{\Theta}\Big)-\oint_{H}\delta  \mathbf{Q}_{\zeta_H}&=0\,.
\end{align}
$m_i$ is tangent to the boundary surface and hence the pullback of ${m_i} \! \cdot \! \mathbf{\Theta}$ over the surface $r=\infty$ vanishes, and we have
\begin{align}\label{Omega}
n_H^a\delta \mathcal{E}_a-k^i\oint_{\infty}\delta  \mathbf{Q}_{m_i}-\oint_{H}\delta  \mathbf{Q}_{\zeta_H}&=0\,,
\end{align}
where
\begin{align}\label{symplectic current}
\delta \mathcal{E}_a\equiv\oint_\infty\ (\delta \mathbf{Q}_{\xi_a} -{\xi_a} .\mathbf{\Theta})\,,
\end{align}
is the variation of the canonical charge corresponding to \sltr Killing $\xi_a$. Also the conserved charge related to axisymmetry Killing vectors $m_i$ are by definition the angular momenta associated to NHEG,
\begin{align}
-\delta J_i&\equiv\oint_\infty \delta \mathbf{Q}_{m_i}-{m_i} \! \cdot \! \mathbf{\Theta}\\
&=\oint_\infty \delta \mathbf{Q}_{m_i}
\end{align}
where in second line, we have used the fact that the pullback of ${m_i} \! \cdot \! \mathbf{\Theta}$ vanishes over any constant $t,r$ since $m_i$ is tangent to that surface. Replacing the above results into \eqref{Omega} we find
\begin{align}\label{delta Q zeta}
\oint_{H}\delta \mathbf{Q}_{\zeta_H}&=k^i\delta J_i+n_H^a\delta \mathcal{E}_a\,.
\end{align}
To show that the left side of \eqref{delta Q zeta}  is actually the perturbation of entropy $\delta S$, we use the result proved in \cite{Iyer:1994ys} 
\begin{proposition}
The Noether charge $\mathcal{Q}_{\xi}$ corresponding to any vector $\xi$ can be decomposed as
\begin{align}
\mathcal{Q}_{\xi}&=\oint_H d\Sigma_{\mu \nu}\lbrack W^{\mu \nu}_{ \,\,\,\alpha}{\xi}^\alpha -2\mathbf{E}^{\mu \nu}_{\quad \alpha \beta}\nabla^\alpha {\xi}^\beta +Y_{\xi}^{\mu\nu}+(dZ_{\xi})^{\mu\nu}\rbrack\,,
\end{align}
where $W$ and $dZ$ are linear in $\xi$ and $Y$ is linear in $\mathcal{L}_\xi \Phi$, and 
\begin{align}
{E}^{ \mu \nu \alpha \beta } \equiv \frac{\delta \mathcal{L}}{\delta R_{\mu \nu \alpha \beta}}\,.
\end{align}
\end{proposition}
Let us use the above decomposition for the Noether charge corresponding to $\zeta_H$. The $W$ and $dZ$ terms vanish since they are linear in $\zeta_H$ and by construction $\zeta_H$ vanishes at surface $H$. The $\delta Y$ term, which is proportional to variation of fields $\delta_\xi\Phi$ needs more attention. Since $\zeta_H=0$,  at surface $H$, $\delta_{\zeta_H}\Phi=0$. This implies that Y vanishes on background over $H$, and also that its perturbation is given by
\begin{align}\label{eq deltaY}
\delta Y(\Phi_0,\delta_{\zeta_H}\Phi)&=Y(\Phi_0,\delta\delta_{\zeta_H} \Phi)\nonumber\\
&=Y(\Phi_0,\delta_{\zeta_H}\delta \Phi)\nonumber\\
&=\delta_{\zeta_H}Y(\Phi_0 ,\delta \Phi)\nonumber\\
&=\zeta_H \cdot dY+d\,(\zeta_H\cdot Y )\,.
\end{align}
In the above we have used the fact that since $\delta\zeta_H=0$, hence we can  interchange $\delta_{\zeta_H}$ and $\delta$. Since \eqref{eq deltaY} is linear in the generator $\zeta_H$, it vanishes as well and does not contribute to the left hand side of \eqref{delta Q zeta}. Therefore
\begin{equation}
\delta \oint _H   \mathbf{Q}_{\zeta_H}=-2\delta\oint_H d\Sigma_{\mu \nu}\mathbf{E}^{\mu \nu}_{\quad \alpha \beta}\nabla^\alpha \zeta_H^\beta \,.
\end{equation}
Over the background we have $\nabla_{[\mu} {\zeta_H}_{\nu]}=\eps_\mn$ where $\eps_\mn$ is the binormal tensor of the bifurcation surface. Also it can be shown \cite{Iyer:1994ys} that this is also valid on the perturbed geometry once the above integral is concerned. Therefore
\begin{align}
\delta \oint _H \mathbf{Q}_{\zeta_H}&=-2\delta\oint_H d\Sigma_{\mu \nu}\mathbf{E}^{\mu \nu}_{\quad \alpha \beta}\eps^{\alpha\beta}=\frac{\delta S}{2\pi}\,.
\end{align} 
Hence we find
\begin{align}\label{entropy-pert-Ea}
\frac{\delta S}{2\pi}&=k^i\delta J_i+n_H^a\delta \mathcal{E}_a\,.
\end{align}
It can be argued that any suitable boundary condition leads to vanishing Hamiltonian charge for \sltr isometries \cite{Amsel:2009ev,Hajian:2014twa}. Therefore $\delta \mathcal{E}_a=0$. By Dropping this term in \eqref{entropy-pert-Ea} we arrive at the entropy perturbation law 
\begin{align}\label{entropy-pert-law}
\frac{\delta S}{2\pi}&=k^i\delta J_i\,.
\end{align}
This is an analog of first law of black hole mechanics. This result was also obtained independently \cite{Johnstone:2013ioa} using the first law of black hole mechanics for an extremal black hole and a near extremal black hole.  
\section{The Kerr/CFT correspondence}\label{Kerr CFT}
In \cite{Guica:2008mu}, a duality was conjectured between the dynamics over the near horizon geometry of extremal Kerr (NHEK), and a \textit{chiral} two dimensional CFT. In this section, we briefly review the argument of Kerr/CFT and mention further supports for this conjecture. 

First Let us Recall the argument of Brown and Henneaux. They showed that the asymptotic symmetry algebra of asymptotically \ads spacetimes ,\ie geometries satisfying boundary conditions \eqref{F-G gauge} and \eqref{BCBH}, is the direct sum of two copies of Virasoro algebra. Then if a dual theory exists on the boundary of \ads, its states should transform in the representations of Vir$\oplus$ Vir. The field theory with such a property is nothing but a 2 dimensional CFT. Later, Strominger argued that using the Cardy formula for the entropy of the dual CFT$_2$ reproduces the entropy of BTZ black hole \cite{Strominger:1997eq}. In the same way, the authors of \cite{Guica:2008mu} tried to find a holographic description  for ``asymptotic NHEK'' geometries.  They first step is to provide a suitable boundary condition. In the coordinate system $\{t,r,\theta,\varphi^i\}$, they proposed the following boundary condition near the boundary $r\to +\infty$
\begin{align}\label{BC}
 \delta g_{\mu\nu} \sim \mathcal{O}\left( \begin{array}{cccc}
r^2&1/r^2 &1/r&1 \\
&1/r^3&1/r^2 &1/r \\
 & &1/r&1/r\\
 & & &1\\
\end{array} \right),
\end{align}
Moreover, they required that the charges associated to \sltr  Killing vectors $\xi_1,\xi_2$ is not excited \cite{Amsel:2009pu}, i.e $\de \mathcal{E}_{1,2}=0$. The condition $\de \mathcal{E}_{1}=0$ especially means that no excitation of energy is allowed. This is reasonable since a finite energy excitation in the near horizon geometry corresponds to an infinite energy excitation in the original black hole geometry. As a result of the above conditions, it turns out that the set of nontrivial asymptotic symmetries are then generated by vectors
\begin{align}\label{Kerr CFT sym}
\chi&= \eps(\varphi)\p_\varphi-r\,\eps'(\varphi)\p_r
\end{align}
where $\eps(\varphi)$ is an arbitrary periodic function of $\varphi$. Now expanding in Fourier modes, and defining $\chi_n=\chi(\eps=e^{in\varphi})$ leads to the algebra
\begin{align}
[\chi_n,\chi_m]&=i(n-m)\chi_{n+m}\,.
\end{align}
The corresponding charges are obtained by \eqref{kgrav} and the algebra of charges is defined by \eqref{charge bracket}. The result after the Dirac quantization rule $\{,\}\to \frac{i}{\hbar}[,]$ is 
\begin{align}
[Q_n,Q_m]&=(n-m)Q_{n+m}+\dfrac{c}{12}n^3\de_{n+m,0}
\end{align}
where 
\begin{align}\label{central charge Kerr/CFT}
c=\dfrac{12J}{\hbar}\,.
\end{align}
This means that the asymptotic symmetry algebra is one copy of Virasoro with the central charge \eqref{central charge Kerr/CFT}. The above observation led the authors of \cite{Guica:2008mu} to conjecture a duality between the phase space of ``asymptotic NHEK'' geometries and \textit{chiral} CFT$_2$. The notion of a chiral CFT was made more precise in \cite{Balasubramanian:2009bg} in relation with discrete light cone quantization (DLCQ). Moreover, by using Cardy formula for the proposed chiral CFT and assuming that the state of the black hole is represented by a canonical ensemble at temperature known as the Frolov Thorne temperature $T_{F.T}=\frac{1}{2\pi k}$ reproduces the entropy of extremal Kerr black hole $S=2\pi J/\hbar=A/4G\hbar$. However, it should be noticed that the applicability of Cardy formula for this problem is not obvious since the temperature is not high,but instead the central charge is large. Moreover, the modular invariance of the dual theory should be assumed which is again not obvious.

The ``Kerr/CFT correspondence'' proposal was later generalized to include other extremal black holes \cite{Hartman:2008pb,Lu:2008jk,Chow:2008dp,Compere:2009dp,Azeyanagi:2008kb} (see \cite{Compere:2012jk} for a review) \eg extremal RN and Kerr-Newmann black holes , Myers-Perry black holes in 5 dimensions, extremal 4d black holes in higher derivative gravity, extremal black holes in Einstein-Maxwell-scalar in 4 and 5 dimensions, black holes in supergravity, etc. The broad extent of evidence for the generalized Kerr/CFT proposal was suggesting that there is a holographic description for near horizon extremal geometries (NHEGs) in general.

However, there appeared a big challenge for Kerr/CFT regarding the dynamics over NHEG. In two independent papers \cite{Amsel:2009ev,Dias:2009ex} a thorough analysis of the linear dynamics allowed by the Kerr/CFT boundary conditions led to the conclusion that there is ``No dynamics'' in NHEK with the given boundary conditions. This raised questions about the dynamical content of Kerr/CFT. A resolution was proposed in \cite{Compere:2015mza,Compere:2015bca}, by attributing the dynamics to ``surface gravitons''. 

\section{No dynamics in NHEK}\label{sec No dynamics}
In this section we summarize the result of \cite{Dias:2009ex,Amsel:2009ev}. The question is to study the solutions to the linearized Eistein equations over the near horizon geometry of Kerr black hole (NHEK), subject to the Kerr/CFT falloff conditions \eqref{BC}.
  
Perturbations of metric are gauge dependent quantities. So one can either solve the linearized field equations in a  fixed gauge, or more systematically to work with gauge invariant quantities which contain the information about field perturbations, similarly to what is usually done in cosmic perturbation theory (see \eg \cite{Mukhanov:2005sc}).

In the context of Petrov type D spacetimes, Teukolsky \cite{Teukolsky:1973ha} introduced a set of gauge invariant scalars built from Weyl tensor and principal null vectors of the space, and  used them to discuss perturbations of Kerr geometry in a series of papers \cite{Press:1973zz,Teukolsky:1974yv}. This formulation is based on the Newman-Penrose null tetrad \cite{Newman:1961qr}, and the corresponding directional derivatives and spin coefficients. Teukolsky then derived a master equation for the Weyl scalars mentioned above. These scalars contain the physical information about the metric perturbations .

\paragraph{Hertz Potential.} In our problem we need the explicit form of metric perturbations in order to compare with the Kerr/CFT falloff conditions. It was shown in \cite{Kegeles:1979an, Chrzanowski:1975wv, Wald:1978vm} (see \cite{price2007developments} for a review) how to construct field perturbations using a gauge invariant scalar, called the ``Hertz potential'' which is a solution of Teukolsky master equation. Given the Hertz potential one can construct metric perturbations in a specific gauge called the \textit{radiation gauge}. The explicit form of metric perturbation in terms of the Hertz potential is given in equation (3.41) of \cite{Dias:2009ex}.

Performing the above computations, with an outgoing boundary condition, we can find the metric perturbations. Their asymptotic behavior is 
\begin{align}
h_{\mn}&\sim r^{(-1 \pm \eta)/2}
\begin{pmatrix}
r^2&1&r&r\\
 &1/r^2&1/r&1/r\\
 & & 1&1\\
 & & & 1
\end{pmatrix}
\end{align}
where $\eta$ is determined in terms of the eigenvalues $n,l,m$ of the ``spin-weighted spheroidal harmonic'' given in equation 2.20 of \cite{Dias:2009ex}. There are two types of solutions
\begin{enumerate}
	\item Traveling modes with imaginary $\eta$,
	\item normal modes with real $\eta$.
\end{enumerate}
Traveling modes do not respect the the Kerr/CFT falloff conditions and hence are excluded. For the normal modes, again since the minimum value of $\eta$ is $\eta=2.74$ for $l=4,|m|=3$, the positive sign is again excluded. Therefore at the linear level, the set of normal modes with negative sign for $\eta$ are allowed. However, an analysis of charges using the second order Einstein equations show that these modes are associated with positive energy excitation, and therefore no solution exist satisfying the Kerr/CFT conditions. This is the essence of No dynamics argument.

As mentioned in \cite{Dias:2009ex}, there are certain modes which are not captured by the Hertz inverse map. These are perturbations preserving the type D property of the geometry. For NHEK, these modes are
\begin{itemize}
	\item perturbations that are locally gauge, \ie of the form $\cd_{(\mu}\chi_{\nu)}$
	\item a parametric perturbation, \ie modes corresponding to a  deformation
	towards an NHEK with slightly perturbed angular momentum. $J\to J+\de J$.
\end{itemize}
We use this freedom to build the NHEG phase space in next chapter. The latter is again excluded by Kerr/CFT fall off conditions. 
\subsection{Kerr/CFT challenges and modifications}
\paragraph{Singular phase space.} The set of asymptotic symmetry transformations derived from the Kerr/CFT fall off conditions does not lead to a smooth phase space. This is because an infinitesimal perturbation produced by $\bar{g}_\mn\to \bar{g}_\mn+\mathcal{L}_\chi \bar{g}_\mn$ is singular at the poles $\theta=0,\pi$. 
Let's assume that it is possible to relax the boundary conditions such that the asymptotic symmetry vector is modified to 
\begin{align}
\chi[{\epsilon}(\varphi)]&= {\epsilon}(\varphi)\pd_{\varphi} -  \pd_{\varphi}{\epsilon(\varphi) }\;(\dfrac{1}{r}\pd_{t}+r\pd_{r}).\label{KCFT}
\end{align}
Expanding in Fouriere modes, one can then see that the algebra is still one copy of Virasoro. Moreover, the phase space transformations produced by this vector is smooth everywhere. In the next chapter, we will construct the NHEG phase space realizing these vectors as symmetries.
\paragraph{Non-conserved charges}
The symplectic structure was defined using the invariant symplectic current $\bomega^{inv}$ introduced in previous chapter (see section \ref{symp-sym-sec}). However, one can check that the symplectic current has an infinite flux through the boundary $r=\infty$. This is a crucial problem to be resolved in our construction of NHEG phase space.
\section{NHEG perturbation uniqueness}\label{sec perturbation uniqueness}
In this section we explain a physical argument that restricts the physically relevant perturbations over NHEG to a great extent. This works for NHEG in any dimension. Then by adding a few asymptotic conditions, we prove the ``NHEG perturbation uniqueness'' which states that the only allowed perturbation is a parametric perturbation, \ie an infinitesimal variation in parameters that specify the NHEG. An outcome of this argument is that the Kerr/CFT falloff conditions are not the physically appropriate conditions and the symmetry vectors are given by \eqref{KCFT} (or a straightforward extension in higher dimensions) instead of \eqref{Kerr CFT sym}. 

\subsection*{Radial dependence of perturbations}
The following argument completely fix the $r$ dependence of metric perturbations. It is not merely an asymptotic falloff expansion. Therefore it leads to a great simplification. 
\paragraph{Proposition.} \emph{Perturbations of an extremal black hole  which survive the near horizon limit and  are well-behaved under the limit, give rise to perturbations on NHEG which are invariant under $\xi_1$ and $\xi_2$ diffeomorphisms.}\cite{Hajian:2014twa}.\\
The logic is simple. Since we are interested ultimately in the semi classical description of extremal black hole, we intend to confine the perturbations over NHEG to those that can originate from a perturbation over the original extremal black. Consider any perturbation $\tilde{g}_\mn\to \tilde{g}_\mn+ \tilde{h}_{\mu\nu}$ over the extremal black hole. The perturbation over black hole leads to a perturbation over the NHEG geometry -say $h_{\mu\nu}$ - after taking the near horizon limit. We show that if we require to have a finite (non-vanishing and non-divergent) perturbation over NHEG, $h_\mn$ necessarily preserves the $\xi_{1,2}$ isometries of the background NHEG, i.e. $\mathcal {L}_{\xi_{1,2}}h_{\mu\nu}=0$.\\
To see this, consider the generic form of an extremal black hole in coordinates $\tau,\rho,\theta,\psi$. For simplicity, we do this in four dimensions, but the same argument applies to higher dimensions.
\begin{align}
ds^2&=-\tilde{f} d\tau^2+\tilde{g}_{\rho\rho}d\rho^2+\tilde {g}_{\theta\theta}d\theta^2 +\tilde{g}_{\psi\psi}(d\psi-\omega d\tau)^2\,.
\end{align}
where $\tilde{g}_{\rho\rho}$ has a double root at $r=r_e$ and near the horizon $\tilde{f}\propto 1/\tilde{g}_{\rho\rho}$. The near horizon limit is defined through the transformations 
\begin{align}
\rho &=r_e(1+\lambda r)\,,\qquad
\tau =\dfrac{r_e t}{\lambda},\qquad
\varphi =\psi-\Omega \tau\,,\qquad \lambda\to 0\,,
\end{align}
where $r_e$ is the horizon radius, $\Omega=\omega(r=r_e)$. For symplicity we set $r_e=1$.

Next, we perturb the extremal black hole geometry by a metric perturbation $\tilde{h}_{\mu\nu}$, that is  the metric for the perturbed geometry is $\tilde{g}_\mn+\tilde{h}_{\mu\nu}$
\begin{align}
\nonumber \tilde{h}_{\mu\nu}dx^\mu dx^\nu&=\tilde{h}_{\tau\tau}d\tau^2+2d\tau (\tilde{h}_{\tau\theta}d\theta+\tilde{h}_{\tau\psi}d\psi+\tilde{h}_{\tau \rho}d\rho)\\
\nonumber & +\tilde{h}_{\rho\rho}d\rho^2+2d\rho(\tilde{h}_{\rho\theta}d\theta+\tilde{h}_{\rho\psi}d\psi)\\
&+\tilde{h}_{\theta\theta}d\theta^2+2\tilde{h}_{\theta\psi}d\theta d\psi+\tilde{h}_{\psi\psi}d\psi^2.
\end{align}
We rewrite the above expressions in terms of the near horizon coordinates $t,r,\theta,\varphi$ before taking the limit $\lambda\to 0$
\begin{align}
\nonumber \tilde{h}_{\mu\nu}dx^\mu dx^\nu&=\dfrac{d {t}^2}{\lambda^2}\Big(\tilde{h}_{\tau\tau}+2\Omega \tilde{h}_{\tau\psi}+\Omega^2 \tilde{h}_{\psi\psi}\Big)\\
\nonumber & +2\dfrac{d {t}}{\lambda} \Big( \lambda d {r}(\tilde{h}_{\tau\rho}+\Omega \tilde{h}_{\rho\psi})+d {\theta}(\tilde{h}_{\tau\theta}+\Omega \tilde{h}_{\psi\theta})+d {\varphi}(\tilde{h}_{\tau\psi}+\Omega \tilde{h}_{\psi\psi})\Big) \\
\nonumber & +\lambda^2d {r}^2\tilde{h}_{\rho\rho}+2\lambda\, d{r}\Big(\tilde{h}_{\rho\theta}d{\theta}+\tilde{h}_{\rho\psi}d {\varphi}\Big)+\Big(\tilde{h}_{\theta\theta}d {\theta}^2+2\tilde{h}_{\theta\psi}d {\theta} d {\varphi}+\tilde{h}_{\psi\psi}d{\varphi}^2\Big).
\end{align}
Therefore perturbation induced on the NHEG (which we denote by $h_{\mu\nu}$) is
\begin{align}\label{h-components}
\nonumber h_{tt}&=\dfrac{\tilde{h}_{\tau\tau}+\Omega \tilde{h}_{\tau\psi}+\Omega^2 \tilde{h}_{\psi\psi}}{\lambda^2},\hspace{1cm}
\nonumber h_{rr}=\lambda^2 \tilde{h}_{\rho\rho}\\
h_{tr}&=\tilde{h}_{\tau\rho}+\Omega \tilde{h}_{\rho\psi},\hspace{1cm}
h_{t\theta}=\dfrac{\tilde{h}_{\tau\theta}+\Omega \tilde{h}_{\theta\psi}}{\lambda},\hspace{1cm}
h_{t\phi}=\dfrac{\tilde{h}_{\tau\psi}+\Omega \tilde{h}_{\psi\psi}}{\lambda}\\
\nonumber h_{\theta\theta}&=\tilde{h}_{\theta\theta},\hspace{1cm}
\nonumber h_{\varphi\varphi}=\tilde{h}_{\psi\psi},\hspace{1cm}
\nonumber h_{r\theta}=\lambda \tilde{h}_{\rho\theta},\hspace{1cm}
\nonumber h_{r\varphi}=\lambda \tilde{h}_{\rho\psi},\hspace{1cm}
h_{\theta\varphi}=\tilde{h}_{\theta\psi}\,.
\end{align}

Solutions to the linearized field equations on black hole geometry can be solved by separation of variables. That is 
\begin{align}\label{BH-modes}
\tilde{h}_{\mu\nu}\sim f(\theta)e^{-i(\nu \tau-m\psi)} (\rho-r_h)^x=f(\theta) e^{i(\frac{\nu-\Omega m}{\lambda}) t}\ e^{im\varphi} (\lambda r)^x\,.
\end{align}
according to above equation, perturbations which result in a finite perturbation after the near horizon limit should have $\nu=m\Omega $, and therefore lead to stationary perturbations over NHEG, i.e. 
$\mathcal{L}_{\xi_1}{h}_{\mu\nu}= 0$.

Using \eqref{BH-modes} in \eqref{h-components} and requiring to have finite ${h}_{\mu\nu}$ in the $\lambda \rightarrow 0$ limit, fixes  the $r$ dependence of the perturbations  as:
\begin{equation}\label{hmunu-pert-r-power}
{h}_{\mu\nu}=
\begin{pmatrix}
r^2 & 1 & r & r \\
& 1/r^2 & 1/r & 1/r\\
& & 1 & 1 \\
& & & 1
\end{pmatrix},
\end{equation}
in the $(t,r,\theta,\varphi)$ basis. Note that higher orders of $r$ lead to terms with positive powers of $\lambda$ in ${h}_{\mu\nu}$ so that they disappear in the $\lambda\rightarrow 0$ limit. Also, lower orders of $r$ lead to divergence in  ${h}_{\mu\nu}$ which is excluded. Therefore, \eqref{hmunu-pert-r-power} gives the exact $r$-dependence of components (and not just a large $r$ behavior). This $r$-dependence is exactly dictated by the condition $\mathcal{L}_{\xi_2}{h}_{\mu\nu}=0$. 

If one also assumes that the perturbations satisfy \textit{asymptotically} the isometries of background, then it can be proved \cite{Hajian:2014twa} that the perturbation is restricted to a variation in NHEG metric with slightly deformed parameters $J_i\to J_i+\de J_i$. Note the similarities and differences between \eqref{hmunu-pert-r-power} and \eqref{BC}.

\chapter{NHEG phase space and its surface degrees of freedom}\label{chapter NHEG phase space}
\section{Introduction}
In this chapter we focus on the class of $d$ dimensional Near Horizon Extremal Geometries, which are solutions to vacuum Einstein gravity and have $SL(2,\mathbb R) \times U(1)^{d-3}$ isometry. These geometries are specified by constant parameters $k^i,\ i=1,2,\cdots, d-3$, which will be collectively denoted as $\vec{k}$, and a set of functions of the coordinate $\theta$.\footnote{The dimensionless vector $\vec{k} $ physically represents the linear change of angular velocity close to extremality, normalized using the Hawking temperature,  $\vec{\Omega} = \vec{\Omega}_{ext} + \frac{2 \pi}{\hbar} \vec{k} \;T_H + O(T_H^2)$, see e.g. \cite{Compere:2012jk,Johnstone:2013ioa}. For the extremal Kerr black hole, $k=1$.} There are then $d-3$ conserved charges $\vec{J}$, associated with $U(1)^{d-3}$.
The NHEG has an entropy $S$ which is related to the other parameters as $\frac{\hbar}{2\pi} S=  k^i J_i= \vec{k}\cdot\vec{J}$ \cite{Astefanesei:2006dd,Hajian:2013lna}.

According to the absence of dynamical degrees of freedom over NHEG with suitable boundary conditions that we discussed in previous chapter, we build the phase space of NHEG (denoted by $\bG$) as a set of  metrics  with $SL(2,\mathbb R) \times U(1)^{d-3}$ isometry diffeomorphic to the background NHEG \eqref{NHEG-metric}. This is similar to the case of \ads phase space that all geometries on an orbit are diffeomorphic to each other while they are physically different due to their nontrivial charges. We show that this construction circumvents all problems encountered by the Kerr/CFT proposal while maintain its achievements. 
\subsection{Summary of results}
Our main result are: 
\begin{enumerate}
\item Different configurations of the NHEG phase space are labeled by an arbitrary periodic function  $F=F(\vec{\varphi})$ on the $d-3$ torus spanned by the $U(1)$ isometries. We call $F$ the \emph{wiggle function}. 
\item A vector tangent to the phase space is hence defined as an infinitesimal diffeomorphisms which induces an arbitrary (but infinitesimal) change in the wiggle function. We show that these vectors are also the \textit{symplectic symmetries} of the phase space, and more strongly they are the \textit{local symplectic symmetries }of the phase space as defined in section \ref{local symplectic symmetries}.
\item The phase space is equipped with a consistent symplectic structure through which we define conserved surface charges associated with each
symplectic symmetry. 
\item We work out the algebra of these conserved charges, the NHEG algebra $\widehat{\mathcal{V}_{\vec{k},S}}$ whose generators $L_{\vec{n}}$,  
$\vec{n} \in \mathbb Z^{d-3}$, satisfy 
\bea\label{NHEG-algebra_i}
[ L_{\vec{m}}, L_{\vec{n}} ] = \vec{k} \cdot (\vec{m}- \vec{n}) L_{\vec{m}+\vec{n}} +  \frac{S}{2\pi}(\vec{k}\cdot \vec{m})^3 \delta_{\vec{m}+\vec{n},0}\,.
\eea
The NHEG algebra generators commute with the generators of angular momenta $J_i$. The algebra is Virasoro in four dimensions but in higher dimensions is a novel extension thereof that to our knowledge has not appeared in the literature before. Also interestingly The entropy $S$ appears as the central charge of the algebra.
\item We give an explicit construction of the charges over the phase space from a one-dimensional ``Liouville stress-tensor'' for a fundamental boson field ${\Psi}$, which is constructed from the wiggle function $F(\vec{\varphi})$. 

\end{enumerate}

\subsection{Outline}

In section \ref{sec rationale} we discuss how we construct the family of geometries which will be promoted as the elements of the NHEG phase space. These geometries are built through a specific family of diffeomorphisms with one arbitrary function in sections \ref{Construction of chi} and \ref{sec finite trans}. 
We then specify the symplectic structure on the set of these geometries in section \ref{sec-symplectic structure}. 
In section \ref{sec-charges-and-algebra} we  compute the algebra of charges and exhibit the central extension. Moreover, we give an explicit representation of the charges over the phase space in terms of the wiggle function $F(\vec{\varphi})$. We also discuss the quantized version of the NHEG algebra.

In the last section \ref{sec-NHEG-discussion}, we further discuss the results and the physical implications of our construction and discuss various ways in which it can be extended.

\section{Rationale behind the NHEG phase space}\label{sec rationale}
Based on the argument in section \ref{sec perturbation uniqueness}, we assume that the No dynamics property continues to hold for generic NHEG with \sltruod isometry (as a solution to vacuum Einstein gravity) in higher dimensions. Due to the absence of well behaved propagating degrees of freedom in the NHEG background, we propose to build the NHEG phase space based on perturbations that escape the No dynamics argument, \ie perturbations produced by infinitesimal coordinate transformations. In the discussion, we will discuss the extension of the phase space by adding \textit{parametric} perturbations. 

To build the phase space, we first consider the set of solutions of the form \eqref{NHEG-metric} with fixed given value of parameters. We work in pure Einstein gravity in which the metric is labeled by $d-3$ continuous parameters $J_i$, the angular momenta of the geometry. Note that according to \cite{Hollands:2009ng}, there can be a couple of discrete parameters as well. These are related to topological invariants of the geometry. However, since these parameters cannot be varied continuously, these parameters would produce disconnected patches of the phase space. Here we discard that possibility and build only the phase space of geometries simply connected to the background \eqref{NHEG-metric}.

The next step is to feed ``surface gravitons" into the phase space. These are perturbations produced by infinitesimal coordinate transformations which are associated with nontrivial surface charges. In contrast to the asymptotic analysis of Brown- Henneaux, we give the form of these coordinate transformations everywhere in the bulk, not merely in the asymptotic region. Arbitrary field configurations of the phase space are then produced by finite coordinate transformations, obtained by the \textit{exponentiation} of these infinitesimal coordinate transformations. The set of geometries thus obtained is the analog of Ba\~nados geometries in 3 dimensions.



The so called surface gravitons are obtained by infinitesimal coordinate transformations generated by a vector field  $\chi^\mu$ through $x^\mu\rightarrow x^\mu+ \chi^\mu$.
In the following, we will first single out the infinitesimal diffeomorphisms around the background using a set of physical requirements. We denote all dynamical fields as $\Phi$. In this chapter $\Phi$ is only the metric, but we keep that notation to facilitate possible generalizations with additional fields. An active coordinate transformation generates a perturbation, denoted as $\delta_\chi \Phi$, which is the Lie derivative of the dynamical field $\delta_\chi \Phi=\mathcal{L}_\chi \Phi$. Such a perturbation automatically obeys the linearized field equations as a consequence of general covariance. 


\textbf{Notations.} All quantities associated with the background metric \eqref{NHEG-metric} will be defined with an overline. In particular, the metric \eqref{NHEG-metric} will be denoted as $\bar\Phi  \equiv \bar g_{\mu\nu}$ and infinitesimal diffeomorphisms around the background will be generated by $ \overline \chi^\mu $. Instead, we denote a generic element of the phase space as $\Phi$ and an infinitesimal diffeomorphism tangent to the phase space as $\chi$.

\section{Generator of infinitesimal transformations}\label{Construction of chi}
	
In the context of asymptotic symmetry group, the set of symmetry transformations are obtained by the requirement that they respect the boundary conditions. Therefore that prescription only fix the asymptotic behavior of the symmetry transformations. Here we need to fix the form of symmetry transformations at any point of the spacetime. Therefore we use a couple of \textit{local} conditions to fix  $\chi$. We will bring supporting arguments for each condition.
	
We start with the most general diffeomorphism around the background $\bar{x}\to \bar{x}+\overline \chi$ and restrict $\bar{\chi}$ through the six conditions listed below.
	
	\paragraph{ (1) $[\overline \chi,\xi_{0}]=[\overline \chi,\xi_{-}]=0$.} These conditions are supported as follows:
	\begin{description}
				\item[1.1) Perturbations $\delta_{\overline \chi}\Phi$ should preserve two isometries $\xi_-,\xi_0$.] As we argued in section \ref{sec perturbation uniqueness} perturbations over NHEG that can be regarded as a perturbation over the original extremal black hole, have to respect the two isometries $\xi_-,\xi_0$ of the background. These conditions are then rephrased as above.
				
				\item[1.2) $\mathcal H$-independent charges.] Any conserved charge is defined through integrating over a $d-2$ dimensional surface $\mathcal{H}$. The two isometries discussed above imply an scaling symmetry in radial direction and a translational invariance in time. Therefore we expect that the charges computed at any radius and at any time give the same result. This will be satisfied if the above condition is met.

		\end{description}
		This condition fixes the $t$ and $r$ dependence of all components  of ${\overline \chi}$:
		\be\label{chi-xi1,2}
		{\overline \chi}=\frac{1}{r}\eps^t\,\p_t+r\eps^r\,\p_r+\eps^\theta\,\p_\theta+\vec{\epsilon}\cdot\vec{\p}_{\varphi},
		\ee
		where the  coefficients $\epsilon$ are only functions of $\theta,\vec{\varphi}$. An outcome of this condition is that $\xi_- = \overline \xi_-$ and $\xi_0 = \overline \xi_0$ for any element of the phase space. That is $\xi_-$, $\xi_0$ will be Killing isometries of each element of the phase space.
		
		\paragraph{(2) $\nabla_\mu {\overline \chi}^\mu=0$.} This condition implies that the volume element 
		\begin{equation}\label{volume-form-d}
		\boldsymbol{\epsilon}=\frac{\sqrt{-g}}{d\,!}{\epsilon}_{\mu_1\mu_2\cdots\mu_d}dx^{\mu_1}\wedge dx^{\mu_2}\wedge\cdots \wedge dx^{\mu_d},
		\end{equation}
		is invariant under the coordinate transformations, since
		\begin{equation}
		\mathcal{L}_{\overline \chi} \boldsymbol{\epsilon}={\overline \chi}\cdot d\boldsymbol{\epsilon}+d({\overline \chi}\cdot \boldsymbol{\epsilon})=d({\overline \chi}\cdot \boldsymbol{\epsilon})=\star(\nabla_\mu {\overline \chi}^\mu),
		\end{equation}
		where $\star$ is Hodge dual operator. Therefore, $\mathcal{L}_{\overline \chi} \boldsymbol{\epsilon}=0$ is equivalent to $\nabla_\mu {\overline \chi}^\mu=0$.
		
		\paragraph{(3) $\delta_{\overline \chi}\mathbf{L}=0$,} where $\mathbf{L}=\frac{1}{16\pi G} R\boldsymbol{\epsilon}$ is the Einstein-Hilbert  Lagrangian $d$-form evaluated on the NHEG background \eqref{NHEG-metric} before imposing the equations of motion. (The functional form of $\Gamma(\theta)$ and $\gamma_{ij}(\theta)$ is therefore arbitrary except for the regularity conditions.)  Since $R$ is a scalar built from the metric, it is invariant under the background $SL(2,\mathbb{R})\times U(1)^{d-3}$ isometries and only admits $\theta$ dependence.\\
		
		\noindent The above properties (2) and (3) then lead to 
		\begin{equation}
		\eps^\theta=0, \qquad \quad \eps^{r}=-\vec{\partial}_{\varphi}\cdot \vec{\epsilon}\,.
		\end{equation}
		
		\paragraph{(4) $ \eps^t = -  \, \vec{\partial}_{\varphi}\cdot \vec{\epsilon}$.}
		This condition can be motivated from two different perspectives: 
		\begin{description}
			\item[4.1) Regularity of $\mathcal H$ surfaces.] As we will discuss in section \ref{sec finite trans}, this condition ensures that constant $t,r$ surfaces $\mathcal H$ are regular without singularities at poles on each element of the phase space. Dropping  the $t$ component of $\chi$ instead (as done in \cite{Guica:2008mu}) will lead to surfaces $\mathcal H$ with singularities at poles.

			\item[4.2) Preservation of a null geodesic congruence.] As discussed in section \ref{sec-Killing-horizon}, the NHEG has two expansion, rotation and shear free null geodesic congruences generated by $\ell_+$ and $\ell_-$ which are respectively normal to constant $v = t+\frac{1}{r}$ and $u=t-\frac{1}{r}$ surfaces \cite{Durkee:2010ea}. We request that $\mathcal{L}_{\overline \chi} v=0$, yielding the above condition.

			\end{description}

			\paragraph{(5) $\vec{\eps}$ are $\theta$-independent and periodic functions of $\varphi^i$.} We impose these conditions as they guarantee (i) smoothness of the $t,r$ constant surfaces ${\cal H}$ of each element of the phase space, as we will show below in section \ref{sec finite trans}, and (ii) constancy of the angular momenta $\vec{J}$ and the volume of ${\cal H}$ over the phase space, as we will also show in section \ref{sec:iso}.
			
			\paragraph{(6)  Finiteness, conservation and regularity of the symplectic structure.} 
			These final conditions crucially depend on the definition of the symplectic structure which is presented in section \ref{sec-symplectic structure}. After fixing the ambiguities in the boundary terms of the symplectic structure, we find 
			\begin{align}
				\vec{\eps} = \vec{k} \,\eps(\varphi^1,\dots \varphi^{d-3})
			\end{align}
 where $\eps$ is a function periodic in all its $d-3$ variables. There is also another possibility discussed in \cite{Compere:2015bca}. However we discard that possibility since it requires the choice of a preferred direction on the torus by hand, and the resulting boundary dynamics is much more restricted.
				

				
				As a result of the above conditions, we end up with the following vector which generates the surface gravitons of NHEG phase space through $\delta_\eps\Phi=\mathcal{L}_{\overline\chi(\eps)} \bar{\Phi}$
				\begin{equation}\label{ASK}
				\boxed{\overline 
				\chi[{\epsilon}(\vec{\varphi})]={\epsilon}\vec{k}\cdot\vec{\pd}_\varphi-(\vec{k}\cdot\vec{\pd}_\varphi{\epsilon})\, r\pd_{r}-\dfrac{(\vec{k}\cdot\vec{\pd}_\varphi{\epsilon})}{r}\pd_{t}}
				\end{equation}
				where $\eps=\eps(\varphi^i)$ is a function of all periodic coordinates.
				
				\section{Finite transformations and phase space configurations}\label{sec finite trans}
				
				We define the NHEG phase space through the \textit{exponentiation} of the infinitesimal perturbations $\de_\eps \Phi$. For doing this, we can instead exponentiate the infinitesimal coordinate transformations $ \overline x \rightarrow x(\overline x)$ to find the corresponding \textit{finite} form of coordinate transformations $ \overline x \rightarrow x(\overline x)$. Applying this finite coordinate transformation on the background metric, then gives the metric of a generic point of the phase space. 
				
				We obtain the form of finite transformations by postulating that the vector $\chi$ be \textit{field independent} throughout the phase space. That is, it does not depend on the dynamical fields, hence $\chi$ keeps its functional form under phase space transformations. More precisely, we require that the coordinate transformation maps the vector $\chi[{\epsilon}(\varphi)]$ defined on a generic metric $g_\mn$ to the vector $\overline \chi[\overline {\epsilon}(\bar{\varphi})]$ defined on the background $\bar{g}_\mn$ 
				\begin{align}\label{chi-vs-chi'}
				\chi^\mu[{\epsilon}(\varphi)]&=\dfrac{\pd x^\mu}{\pd \bar{x}^\alpha}\overline \chi^\alpha[\overline\epsilon(\bar{\varphi})].
				\end{align}
				In this section we multiply the $t$ component of \eqref{ASK} with a parameter $b$. This will allow us to derive the property \textbf{(4.2)} claimed in the previous subsection. Later we set this parameter back to 1. 				
				
				The finite coordinate transformation corresponding to the infinitesimal transformations through \eqref{ASK} takes the form
				\begin{align}\label{finite ansatz}
				\bar{\varphi}^i=\varphi^i + k^i F(\vec{\varphi}), \qquad \bar{r} =re^{-{\Psi(\vec{\varphi})}},\qquad \bar{\theta}=\theta,\qquad \bar{t} =t-\frac{b}{r}(e^{\Psi(\vec{\varphi})}-1),
				\end{align}
				The vector $\chi$ in \eqref{ASK} has one arbitrary function, while here we have the two functions $F(\varphi^i)$ and ${\Psi}(\varphi^i)$ periodic in all of their arguments in order to ensure smoothness. It was shown in appendix B.3 of \cite{Compere:2015bca}, that \eqref{chi-vs-chi'} implies the following relation between  $\Psi$, $F$ 
				\begin{align}\label{Psi-def}
					e^{\Psi}&=1+\vec{k}\cdot\vec{\pd}_\varphi F\\
					\bar\eps(\vec{\bar\varphi})&=e^{\Psi}\ \eps(\vec{\varphi}).\label{eps-eps-bar}
				\end{align}
				Therefore we have found the exact form of finite transformations. Note that the latter equation in \eqref{finite ansatz} can be rewritten as 
				\begin{align}
				\bar{t}+\dfrac{b}{\bar{r}}=t+\dfrac{b}{r}
				\end{align}
				this is because at the infinitesimal level $\overline  \chi$ commutes with the vector
				\bea
				\eta_b \equiv \frac{b}{r}\p_t+r\p_r\,. \label{defetab}
				\eea			
Using the finite coordinate transformations we can finally derive the one-function family of metrics which constitute the phase space in the $(t,r,\theta,\varphi^i)$ coordinate system:
				\begin{align}\label{g-F}
				ds^2=\Gamma(\theta)&\Big[-\left( \boldsymbol\sigma - b \,d \Psi \right)^2+\Big(\dfrac{dr}{r}-d{\Psi}\Big)^2
				+d\theta^2+\gamma_{ij}(d\tilde{\varphi}^i+{k^i}\boldsymbol{\sigma})(d\tilde{\varphi}^j+{k^j}\boldsymbol{\sigma})\Big],
				\end{align}
				where $v_b=t+\frac{b}{r}$ and
				\be\label{tilde-varphi}
				\boldsymbol{\sigma}=e^{-{\Psi}}r\,dv_b+b\dfrac{dr}{r},\qquad \tilde{\varphi}^i\equiv \varphi^i+k^i (F-b{\Psi})\,.
				\ee
				We note that, by virtue of periodicity of $F$ and ${\Psi}$, all angular variables $\bar{\varphi}^i$, $\varphi^i$ and $\tilde{\varphi}^i$ have $2\pi$ periodicity \footnote{Note that $\tilde{\varphi}^i$ is not considered as coordinate since it would involve a transformation which is not allowed in the phase space. It is merely a function defined for simplifying the form of metric. }. The collection of geometries with arbitrary \textit{wiggle} function $F$ can be considered  as a manifold $\bG$. In the next section, we build the NHEG phase space over $\bG$ by introducing a suitable symplectic form $\Omega$.
				
				As a cross-check one can readily observe that $\xi_-=\p_t$ and $\xi_0=t\p_t-r\p_r$ are isometries of the metric \eqref{g-F}. Moreover, one can check that for $|b|=1$, constant $v_b$ are null surfaces at which $\p_r$ becomes null.

				The induced metric over the compact surface of constant $t,r$, called ${\cal H}$ is
				\be\label{H-surface}
				ds^2_{\cal H}=\Gamma(\theta)\left[(1-b^2)d{\Psi}^2+d\theta^2+\gamma_{ij}(\theta)\,d\tilde{\varphi}^i\, d\tilde{\varphi}^j\right].
				\ee
				For a generic function $F(\vec{\varphi})$ (and hence ${\Psi}$), the above metric \eqref{H-surface} has the same  topology as the same surface on the background \eqref{NHEG-metric} metric if and only if $|b|=1$. This provides the justification for the condition \textbf{(4)} in last section.

				Finally we set $b=1$ and reexpress the main result of this section which is the generic metric of the NHEG phase space labeled uniquely by the wiggle function $F(\varphi^i)$
\begin{align}
				ds^2&=\Gamma(\theta)\Big[-\left( \boldsymbol\sigma -  d \Psi \right)^2+\Big(\dfrac{dr}{r}-d{\Psi}\Big)^2+d\theta^2+\gamma_{ij}(d\tilde{\varphi}^i+{k^i}\boldsymbol{\sigma})(d\tilde{\varphi}^j+{k^j}\boldsymbol{\sigma})\Big], \nn \\
				\boldsymbol{\sigma}&=e^{-{\Psi}}r\,d(t+\frac{1}{r})+\dfrac{dr}{r},\qquad \tilde{\varphi}^i=\varphi^i+k^i (F-{\Psi})\,,\qquad e^{\Psi}=1+\vec{k}\cdot\vec{\pd}_\varphi F .  \label{finalphasespace}
\end{align}
\subsection{Algebra of generators}

One can expand the periodic function $\epsilon(\vec{\varphi})$ in its Fourier modes:
\begin{equation}
\epsilon(\vec{\varphi})=- \sum _{\vec{n}} c_{\vec{n}}\,e^{- i(\vec{n}\cdot \vec{\varphi})}\,
\end{equation}
for some constants $c_{\vec{n}}$ and $\vec{n}\equiv (n_1,n_2,\dots,n_n)$, $n_i\in \mathbb{Z}$.\footnote{ The sign conventions are fixed such that the algebra takes the form \eqref{chi-algebra} and such that the central charge takes the form \eqref{CC}.} 
Therefore the generator $\chi$ decomposes as
\begin{equation}\label{Ln expansion}
\chi=\sum _{\vec{n}} c_{\vec{n}}\chi_{_{\vec{n}}}\,,
\end{equation}
where
\begin{equation}\label{Final Ln}
\chi_{_{\vec{n}}}=-e^{-i(\vec{n}\cdot\vec{\varphi})}\bigg(i(\vec{n}\cdot\vec{k})(\frac{1}{r}\partial_t+ r\partial_r)+\vec{k}\cdot\vec{\partial_{\varphi}} \bigg)\,.
\end{equation}
The Lie bracket between two such Fourier modes is given by
\begin{equation}\label{chi-algebra}
i \left[\chi_{_{\vec{m}}},\chi_{_{\vec{n}}}\right]_{L.B.}= \vec{k}\cdot(\vec{m}-\vec{n})\chi_{_{\vec{m}+\vec{n}}}\,.
\end{equation}
We will discuss the representation of this algebra by conserved charges in section \ref{sec-charges-and-algebra}.

\subsection{$SL(2,\mathbb{R})\times U(1)^{d-3}$ isometries of the phase space}
\label{sec:iso}

Since the whole phase space is constructed by coordinate transformations from the NHEG  background \eqref{NHEG-metric}, all geometries will still have the same isometries.  
The isometries in the phase space are defined by the pushforward of the background Killing vectors under the coordinate transformations. Explicitly,
\begin{align}
\bar{\xi}&=\bar{\xi}^\nu \frac{\p}{\p \bar{x}^\nu}=\left(\bar{\xi}^\nu \dfrac{\p x^\mu}{\p \bar{x}^\nu}\right)\frac{\p}{\p {x}^\mu}\,.\nonumber
\end{align}
As a result, the Killing vectors  of a generic geometry in the phase space are given by
\begin{align}
\xi^\mu&=\dfrac{\pd x^\mu }{\pd \bar{x}^\nu}\bar{\xi}^\nu
\end{align}
where $\bar{\xi}^\nu$ are defined in \eqref{xi1-xi2}. Note that the transformation matrix $\frac{\pd x^\mu }{\pd \bar{x}^\nu}$ is a function of $F(\vec{\varphi})$ and hence $\xi^\mu$ are \textit{field dependent} Killing vectors. This is qualitatively different from the symplectic symmetries $\chi$ which are field independent as shown in previous section. 

After a straightforward computation, the $SL(2,\mathbb{R})\times U(1)^{d-3}$ isometries are explicitly 
\begin{align}\label{m-barm}
\xi_- =\partial_t\,,\qquad \xi_0&=t\partial_t-r\partial_r,\qquad \xi_+ =\dfrac{1}{2}(t^2+\frac{1}{r^2})\partial_t-tr\partial_r-\frac{1}{r}{k}^i{\p}_{\varphi^i}+\dfrac{1}{r}\vec{k}\cdot\vec{\p}_\varphi (F-{\Psi}){\eta_+},\nonumber\\
\mathrm{m}_i&=\left(\delta^j_i-e^{-\Psi}k^j\,\p_i F\right)\p_{\varphi^j} +(\p_i {\Psi}-e^{-\Psi}\vec{k}\cdot\vec{\p}_\varphi  {\Psi}\,\p_i F) \,{\eta},
\end{align}
where $\eta = \eta_{b=1}$ is defined in \eqref{defetab}. As a consequence of the construction,  $\xi_-,\xi_0$ are not field dependent; i.e. they are independent of the function $F$, but other isometries are field dependent. 

\paragraph{Summary of the section.}  The NHEG phase space $\bG$ is a one-function  family of everywhere smooth metrics given in \eqref{finalphasespace}. These are obtained through finite coordinate transformations \eqref{finite ansatz} acting on the NHEG background \eqref{NHEG-metric}, which is the $F=0$ element in the phase space. All the metrics of the form \eqref{finalphasespace} have the same angular momentum and same parameters $\vec{k}$. and accordingly the same entropy. This  is schematically depicted in Fig \ref{fig phase space}.

\definecolor{light gray}{RGB}{220,220,220}
\begin{figure}[!bht]
	\captionsetup{width=0.8\textwidth}	
	\centering
	\begin{tikzpicture}[scale=1.2]
	\fill[even odd rule,light gray]
	(0,0) to (5,1.5) to  (11,1.5) to
	(6,0) to (0,0);
	\draw[very thick] (5.55,0) -- (5.55,-1) node[align=center,anchor=north] { };
	\draw[very thin,white] (1,0) to (6,1.5);
	\draw[very thin,white] (2,0) to (7,1.5);
	\draw[very thin,white] (3,0) to (8,1.5);
	\draw[very thin,white] (4,0) to (9,1.5);
	\draw[very thin,white] (5,0) to (10,1.5);
	\draw[very thin,white] (0.8,0.25) to (6.8,0.25);
	\draw[very thin,white] (1.6,0.5) to (7.6,0.5);
	\draw[very thin,white] (2.5,0.75) to (8.5,0.75);
	\draw[very thin,white] (3.3,1) to (9.3,1);
	\draw[very thin,white] (4.1,1.25) to (10.1,1.25);
	\draw[very thick,->] (5.55,0.75) -- (5.55,3) node[anchor=north west] {$J_i$};
	\fill[black] (5.55,0.75) circle (0.05) node[anchor=north] {\footnotesize $g[F=0]=\bar{g}$};
	\fill (8,1.25) circle (0.05) node[anchor=north] {\footnotesize $g[F]$};
	\draw (2,0) node[anchor=north] {$\bG$};
	\end{tikzpicture}
	\caption[Schematic depiction of the NHEG phase space]{\footnotesize  A schematic depiction of the NHEG phase space $\bG$. The vertical axis shows different background NHEG solutions of the form \eqref{NHEG-metric} specified by different angular momenta $J_i$, and the horizontal plane shows the phase space constructed by  the action of the finite coordinate transformation \eqref{finite ansatz}. Each geometry in the phase space is  identified by a periodic function $F(\vec{\varphi})$ and admits the same angular momenta $J_i$ and entropy.}\label{fig phase space}
\end{figure}
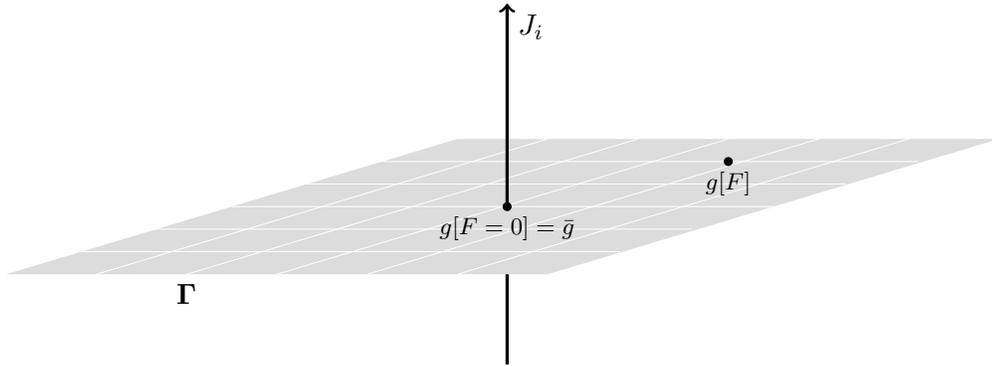

\section{Symplectic Structure}\label{sec-symplectic structure}

The set $\bG$ consisting of field configurations (metrics) \eqref{g-F}, can be viewed as a manifold, where each point of this manifold represents a metric $g[F]$ over the spacetime, determined by the functional form of the wiggle function $F[\varphi^i]$. In order for $\bG$ to be a \textit{phase space}, it should be accompanied by a \textit{symplectic structure} as duscussed in detail in section \ref{sec-cov-phase-space}. The aim of this section is to define a consistent symplectic form on the set of metrics \eqref{g-F}. 

In the particular case at hand, a complete basis of one-forms at any point of $\bG$, is given by the Lie derivative of fields with respect to generators $\chi_{\vec{n}}$ \eqref{Final Ln}. In other words, we can expand any variation $\de\Phi$ as
\begin{align}\label{tangent space basis}
\de\Phi&=\sum_{\vec{n}} c_{\vec{n}}\, \mathcal L_{\chi_{{\vec{n}}}} \Phi .
\end{align}


According to \eqref{symplectic form def}, the presymplectic structure is completely determined when its action on two arbitrary variations is known. Since according to \eqref{tangent space basis},  $\de_{\vec{m}}\Phi$ forms a complete basis for variations tangent to the phase space, the presymplectic form is completely determined once it is computed for two variations $\de_{\vec{m}}\Phi,\de_{\vec{n}}\Phi$ with arbitrary $\vec{m},\vec{n}$.  
Since any geometry in the phase space and any perturbation tangent to the phase space is invariant under $\xi_-,\xi_0$ \eqref{xi1-xi2}, also is the presymplectic current. Hence  there is no time dependence in the presymplectic structure and the radial dependence is fixed as
\begin{align}
\omega^t \propto \frac{1}{r}, \qquad \omega^r \propto r, \qquad \omega^\theta \propto r^0, \qquad \omega^{\varphi^i} \propto r^0.
\end{align}
Also, since constant $v = t + \frac{1}{r}$ surfaces are preserved in the phase space, one has
\bea
\omega^t &= \dfrac{1}{r^2} \omega^r. \label{omt}
\eea
The conservation of symplectic form requires that $\omega^r$ (or at least its integral) is vanishing at the boundary. The above equality then implies that $\omega^t$ will vanish as well, and the presymplectic structure computed on a constant time hypersurface will be zero on shell. Also note that  if $\omega^t$ is not vanishing, the presymplectic form will be logarithmically divergent. We will see in the next subsection that the usual Lee Wald symplectic current does not satisfy these properties.  Therefore we use the freedom in the definition of symplectic form to find one that is vanishing on shell so that it is finite and conserved. However, it is important to note that we impose these conditions only on-shell. If the presymplectic structure is zero on-shell but non-zero off-shell, it still allows  to define non-trivial surface charges and Poisson bracket.

\subsection{Lee-Wald symplectic structure}
\label{sec-cov-phase-space}

The standard presymplectic structure as defined by Lee-Wald is given by
\be\label{LW2}
{\boldsymbol \omega}_{(LW)}[\delta_1 \Phi ,\delta_2 \Phi,\Phi ] = \delta_1 {\boldsymbol  \Theta_{ref}}[ \delta_2 \Phi,\Phi ] - \delta_2 {\boldsymbol \Theta_{ref}}[\delta_1 \Phi, \Phi ] ,
\ee
where for Einstein gravity and for perturbations which preserve the $d$ dimensional volume, $h \equiv g^{\mu\nu}\delta g_{\mu\nu} = 0$, we have
\bea
\Theta_{ref}^\mu = \frac{1}{16 \pi G}\nabla_\nu h^{\mu\nu}.
\eea
It is straightforward to check that $\omega^r_{(LW)}$ is non-vanishing. Therefore the set of metrics $\bG $ equipped with the Lee-Wald symplectic structure does not define a well-defined phase space. 

More precisely, in four spacetime dimensions we find around the NHEG background 
\begin{align}
\begin{split}
\omega^t[\delta_m g,\delta_n g, \bar g] &= \frac{1}{r^2} \omega^r[\delta_m g,\delta_n g, \bar g] ,\\
\sqrt{-g}\omega^r[\delta_m g,\delta_n g, \bar g] &= \frac{ \Gamma (-1+k^2 \gamma)\;r}{8 \pi G \sqrt{\gamma}} e^{i (m+n)\varphi}\;k^2 m n (m-n)  (m+n-i b k \gamma), \\
\sqrt{-g}\omega^\theta[\delta_m g,\delta_n g, \bar g] &= -i\;\frac{  \Gamma \gamma'}{16 \pi G \sqrt{\gamma}}e^{i (m+n)\varphi}\;b\, k^3 m n (m-n),\\
\sqrt{-g}\omega^\varphi[\delta_m g,\delta_n g, \bar g] &= i\;\frac{ \Gamma  (-1+k^2 \gamma)}{8\pi G \sqrt{\gamma}}e^{i (m+n)\varphi}\;k^2 m  n (m-n).
\end{split}
\end{align}
Unfortunately the integral $\int_\Sigma \boldsymbol{\omega}[\de_m g,\de_n g,\bar{g}]$ over a constant $t$ surface $\Sigma$ is divergent for $m=-n \neq 0$. Also, $\omega^r\propto r$ so the boundary flux is not vanishing and in fact divergent. Moreover, according to \eqref{NHEK-solution} $\sqrt{\gamma}=\gamma_{11}\rightarrow 0$ at the poles $\theta=\{0,\pi\} $, and the Lee Wald symplectic current is divergent at the poles. In other words,  there is a line singularity along the $z$ axis of spacetime. We stress again that the only symplectic current which is not leaking through the boundary and meanwhile maintaining the $\xi_-,\xi_0$ isometry of the phase space, is the one that is vanishing on shell. In the following, we will find a suitable symplectuc form with these properties.

\subsection{Regularization of symplectic structure}\label{section: Y-terms}
As discussed around \eqref{omega with Y}, the presymplectic potential ${\boldsymbol  \Theta}[ \delta \Phi,\Phi ]$ is ambiguous up to the addition of boundary terms. The total presymplectic potential therefore has the form
\bea
\Theta^\mu[\delta \Phi ,\Phi] = \frac{1}{16 \pi G}\nabla_\nu h^{\mu\nu} + \nabla_\mu Y^{\mu\nu}.
\eea
where $Y^{\mu\nu} = Y^{[\mu\nu]}$ defines a $d-2$ form $\mathbf{Y}[\delta \Phi ,\Phi] $ which is linear in the field variations but non-linear in the fields. This leads to the total presymplectic form
\begin{align}\label{ambiguity omega}
{\boldsymbol{\omega}}[\delta_{1}\Phi,\delta_{2}\Phi,\Phi ] &=\boldsymbol{\omega}_{(LW)} [ \delta_{1}\Phi,\delta_{2}\Phi,\Phi]+d\Big(\delta_{1}\mathbf{Y} [\delta_{2}\Phi,\Phi]-\delta_{2}\mathbf{Y}[ \delta_{1}\Phi,\Phi ]\Big).
\end{align}
In the following, we will  define $\mathbf{Y}[\delta \Phi ,\Phi] $ in order to cure the problem mentioned in previous subsection . 

We first note that $\eta$ defined as
\begin{align}
\eta&=\dfrac{1}{r}\p_t+r\p_r
\end{align}
satisfies $[\eta , \chi] = 0 $. As a result of this, one can check that for any two variations tangent to the phase space $\delta_{ 1}\bar\Phi,\delta_{ 2}\bar\Phi$ we have
\begin{equation}\label{eta wald}
\mathcal{L}_{\eta_b}\boldsymbol{\omega}_{(LW)} [\delta_{ 1}\bar\Phi,\delta_{ 2}\bar\Phi,\bar\Phi]=\boldsymbol{\omega}_{(LW)} [\delta_{ 1}\bar\Phi,\delta_{ 2}\bar\Phi,\bar\Phi ].
\end{equation}
at any element of the phase space. This can be rewritten on-shell as 
	\begin{equation}\label{peq1}
	\boldsymbol{\omega}_{(LW)} [\delta_{ 1}\Phi,\delta_{ 2}\Phi,\Phi ] \approx d \left( \eta \cdot \boldsymbol{\omega}_{(LW)} [\delta_{ 1}\Phi,\delta_{ 2}\Phi,\Phi ]\right) 
	\end{equation}
	after using Cartan's identity  $\mathcal{L}_\eta X=\eta\cdot dX+d(\eta\cdot X)$ and recalling the fact that the presymplectic structure is closed on-shell, $d \boldsymbol{\omega} \approx 0$.
	
	Therefore, it is natural to define
	\bea
	\mathbf{Y}[\delta \Phi ,\Phi]  =- \eta \cdot \boldsymbol{\Theta}_{(LW)}[\delta \Phi ,\Phi]  + \mathbf{Y}_{comp}[\delta \Phi ,\Phi] \label{defY1}
	\eea
	and we obtain from \eqref{ambiguity omega} and \eqref{peq1}, 
	\begin{align} 
	{\boldsymbol{\omega}}[\delta_{1}\Phi,\delta_{2}\Phi,\Phi ] &=\boldsymbol{\omega}_{(LW)} [ \delta_{1}\Phi,\delta_{2}\Phi,\Phi]+d\Big(\delta_{1}\mathbf{Y} [\delta_{2}\Phi,\Phi]-\delta_{2}\mathbf{Y}[ \delta_{1}\Phi,\Phi ]\Big)\nnr
	&\approx d \left( \eta \cdot \boldsymbol{\omega}_{(LW)} [\delta_{ 1}\Phi,\delta_{ 2}\Phi,\Phi ] -\delta_1 (\eta \cdot  \boldsymbol{\Theta}_{(LW)}[\delta_2 \Phi ,\Phi] )\right)\nnr
	&+d \left(\delta_2 (\eta \cdot  \boldsymbol{\Theta}_{(LW)}[\delta_1 \Phi ,\Phi] )
	+  \delta_1 \mathbf{Y}_{comp}[\delta_2 \Phi ,\Phi] -\delta_2 \mathbf{Y}_{comp}[\delta_1 \Phi ,\Phi] \right)\nnr
	&\approx d (\delta_1 \mathbf{Y}_{comp}[\delta_2 \Phi ,\Phi] -\delta_2 \mathbf{Y}_{comp}[\delta_1 \Phi ,\Phi] )
	\end{align}
	where we used the fact that $\eta$ does not vary in the phase space. We therefore obtained that for any $ \mathbf{Y}_{comp}$ such that 
	\bea
	d (\delta_1 \mathbf{Y}_{comp}[\delta_2 \Phi ,\Phi] -\delta_2 \mathbf{Y}_{comp}[\delta_1 \Phi ,\Phi] ) \approx 0,\label{propZ}
	\eea
	the total symplectic structure is vanishing on-shell. A phase space therefore exists for the set of metrics \eqref{g-F} for all symplectic structures defined off-shell by \eqref{ambiguity omega}-\eqref{defY1}-\eqref{propZ}. In particular $\mathbf{Y}_{comp} = 0$ defines a symplectic structure. The fact that $\mathbf{Y}_{comp}$ is not fixed constitutes a remaining dynamical ambiguity that we need to fix through additional considerations. 

We fix the complementary term $Y_{comp}$ by requiring to have a ``correct'' central charge. By this we mean a central charge whose value seems to reproduce the correct entropy of black hole if one assumes a Cardy like formula. We do not present the details of the computaion here and only state the last result (see \cite{Compere:2015bca} for detailed argument)
\bea \label{Y-total}
\boxed{
(16\pi G){Y}^{\mn}[\delta\Phi,\Phi ]= \eta^{[\mu} \nabla_\rho h^{\nu ] \rho}- \left( \frac{1}{\Gamma} \de g_{\alpha\beta}\;\eta^{\;\alpha} \eta_2^{\;\beta }\right)\;{\epsilon}_\perp^\mn \, .
}
\eea
where $\overline \eta_2 = \frac{1}{\bar r}\p_{\bar t}$ and transform as a vector in phase space, and  ${\epsilon}_\perp^\mn$ is the binormal tensor to $\mathcal{H}$.

\section{Algebra of Charges: the ``NHEG Algebra''}\label{sec-charges-and-algebra}

In the previous sections we built the NHEG phase space and its symplectic structure. In this section, we show that the set of vector fields which generate the phase space indeed constitutes the set of symplectic symmetries and  analyze their conserved charges and their algebra. To this end, we first observe that any symplectic symmetry is integrable, namely it leads to well-defined charges over the phase space. We then  construct the algebra of charges and  provide an explicit representation of the charges in terms of a Liouville-type stress-tensor on the phase space.


 In the previous section we constructed $\boldsymbol{\omega}$ such that $\boldsymbol{\omega}[ \de_1 \Phi,\de_2 \Phi ,\Phi] \approx 0$  for any two perturbations around an arbitrary element of the phase space $\Phi$. This implies that the vectors $\chi$ are ``local symplectic symmetries'' of the NHEG phase space as defined in section \ref{local symplectic symmetries}. Therefore the charges can be defined at any codimension 2 compact surface $\cH$
\begin{align}\label{def charge variation}
\de H_\chi&= \oint_{\mathcal{H}} {\boldsymbol k}_\chi [\de\Phi,\Phi ]. 
\end{align}
The integrability of charges is also guaranteed by proposition \ref{prop local symplectic symmetries}.
\subsection{Algebra of charges}\label{sec-NHEG-Algebra}

Let us use the Fourier decomposition \eqref{Ln expansion}. We denote the surface charge associated with $\chi_{\vec{n}}$ as $H_{\vec{n}}$. As discussed in section \ref{sec:iso}, we also have the charges associated with the Killing vectors $\mathrm{m}_i$,  $J_i$, $i=1,\dots d-3$, and charges associated with  $SL(2,\mathbb R)$ Killing vectors $H_{\xi_{\pm,0}}$. $J_i$ are constant over the phase space and $H_{\xi_{\pm,0}}$ are vanishing.
The bracket between charges $H_{\vec{n}}$ is defined as 
\be\label{H-commutator}
\{H_{\vec{m}}, H_{\vec{n}}\}=\delta_{\vec{n}}H_{\vec{m}}=\oint_\mathcal{H}{\boldsymbol k}_{\chi_{\vec{m}}} [\delta_{\vec{n}}\Phi,{\Phi}],
\ee
for an arbitrary point in the phase space $\Phi$ and field variations $\delta_{\vec{n}}\Phi$. The right-hand side  is indeed anti-symmetric as a consequence of the integrability conditions.

Using the representation theorem \cite{Brown:1986ed,Barnich:2007bf} ( see the discussion around \eqref{representation of algebra}), the charges obey the same algebra as the symmetry generators \eqref{chi-algebra} up to a possible central term, i.e.
\be\begin{split}\label{calc central extension}
\{H_{\vec{m}}, H_{\vec{n}}\} &= -i\vec{k} \cdot (\vec{m}- \vec{n}) H_{\vec{m}+\vec{n}} + C_{\vec{m},\vec{n}}\,\cr
\{H_{\vec{p}}, C_{\vec{m},\vec{n}} \}& =\{H_{\vec{m}}, J_{i} \} =\{H_{\vec{m}}, H_{\xi_{\pm,0}} \} =0, \qquad \forall \vec{p},\vec{n},\vec{m}.
\end{split}
\ee
Note that the vanishing bracket between $H_{\vec{m}}$ and the angular momenta follows from either the fact that the angular momenta are constant, or from the fact that the vector fields $\mathrm{m}_i$ are Killing symmetries so that $\oint k_{\chi}[\mathcal L_{\mathrm{m}_i} g,g] = 0$. Even though the Lie bracket $[\chi,\mathrm{m}_i]_{L.B.} \neq 0$, the vanishing charge bracket is also consistent with the representation theorem since the total bracket $[\chi,\mathrm{m}_i] = [\chi,\mathrm{m}_i]_{L.B.} - \delta_\chi^g \mathrm{m}_i = 0$. The same reasoning holds for $H_{\xi_+}$. 

As mentioned before, the angular momenta $J_i$ and the $SL(2,\mathbb R)$ charges are constants over the phase space (the latter are in fact vanishing). To see this, we note that
\begin{equation}\label{Komar-integral}
\delta J_i =- \int_{\mathcal H} \boldsymbol k_{{\mathrm{m}}_i}[\delta_\chi \Phi, \Phi]=- \int_{\mathcal H} \boldsymbol k_{\bar{\mathrm{m}}_i}[\delta_\chi \bar \Phi, \bar \Phi]=0.
\end{equation}	
The second equality follows from general covariance of all expressions and the $\xi_-,\xi_0$ invariance which allows to freely move the surface $\mathcal H$, and the last equality is a result of the fact that $\bar \Phi$ is axisymmetric, and the only $\varphi^i$ dependence coming from $\chi$ makes the integral vanishing. This argument can also be repeated  for $SL(2,\mathbb R)$ charges.

The central extension $C_{\vec{n},\vec{m}}$ is a constant over the phase space which can be computed on the background using equations \eqref{H-commutator} and \eqref{calc central extension}. The generators can be shifted by constants in order to cancel terms proportional to $(\vec{m}-\vec{n})$ in the central term. In this case, it amounts to fixing the value of all charges on the NHEG background geometry $H_{\vec{n}} = 0$, $\forall \vec{n} \neq 0$ and $H_{\vec{0}}= \vec{k} \cdot \vec{J}$. The central extension is then found to be proportional to the entropy $S$, 
\begin{align}\label{CC}
 C_{\vec{m},\vec{n}}&= -i (\vec{k}\cdot \vec{m})^3 \frac{\hbar S}{2\pi}\,\delta_{\vec{m}+\vec{n},0},
\end{align}
{after multiplying and dividing by one power of $\hbar$, \cf section \ref{section-quantized-NHEG-algebra}.} 
The fact that entropy appears as the central element of the algebra dovetails with the fact that the area and therefore the entropy does not vary over the phase space.

Therefore we find the classical NHEG algebra
\begin{align}\label{NHEG-algebra}
 i\{H_{\vec{m}}, H_{\vec{n}}\}&= \vec{k} \cdot (\vec{m}- \vec{n}) H_{\vec{m}+\vec{n}} + (\vec{k}\cdot \vec{m})^3 \frac{\hbar  S}{2\pi}\,\delta_{\vec{m}+\vec{n},0} \,,\\
\label{Hn-J-Hxi} \{H_{\vec{m}}, J_{i} \} &=\{H_{\vec{m}}, H_{\xi_{\pm,0}} \}  = \{H_{\vec{m}}, S \} = 0.
\end{align}

\section{Charges on the phase space}\label{sec-charge-T}

As discussed earlier, the phase space  $\bG$ consists of the one-function family of metrics $g[F]$ given in \eqref{g-F} which is specified by the wiggle function $F(\vec{\varphi})$. This wiggle function defines an auxiliary quantity $\Psi$ defined in \eqref{Psi-def} which we will interpret in the following.

We have proven so far that the charges $H_{\vec{n}}$ are well-defined over phase space and that they obey the algebra \eqref{NHEG-algebra}. We now provide an explicit expression for the charges $H_{\vec{n}}$ as a functional of $\Psi$.
We can plug in the phase space metric and the symplectic symmetries $\chi_{\vec{n}}$ into the explicit formula for the charges in Einstein gravity in order to obtain the explicit expression for the charges $H_{\vec{n}}$. This computation  is explicitly performed in \cite{Compere:2015bca} with the result
\begin{align}\label{charge}
H_{\vec{n}}&= \oint_\mathcal{H} \boldsymbol{\epsilon}_\mathcal{H}\  T[\Psi] e^{-i \vec{n}\cdot \vec{\varphi} },
\end{align}
where $\boldsymbol{\epsilon}_\mathcal{H}$ is the volume form on $\mathcal H$ and
\begin{align}\label{resd4}
T[\Psi]&= \frac{ 1}{16 \pi G}  \Big( (\Psi' )^2 -2 \Psi''+ 2 e^{2 \Psi  } \Big)
\end{align}
where primes are directional derivatives along the vector $\vec{k}$, \ie $\Psi' = \vec{k}\cdot \vec{\partial}_\varphi \Psi$. 
The charges $H_{\vec{n}}$ are therefore the Fourier modes of $T[\Psi]$.

In order to understand this result, it is interesting to first note how the wiggle function $F$ transforms under a symplectic symmetry transformation generated by $\chi[\eps]$. To this end, we recall that by construction 
\be
\mathcal L_{\chi[\eps]}(g_{\mu\nu}[F]) = g_{\mu\nu}[F+\delta_\eps F]-g_{\mu\nu}[F].
\ee
We find 
\bea
\delta_{\eps} F=  (1 +\vec{k}\cdot \vec{\partial}_\varphi F )\eps=e^{{\Psi}}\eps.
\eea
The field $\Psi$ then transforms as
\bea\label{Psi-transform}
\delta_\eps {\Psi} =\eps\,  {\Psi}'+   \eps'.
\eea
where prime denotes again the directional derivative $\vec{k}\cdot \vec{\partial}_\varphi $. Therefore, $\Psi$ transforms like a Liouville field. In particular note that $\delta_\eps e^{{\Psi}}=(e^{{\Psi}}\eps)'$ and hence $e^{{\Psi}}$ resembles a ``weight one operator'' in the terminology of conformal field theory. It is then natural to define the Liouville stress-tensor
\begin{align}\label{T-tensor}
T[\Psi]&= \frac{ 1}{16 \pi G}  \Big( (\Psi' )^2 -2 \Psi''+ \Lambda e^{2 \Psi  } \Big)
\end{align}
with ``cosmological constant'' $\Lambda$ which transforms as 
\bea\label{T-transform}
\delta_\eps T = \eps  T' + 2  \eps' T - \frac{1}{8\pi G}  \eps'''.
\eea
Expanding in Fourier modes as in \eqref{charge}, it is straightforward to check from the transformations law \eqref{T-transform} that the algebra \eqref{NHEG-algebra} is recovered. Using the explicit computation for the surface charges \eqref{charge} we identify the cosmological constant to be $\Lambda = 2$. 

The above resembles the transformation of the energy momentum tensor, a ``quasi-primary operator of weight two''. However, we would like to note that ${\Psi}$ and hence $T[\Psi]$ are not function of time but are functions of all coordinates $\varphi^i$, in contrast with the standard Liouville theory.

Given \eqref{charge} and \eqref{resd4}, one can immediately make the following interesting observation: The charge associated with the zero mode $\vec{n}=0$, $H_{\vec{0}}$, is positive definite over the whole phase space. This is due to the fact that the $\p^2 \Psi$ term does not contribute to $H_{\vec{0}}$ and the other two terms in \eqref{T-tensor} give positive contributions.

\subsection{Quantization of algebra of charges: The \textit{NHEG algebra} }
\label{section-quantized-NHEG-algebra}

Since the symplectic structure is nontrivial off-shell and the resulting surface charges are integrable, we were able to define physical surface charges $H_{{\vec{n}}}$ associated with the symplectic symmetries $\chi[\eps_{\vec{n}}]$, where $\eps_{\vec{n}}=e^{-i \vec{n} \cdot \vec{\varphi}}$, $n_i\in \mathbb{Z}$. The generators of these charges satisfy the same algebra as $\chi$ themselves, but with the entropy as the central extension in \eqref{NHEG-algebra}. One can use the Dirac quantization rules
\be\label{quantization}
\{\quad\}\to \frac{1}{i\hbar}[\quad]\,,\qquad\mathrm{and}\qquad  H_{\vec{n}}\to \hbar \,L_{\vec{n}},
\ee 
to  promote the symmetry algebra to an operator algebra, the \emph{NHEG algebra} $\widehat{\mathcal{V}_{\vec{k},S}}$
\bea\label{NHEG-algebra quantized}
[L_{\vec{m}}, L_{\vec{n}}] = \vec{k} \cdot (\vec{m}- \vec{n}) L_{\vec{m}+\vec{n}} +  \frac{S}{2\pi}(\vec{k}\cdot \vec{m})^3 \delta_{\vec{m}+\vec{n},0}\,.
\eea
The angular momenta  $J_i$ and the entropy $S$ commute with $L_{\vec{n}}$, in accordance with \eqref{Hn-J-Hxi}, and are therefore central elements of the NHEG algebra  $\widehat{\mathcal{V}_{\vec{k},S}}$. Explicitly, the \emph{full symmetry of the phase space} is
\be\label{full-algebra}
\mathrm{Phase\ Space\ Symmetry\ Algebra}=\mathrm{sl}(2,\mathbb{R})\oplus \mathrm{u}(1) \;``d-3\; \text{times}" \;\oplus \nhegalgebra.
\ee
We reiterate that all geometries in the phase space have vanishing $SL(2,\mathbb{R})$ charges and $U(1)$ charges equal to $J_i$.

\paragraph{The case $d=4$.}
For the four dimensional {Kerr case},  $k=1$ and one obtains the familiar Virasoro algebra
\bea\label{Virasoro}
[L_m, L_n  ]= (m-n) L_{m+n} + \frac{c}{12}m^3\delta_{m+n,0}
\eea
with central charge $c= 12 \frac{S}{2 \pi}=\frac{12J}{\hbar}$, as in  \cite{Guica:2008mu}. We indeed fixed the dynamical ambiguity in the definition of the symplectic structure in order that the resulting central charge be independent of the choice of constant $b$ in the definition of the generator. Since $b=0$ corresponds to the Kerr/CFT generator, we reproduce their central charge.

\paragraph{The cases $d>4$.} In higher dimensions, the NHEG algebra  $\widehat{\mathcal{V}_{\vec{k},S}}$ \eqref{NHEG-algebra} is a more general infinite-dimensional algebra in which the entropy appears as the central extension.
For $d>4$ the NHEG algebra contains infinitely many Virasoro subalgebras. To see the latter, first we note that vectors $\vec{n}$ construct a $d-3$ dimensional lattice. $\vec{k}$ may or may not be on the lattice.  Let $\vec{e}$ be any given vector on this lattice  such that $\vec{e}\cdot \vec{k}\neq 0$. Consider the set of generators $L_{\vec{n}}$ such that $\vec{n}=n \vec{e}$. Then one may readily observe that these generators form a Virasoro algebra of the form \eqref{Virasoro}. If we define
\begin{equation}\label{Virasoro-subalgebra-generators}
\ell_n\equiv \frac{1}{\vec{k}\cdot \vec{e}} L_{\vec{n}}\,,
\end{equation}
then
\begin{align}
[\ell_m,\ell_n]&=[\frac{L_{\vec{m}}}{\vec{k}\cdot \vec{e}} ,\frac{L_{\vec{n}}}{\vec{k}\cdot \vec{e}} ]
=\frac{\vec{k}\cdot(\vec{m}-\vec{n})}{\vec{k}\cdot \vec{e}}\frac{L_{\vec{m}+\vec{n}}}{\vec{k}\cdot \vec{e}}+\frac{(\vec{k}\cdot \vec{m})^3}{(\vec{k}\cdot \vec{e})^2} \frac{S}{2\pi} \,\delta_{\vec{m}+\vec{n},0}\cr
&=(m-n)\ell_{m+n}+\frac{c_{\vec{e}}}{12}m^3 \,\delta_{m+n,0}\,.
\end{align}
As a result, the central charge for the selected subalgebra would be:
\begin{equation}\label{virasoro-subalgenra-central-charge}
c_{\vec{e}}=12 (\vec{k}\cdot \vec{e}) \frac{S}{2\pi}\,.
\end{equation}
The entropy might then be written in the suggestive form $S= \frac{\pi^2}{3} c_{\vec{e}} \,T_{F.T.}$ where
\begin{equation}
T_{F.T.}=\frac{1}{2\pi(\vec{k}\cdot \vec{e})}
\end{equation}
is the { extremal Frolov-Thorne chemical potential associated with $\vec{e}$}, as reviewed in \cite{Compere:2012jk}.

We also comment that  $\widehat{\mathcal{V}_{\vec{k},S}}$ contains many Abelian subalgebras spanned by generators of the form $L_{\vec{n}}$ where $\vec{n} = n \vec{v} $ and $\vec{v} \cdot \vec{k} = 0$, if $\vec{v}$ is on the lattice.

\section{Discussion}\label{sec-NHEG-discussion}

In this chapter we introduced a consistent phase space for near-horizon rotating extremal geometries in four and higher dimensions which we dubbed the NHEG phase space \cite{Compere:2015bca,Compere:2015mza}. We identified its symmetries as a direct product of the  $SL(2,\mathbb{R})\times U(1)^{d-3}$ isometries and a class of symplectic symmetries. The symplectic symmetries form a novel generalized Virasoro algebra which we dubbed the NHEG algebra and denoted as $\nhegalgebra$. The phase space is  generated by diffeomorphisms corresponding to the symplectic symmetries. All elements of the phase space have the same angular momenta and entropy. We will comment below on various aspects of our construction, on the comparison with existing literature and on possible future directions. 

\paragraph{Comments on the NHEG algebra.}
One of our results is the representation of the infinite dimensional NHEG algebra $\nhegalgebra$ \eqref{NHEG-algebra quantized} in the phase space of near-horizon geometries. Its structure constants are specified by the vector $\vec{k}$ obtained from the near-extremal expansion of the black hole angular velocity $\vec{\Omega} = \vec{\Omega}_{ext} + \frac{2 \pi}{\hbar} \vec{k} \, T_H + O(T_H^2)$. The central charge is given by the black hole entropy $S$. 
The generators of the isometries $SL(2,\mathbb{R})\times U(1)^{d-3}$ commute with the generators $L_{\vec{n}}$ and hence the Killing charges are trivial in this phase space. However, we expect a generalization of the phase space that excitations of Killing charges are also included as well.

Generalized or higher rank Virasoro algebras have been considered in the mathematics literature \cite{patera1991higher,mazorchuk1998unitarizable,guo2012classification} but to our knowledge none of these algebras depends upon a real vector $\vec{k}$. It is desirable to explore further various interesting mathematical aspects of this algebra, including its unitary representations, the corresponding group manifold and its coadjoint orbits \cite{Javadinedjad:2016}.

\paragraph{NHEG phase space vs Kerr/CFT.} Our construction shares several features with the original Kerr/CFT proposal \cite{Guica:2008mu}. We both use covariant phase space methods to describe the microscopics of extremal rotating black holes and (at least) a Virasoro algebra appears as a symmetry algebra. However, we would like to emphasize that our results are both conceptually and technically distinct from the Kerr/CFT proposal. 

\begin{enumerate}

\item As a consequence of invariance under two out of the three generators of $SL(2,\mathbb R)$, the NHEG phase space admits a transitive action which maps any codimension two surface at fixed $t_{\mathcal H},r_{\mathcal H}$ to another such surface at fixed $t,r$. Therefore, surface charges are not only defined at infinity but rather on any sphere $t,r$ in the bulk of spacetime, which leads to the feature that symmetries are symplectic instead of only asymptotic.

\item We explicitly construct the phase space, with a consistent symplectic structure,  and specify the set of smooth metrics. Specifying the phase space in the Kerr/CFT setup has faced various issues, including non-smoothness of the candidate metrics at the poles \cite{Guica:2008mu,Amsel:2009ev}. We resolve these issues here thanks to the change of symmetry ansatz. 

\item Our construction in higher dimensions than four provides a democratic treatment of all $U(1)$ directions by implementing the vector $\vec{k}$ that already exists in the metric \eqref{NHEG-metric}. This is advantageous to generalizations of Kerr/CFT to 5 dimensions that singled out a special circle on the torus and build the Virasoro along it.\cite{Compere:2009dp}

\end{enumerate}

\paragraph{Dynamical ambiguity and central charge.} As our construction shows, the symplectic structure is determined upon the addition of a specific class of boundary terms which might contribute to the central charge. One possible more solid way to fix these boundary terms would be to study the boundary terms necessary to obtain a well-defined variational principle \cite{Compere:2008us}.

\paragraph{Conserved charges from a Liouville-type stress-tensor.} The phase space is labelled by the periodic wiggle function $F(\vec{\varphi})$ over the $d-3$ dimensional torus which allows defining the periodic function $\Psi$. We showed that the charges defined over the phase space can be expressed in terms of the Fourier modes  of the functional $T[\Psi]$ \eqref{T-tensor} over the torus. The functional $T[{\Psi}]$ has a striking resemblance to (a component of) the energy-momentum tensor of a Liouville field theory. However, there are also major differences since there is no time dependence here and instead there are multidimensional circle directions. 
While the relationship between $3d$ Einstein gravity and Liouville theory is well understood using the Chern-Simons formulation \cite{Coussaert:1995zp}, to our knowledge, it is the first occurrence of a connection between four and higher dimensional gravity and Liouville theory. We also remark that the zero mode of the NHEG algebra $H_{\vec{0}}$  is positive definite over the whole phase space. Therefore, one might be tempted to use $H_{\vec{0}}$  as a defining Hamiltonian for such a Liouville-type theory. It is natural to ask where such a ``holographically dual'' theory would be defined. In that regards, we note that a special role in the construction is played by one null shear-free rotation-free and expansion-free geodesic congruence \cite{Durkee:2010ea} which is kept manifest in the phase space and thereby provides a natural class of null ``holographic screens''.

\paragraph{Comparison with $3d$ Einstein gravity.} 

In section \ref{sec decoupling sector}, we studied the near horizon limit of extremal BTZ and its descendants. This is analogous to the NHEG phase space we discussed in this section. They both have a global \sltr isometry and an infinite dimensional symplectic symmetry. While the Killing vectors are the same, there is a discrepancy in the form of vectors generating symplectic symmetries. The $t$ component of \eqref{NH-chi-Banados} involves second derivative, while the same component in \eqref{ASK} involves first derivative. The latter was fixed by the condition that compact surfaces $\cH$ over which the charges are defined remain smooth near its poles $\theta=\theta_*$ over any configuration of the phase space . However, in three dimensions, the integration surface $\cH$ is a circle and therefore there is no $\theta$. We expect that the symmetry vectors for the phase space of extremal black \textit{rings} involve second derivatives as \eqref{NH-chi-Banados}. 


In AdS$_3$ gravity the Virasoro central charge depends upon the theory but does not depend upon the physical parameters of the black hole solution, unlike the higher dimensional case where the entropy, an intrinsic property of the NHEG solution, appears as the central charge. This feature is therefore radically different in $3d$ as compared with higher dimensions. 

A natural question is if, like the AdS$_3$ case, there exists  a bigger algebra which contains the physics before taking the near-horizon limit and/or physics beyond extremality. The AdS$_3$ example, then suggests that such a generalization may require a ``non-chiral'' extension of the NHEG algebra; e.g. by doubling it with left-movers, which is frozen out as a result of extremality and the near-horizon limit.

\paragraph{Extension to other near-horizon extremal geometries.} In this work 
we focused on the specific example of $d$ dimensional Einstein vacuum solutions with $SL(2,\mathbb{R})\times U(1)^{d-3}$ isometry. More general near horizon geometries exist and we expect our construction to be extendible to any such geometries.

\paragraph{Possible relationship with black hole microstates.} Understanding the microstates of extremal black holes was our main motivation in this study. The existence of a large symmetry algebra in near-horizon geometries together with the application of Dirac semi-classical quantization rules, imply that black hole quantum states, whatever they might be, form a representation of the quantized NHEG algebra $\nhegalgebra$ \eqref{NHEG-algebra quantized}. A stronger statement would be that the low energy description of these microstates is entirely captured by a quantization of the phase space (which might be possible thanks to the existence of a symplectic structure). If such a low energy description is available, $H_{\vec{0}}$ would appear as the natural ``Hamiltonian'' governing the dynamics on this Hilbert space.

\chapter{Summary and Outlook}
In this thesis, we stressed on the concept of ``surface degrees of freedom" in gauge theories and especially gravity. These are a class of ``would be" gauge degrees of freedom that become physical when they are large enough near the boundary of spacetime. We avoid calling them ``boundary" degrees of freedom (as is usually done) since these states are not localized at the boundary and their effect can be captured even deep inside the bulk of spacetime. Instead  ``surface" refers to the fact that they produce a \textit{lower} dimensional dynamics in the theory. 

Surface degrees of freedom show up naturally in the Hamiltonian formulation of gauge theories, or its covariant version, the covariant phase space. These formulations were studied in chapters \ref{chapter-Hamiltonian} and \ref{chapter-covariant phase space}. In these formulations, surface degrees of freedom are generated by local transformations that correspond to \textit{second class} constraints while gauge transformations are generated by first class constraints. In other words, they correspond to nontrivial ``charge". The charges commute with the constraints and hence are new observables that can be used to detect and distinguish these new states. A subtlety appears in the relation between charge and generator in spacetimes with disconnected boundaries. This issue seems to be not fully addressed in the literature.

We showed in chapters \ref{chapter-AdS3} and \ref{chapter NHEG phase space} that surface gravitons indeed play an important role in interesting problems in gravitational physics. Einstein gravity with negative cosmological constant in three dimensions, lacks propagating degrees of freedom. Accordingly nontrivial solutions originate from I) solutions with nontrivial topology which can be obtained by discrete identifications of global \ads solution, and II) surface degrees of freedom. Local dynamics can be added to these states by either adding matter, or modifying the Einstein Hilbert action. We also discussed that surface degrees of freedom can serve as (at least part of) the microstates of the BTZ black hole, producing its entropy at the semiclassical level.

In four and higher dimensions, we investigated the role of surface gravitons in relation with extremal black hole. We showed that restricting to the throat geometry of extremal black holes, again there is no local bulk dynamics respecting the boundary conditions. Therefore the situation is much similar to the case of \ads. As a result, we argued that the only dynamics is produced by surface gravitons. The phase space thus constructed, has a rich dynamics and symmetry structure \footnote{Regarding the surface degrees of freedom, the more symmetries, the richer dynamics we have. This is contrary to the usual case where symmetries can be used to restrict the dynamics.}. 

We argued that these surface gravitons can be related to the microscopic description of extremal black hole. The idea is that their dynamics is governed by a lower dimensional field theory which can be considered as the holographic dual theory. We took some steps in that direction by showing that the charges can be obtained through a (component of a) stress tensor similar to that of a Liouville field theory (see section \ref{sec-charge-T}). We also showed that the entropy of the extremal black hole and the central charge of the algebra  fit into a Cardy like formula. 

The dual field theory mentioned above need to be investigated more. Suggestions in this directions are
\begin{itemize}
	\item The charges exhibit only one component of the dual stress tensor. The full stress tensor can be obtained by computing the Brown York quasi local stress tensor\cite{Brown:1992br}. This is similar to what was done in \cite{Balasubramanian:1999re} for asymptotic AdS geometries.
	\item Another approach is to formulate the problem using the first order formulation of tetrads. Recall that in three dimensions this approach led to the Chern Simons formulation of gravity \cite{Achucarro:1987vz} that induced a Liouville theory on the boundary \cite{Coussaert:1995zp}.
	\item The fields in the field theory dual the the NHEG phase space are representations of the NHEG algebra. Accordingly, it is necessary to build and classify the (unitary)  representations of this algebra. 
	\item If it is possible to find a counterpart of Cardy formula in field theories realizing NHEG algebra, then we can give in principle give a \textit{microscopic counting} of extremal black hole in higher dimensions. 
	\item The NHEG phase space constructed here, is isomorphic to one \textit{orbit} in the ``coadjoint representation" of the NHEG algebra. Classification of the coadjoint orbits of NHEG algebra hence will be important in extending the phase space, or to find other phase spaces with similar symmetry group. This is similar to the coadjoint representations of Virasoro algebra that we mentioned in section \ref{Banados-Orbits-sec}. This classification is under investigation in \cite{Javadinedjad:2016}.
\end{itemize}
Recall that the phase space of asymptotic \ads geometries admits two types of field variations: I) variations generated by an infinitesimal diffeomorphism (surface gravitons), and II) parametric variations (like a perturbation over BTZ to a nearby BTZ with slightly different parameters). However the NHEG phase space constructed here admits only variations of the first type. We expect that this phase space can be generalized to contain NHEGs with different angular momenta and entropy. However, this is not straightforward, since the NHEG algebra as the symmetry algebra of the phase space contains the information about a specific NHEG and therefore cannot be the symmetry algebra of the extended phase space. This extension is therefore an interesting open problem.

Surface gravitons can play even a more significant role in black hole physics. Recently Hawking, Perry and Strominger \cite{Hawking:2016msc} proposed that surface degrees of freedom on the horizon of black hole geometry can even solve the information paradox and the unitarity of black hole formation and evaporation can be restored by considering the information encoded in surface gravitons. Let us briefly explain their idea. 

Black hole information paradox is based on the uniqueness theorem of black holes. That is black holes are uniquely determined by few parameters like mass, angular momenta and electric/magnetic charges. Therefore all information about the initial collapsing matter is lost. However, the uniqueness theorem determines the black hole geometry \textit{up to} diffeomorphisms. Accordingly surface gravitons can be considered as the ``soft hairs of black hole". They argue that the existence of zero energy surface gravitons, implies that the vacuum in quantum gravity is not unique and indeed has infinite degeneracy. The formation and evaporation of black hole then corresponds to a transition between these degenerate vacua. The final state may in principle be correlated with the Hawking radiation in such a way as to maintain the unitarity of evolution.


\appendix
\chapter{Technical Proofs.}
\section{Proof of proposition \ref{prop local symplectic symmetries}}\label{appendix proofs}
Proposition \ref{prop local symplectic symmetries} states that the set of local symplectic symmetries $\chi$ satisfying $\bomega [\delta \Phi , \delta_\chi \Phi,\Phi ]\approx 0$, the following properties are satisfied
\begin{enumerate}
\item they form a closed algebra.
\item their corresponding charge is integrable.
\item Their corresponding charge can be computed over any codimension 2 closed surface in the bulk that can be continuously deformed from the asymptotics.
\end{enumerate}

\begin{proof} It was shown in \ref{section-conserved charges} that the on-shell symplectic current contracted with a Lie variation along any vector $\chi$ is an exact form, that is
\begin{align}
 	\bomega [\delta \Phi , \delta_\chi \Phi,\Phi ] = d \,{\boldsymbol k}_\chi [\delta\Phi,\Phi]
\end{align}
when $\Phi,\de\Phi$ are solutions to the field equations and the linearized field equations respectively. The existence of local symplectic symmetries depends on the following condition
\begin{align}\label{LSS condition}
{\boldsymbol k}_\chi [\delta\Phi,\Phi]&=\de \mathcal{H}_\chi [\Phi]+ N[\chi,\Phi]\,,\qquad dN[\chi,\Phi]=0
\end{align}
This is a local condition over spacetime and is not satisfied in any spacetime. The logic is that if this condition is satisfied for a specific phase space, then local symplectic symmetries exist and satisfy the above mentioned properties. Therefore in the following, we assume that \eqref{LSS condition} is satisfied. First we present a corollary of \eqref{LSS condition}: For two local symplectic symmetries $\chi,\eta$ satisfying \eqref{LSS condition}, we have
\begin{align}
\de_\chi \cH_\eta&=\cH_{[\eta,\chi]}+M[\Phi,\chi,\eta]\,,\qquad d\de M[\Phi,\chi,\eta]=0
\end{align}
To see this, note that 
\begin{align}
\cL_\chi \cH_\eta&=\de_\chi \cH_\eta + \cH_{\cL_\chi\eta}
\end{align}
This is because $\de_\chi$ acts like a Lie derivative on dynamical fields $\Phi$ but commutes with non-dynamical fields, which is $\eta$. The last term on the r.h.s then compensates this difference by noticing the fact that $\cH_\eta$ is linear in $\eta$. By definition $[\chi,\eta]=\cL_\chi\eta$, hence we find that 
\begin{align}
\de_\chi \cH_\eta&=\cH_{[\eta,\chi]} +\cL_\chi \cH_\eta
\end{align}
It is enough to show that $d\de \cL_\chi \cH_\eta=0$.
\begin{align}
d\cL_\chi \cH_\eta&=d(\chi\cdot d\cH_\eta+d(\chi\cdot \cH_\eta))=d(\chi\cdot d\cH_\eta)
\end{align}
Therefore 
\begin{align}
\de\, d\,\cL_\chi \cH_\eta&=\de\, d\,(\chi\cdot d\cH_\eta)=d\,(\chi\cdot d\de\cH_\eta)\nnr
&=d\,(\chi\cdot d{\boldsymbol k}_\eta [\delta\Phi,\Phi])=d\,(\chi\cdot\bomega[\delta \Phi , \delta_\chi \Phi,\Phi ])
\end{align}
Since $\chi$ is by definition a local symplectic symmetry, then the right hand side is zero and we obtain the desired property.\\
\noindent\textbf{Closedness. }The first property states that if $\chi,\eta$ are local symplectic symmetries, then their Lie bracket is also a local symplectic symmetry; that is 
\begin{align}
\bomega [\delta \Phi , \delta_{[\eta,\chi]} \Phi,\Phi ]\approx 0
\end{align}
The proof is that 
\begin{align}
\bomega [\delta \Phi , \delta_{[\eta,\chi]} \Phi,\Phi ]&=d \,\de\, \cH_{[\eta,\chi]} [\Phi]\\
&=d \,\de \,(\de_\chi\cH_{\eta} -\cL_\chi \cH_\eta)=d \,\de\, \de_\chi\cH_{\eta}\\
&=d\, \de\, \bomega [\delta_\chi \Phi , \delta_{\eta} \Phi,\Phi ]\approx 0
\end{align}
\\
\noindent\textbf{Integrability.} A charge $H_\chi$ is integrable if 
\begin{align}
\de_1\de_2 H_\chi -\de_2\de_1 H_\chi&=0
\end{align}
For a symplectic symmetry, according to \eqref{LSS condition} the charge variation is $\de$ exact
\begin{align}
\de H_\chi&=\oint \de \cH_\chi
\end{align}
Since the charge variation is $\de$ exact, it is trivially $\de$ closed and therefore the integrability follows.\\\\
\textbf{Charges everywhere. }According to \eqref{charge difference}, the difference between charges computed at two different closed surfaces is given by an integral of the symplectic current. By using the property \eqref{symp-sym-def} for local symplectic symmetries, the relevant component of symplectic current vanishes on-shell, and we directly see that the charges \textit{a priori} defined at infinity, can be computed over any closed codimension 2 surface in the bulk. However, note that the surface should be obtainable from a continuous deformation of the surface at infinity. In other words, it should preserve holonomies. This will be important \eg in the case of multi BTZ solutions, where the geometry has a couple of disconnected ``inner'' boundaries.

\end{proof}

\section{Charge expression for field dependent gauge transformations}\label{sec field dependent vec charges}

Assume we have a vector $\chi$ which is a function of the dynamical fields $\Phi$ such as the metric. In our example, the metric dependence reduces to  $\chi=\chi(L_+,L_-)$. We call this a field dependent vector. We want to find the corresponding charge $\delta \boldsymbol{Q}_\chi$ and the integrability condition for such vectors.
We proceed using the approach of Iyer-Wald \cite{Iyer:1994ys} and carefully keep track of the field dependence. We adopt the convention that $\delta\Phi$ are Grassman even. First define the Noether current associated to the vector $\chi$ as
\begin{align}
\boldsymbol{J}_\chi[\Phi]&=\boldsymbol{\Theta}[\delta_\chi\Phi,\Phi]-\chi\cdot \boldsymbol{L}[\Phi],
\end{align}
where $\boldsymbol{L}[\Phi]$ is the Lagrangian (as a top form), and $\boldsymbol{\Theta}[\delta_\chi\Phi,\Phi]$ is equal to the boundary term in the variation of the Lagrangian, i.e $\de \boldsymbol{L}=\frac{\delta L}{\delta \Phi}\delta \Phi + d\boldsymbol{\Theta}[\delta\Phi,\Phi]$.  Using the Noether identities one can then define the on-shell vanishing Noether current as $\frac{\delta L}{\delta \Phi}\mathcal L_\chi \Phi = d \boldsymbol S_\chi[\Phi]$. It follows that $\boldsymbol J_\chi + \boldsymbol S_\chi$ is closed off-shell and therefore $\boldsymbol{J}_\chi\approx d \boldsymbol{Q}_\chi$, where $\boldsymbol{Q}_\chi$ is the Noether charge density (we use the symbol $\approx$ to denote an on-shell equality).
Now take a variation of the above equation
\begin{align}
\delta \boldsymbol{J}_\chi&=\delta \boldsymbol{\Theta}[\delta_\chi\Phi,\Phi]-\delta(\chi\cdot \boldsymbol{L})\cr
&=\delta \boldsymbol{\Theta}[\delta_\chi\Phi,\Phi]-\chi\cdot \delta \boldsymbol{L}-\delta\chi\cdot \boldsymbol{L} \cr
&\approx\delta \boldsymbol{\Theta}[\delta_\chi\Phi,\Phi]-\chi\cdot d\boldsymbol{\Theta}[\delta \Phi,\Phi]-\delta\chi\cdot\boldsymbol{L}\,.
\end{align}
Using the Cartan identity  $\mathcal{L}_\chi \boldsymbol{\sigma}=\chi\cdot d\boldsymbol{\sigma}+d(\chi\cdot \boldsymbol{\sigma})$ valid for any vector $\chi$ and any form $\boldsymbol{\sigma}$, we find
\begin{align}
\delta \boldsymbol{J}_\chi&=\Big(\delta \boldsymbol{\Theta}[\delta_\chi\Phi,\Phi]-\delta_\chi\boldsymbol{\Theta}[\delta \Phi,\Phi]\Big)+d(\chi\cdot\boldsymbol{\Theta}[\delta\Phi,\Phi])-\delta\chi\cdot \boldsymbol{L}.
\end{align}
The important point here is that
\begin{align}\label{delta vs deltaPhi}
\delta \boldsymbol{\Theta}[\delta_\chi\Phi,\Phi]&= \delta^{[\Phi]} \boldsymbol{\Theta}[\delta_\chi\Phi,\Phi]+ \boldsymbol{\Theta}[\delta_{\delta\chi}\Phi,\Phi]\,,
\end{align}
where we define $\delta^{[\Phi]}$ to act only on the explicit dependence on dynamical fields and its derivatives, but not  on the implicit field dependence in $\chi$. Therefore, we find
\begin{align}\label{presymplectic LW}
\nonumber \delta \boldsymbol{J}_\chi&=\Big(\delta^{[\Phi]} \boldsymbol{\Theta}[\delta_\chi\Phi,\Phi]-\delta_\chi\boldsymbol{\Theta}[\delta \Phi,\Phi]\Big)+d(\chi\cdot\boldsymbol{\Theta}[\delta\Phi,\Phi])+\Big(\boldsymbol{\Theta}[\delta_{\delta\chi}\Phi,\Phi]-\delta\chi\cdot \boldsymbol{L}\Big)\\
&=\boldsymbol{\omega}^{LW}[\delta\Phi\,,\,\delta_\chi\Phi\,;\,\Phi] + d(\chi\cdot\boldsymbol{\Theta}[\delta\Phi,\Phi])+\boldsymbol{J}_{\delta\chi}\,,
\end{align}
where
\be
\boldsymbol{\omega}^{LW}[\delta\Phi\,,\,\delta_\chi\Phi\,;\,\Phi]=\delta^{[\Phi]} \boldsymbol{\Theta}[\delta_\chi\Phi,\Phi]-\delta_\chi\boldsymbol{\Theta}[\delta \Phi,\Phi],
\ee
is the Lee-Wald presymplectic form \cite{Lee:1990nz}. Note that the variation acting on $\boldsymbol{\Theta}[\delta_\chi\Phi,\Phi]$, only acts on the explicit field dependence. This is necessary in order for $\boldsymbol{\omega}^{LW}[\delta\Phi\,,\,\delta_\chi\Phi\,;\,\Phi]$ to be bilinear in its variations. Reordering the terms we find
\begin{align}
\nonumber \boldsymbol{\omega}^{LW}[\delta\Phi\,,\,\delta_\chi\Phi\,;\,\Phi]&=\delta \boldsymbol{J}_\chi-\boldsymbol{J}_{\delta\chi}-d(\chi\cdot\boldsymbol{\Theta}[\delta\Phi,\Phi])\\
&=\delta^{[\Phi]}\boldsymbol{J}_\chi -d(\chi\cdot\boldsymbol{\Theta}[\delta\Phi,\Phi]).
\end{align}

If $\delta \Phi$ solves  the linearized field equations, then $\boldsymbol{J}_\chi\approx d \boldsymbol{Q}_\chi$ implies  $\delta^{[\Phi]}\boldsymbol{J}_\chi \approx d(\delta^{[\Phi]}\boldsymbol{Q}_\chi)$. As a result we obtain
\bea\label{omega vs charge}
\boldsymbol{\omega}^{LW}[\delta\Phi\,,\,\delta_\chi\Phi\,;\,\Phi]& \approx d \boldsymbol{k}_\chi^{IW} [\delta\Phi; \Phi]
\eea
where $ \boldsymbol{k}_\chi^{IW}$ is the Iyer-Wald surface charge form
\bea
 \boldsymbol{k}_\chi^{IW} = \Big(\delta^{[\Phi]}\boldsymbol{Q}_\chi-\chi\cdot\boldsymbol{\Theta}[\delta\Phi,\Phi]\Big).
\eea
Therefore the infinitesimal charge associated to a field dependent vector and a codimension two, spacelike compact surface $S$ is defined as the Iyer-Wald charge
\begin{align}\label{charge variation}
\delta H_\chi &= \oint_S \boldsymbol{k}_\chi^{IW} [\delta\Phi; \Phi]= \oint_S \Big(\delta^{[\Phi]}\boldsymbol{Q}_\chi-\chi\cdot\boldsymbol{\Theta}[\delta\Phi,\Phi]\Big).
\end{align}
The key point in the above expression is that the variation does not act on $\chi$. Since this is what we had assumed in the main text, we conclude that the expression for charges in Einstein gravity given in \eqref{kgrav} holds even if $\chi$ is field dependent, \ie $\chi=\chi(\Phi)$.

\section{On the generator of Killing horizon in NHEG}\label{sec vec na}
the Killing vector $\zeta_{H}$ that generates a bifurcate Killing horizon at $\cH$ is defined as
\begin{align}
\zeta_{ H}&
=n_{H}^a\xi_a - \vec{k}\cdot\vec{\mathrm{m}},
\end{align}
where $n_{H}^a$ are given by the following functions computed at the constant values $t=t_{H},\,r=r_{H}$ 
\begin{align}\label{n a}
n^-=-\frac{t^2r^2-1}{2r}\,, \qquad n^0=t\,r\,,\qquad n^+=-r.
\end{align}
One can obtain this vector by starting from the Killing vector $\xi_+$
\begin{align}\label{xi+}
\xi_+ =\dfrac{1}{2}(t^2+\frac{1}{r^2})\partial_t-tr\partial_r-\frac{1}{r}{k}^i{\p}_{\varphi^i},
\end{align}
One can invert $\xi_-=\p_t$ and $\xi_0=t\p_t-r\p_r$ to find 
\begin{align}
\p_t&=\xi_-,\qquad r\p_r=t\xi_-\xi_2
\end{align}
Replacing these in \eqref{xi+}, and rearranging, we find 
\begin{align}
n^a\xi_a-k^i\mathrm{m}_i&=0
\end{align}
where $n^a$ is given by \eqref{n a}. Note that $n^a$ are spacetime varying functions. Now if we define the constants $n_{H}^a$ as the value of $n^a$ at $t=t_H,r=r_H$, we find the vector $\zeta_H$ that is a combination of Killing vectors that vanishes at the surface $\cH$. Also interestingly it can be checked that $\cd _{(\mu}{\zeta_H}_{\nu)}=\eps_\mn$ where $\eps_\mn$ is the binormal tensor of $\cH$.
 
It can be shown that the functions $n^a$ form the coadjoint representation of $SL(2,\mathbb{R})$ as follows. The space of functions of $t,r$ forms a vector space in $\mathbb R$. The $SL(2,\mathbb{R})$ action is defined by  $\xi_a f(t,r)=\xi_a^\mu \p_\mu f(t,r)$. Now consider the subspace spanned by the three functions $n_a$ (with lower indices) defined as
\begin{align}
n_+= \frac{t^2r^2-1}{2r}, \quad n_0= t\,r ,\quad n_-= r.
\end{align}
One can check that the action of $SL(2,\mathbb R)$ vectors $\xi_a$ on the functions $n_b$ is given by a matrix whose components are the $SL(2,\mathbb R)$ structure constants, 
\begin{align}
\xi_a n_b&= f_{ab}^{\;\;\; \;c} \;n_c.
\end{align}
Therefore, the subspace spanned by $\{n_+,n_0,n_-\}$ forms the adjoint representation space of the $SL(2,\mathbb R)$ algebra. The functions $n^a$ are then defined as $n^a = K^{ab}n_b$, using the Killing form of $SL(2,\mathbb{R})$ in $(-,0,+)$ basis
\begin{align}
K_{ab}=K^{ab}=\begin{pmatrix}
0&0&-1\\
0&1&0\\
-1&0&0
\end{pmatrix}.
\end{align}
Accordingly the functions $n^a$ form the coadjoint representation. Since the Killing vectors $\xi_a$ \eqref{xi1-xi2} also form an adjoint representation of $SL(2,\mathbb{R})$, one can consider the direct product $n_a \otimes \xi_b$ which can be decomposed into $\textbf{3}\otimes \textbf{3}= \textbf{5} \oplus \textbf{3}\oplus \textbf{1}$.  The singlet $\textbf{1}$  is given by the vector $n^a\, \xi_a^{\;\mu} = K^{ab}n_b \xi_a^{\;\mu}$. This is obviously a singlet representation, since it is constructed by the contraction of the Killing form with two vectors. Indeed it can be shown that $n^a\, \xi_a=\vec{k}\cdot\vec{\mathrm{m}}$ and therefore the Killing vector $\zeta_{\cal H}$ vanishes on the surface $\cal H$.

The three vector $n^a$ can also be interpreted as the position vector of an $AdS_2$ surface
embedded in a three dimensional flat space $\mathbb{R}^{2,1}$ with the metric given by $-K_{ab}$. Explicitly
\begin{align}
n^2\equiv -K_{ab}n^a n^b=2n^+n^--(n^0)^2=-1.
\end{align}
The vector $n^a_{\cal H}$ is a specific point on this surface, but any other point can be obtained by an $SL(2,\mathbb{R})$ group action on this vector.\\

\addcontentsline{toc}{chapter}{References}
\bibliographystyle{unsrt}
\bibliography{Biblio}{}

\end{document}